\documentclass[final,3p,times,twocolumn]{elsarticle}
\usepackage{latexsym}
\usepackage{amssymb}
\usepackage{amsthm}
\usepackage[cmex10]{amsmath}
\DeclareMathOperator\arctanh{arctanh}
\usepackage{color,soul}

\usepackage{graphicx}
\usepackage{float}
\usepackage{subfig}
\usepackage[]{placeins}
\usepackage[]{algorithm}
\usepackage{algpseudocode}
\usepackage{multirow}
\usepackage[subnum]{cases}
\usepackage{tabularx}
\usepackage{booktabs}

\usepackage[]{hyperref} 
\usepackage[]{amsrefs} 

\usepackage{todonotes}

\usepackage{mathtools}
\DeclarePairedDelimiter\ceil{\lceil}{\rceil}
\DeclarePairedDelimiter\floor{\lfloor}{\rfloor}
\DeclarePairedDelimiter\abs{\lvert}{\rvert}%

\newcommand{\an}[1]{\textsc{#1}}

\renewcommand{\algorithmicforall}{\textbf{for each}}

\makeatletter
\renewcommand\appendix{\par
  \setcounter{section}{0}%
  \setcounter{subsection}{0}%
  \setcounter{equation}{0}%
  \setcounter{table}{0}
  \setcounter{figure}{0}
  \gdef\theequation{\@Alph\c@section.\arabic{equation}}%
  \gdef\thefigure{\@Alph\c@section.\arabic{figure}}%
  \gdef\thetable{\@Alph\c@section.\arabic{table}}%
  \gdef\thesection{\Alph{section}}%
  \@addtoreset{equation}{section}%
  \@addtoreset{table}{section}
  \@addtoreset{figure}{section}
}
\makeatother

\newcommand{\ra}[1]{\renewcommand{\arraystretch}{#1}}
\begin{document}
\begin{frontmatter}
\title{UDDSketch: Accurate Tracking of Quantiles in Data Streams}
\author [unile] {Italo Epicoco}
\ead{italo.epicoco@unisalento.it}
\cortext[cor1]{Corresponding author}
\author [unile] {Catiuscia Melle}
\ead{catiuscia.melle@unisalento.it}
\author [unile] {Massimo~Cafaro\corref{cor1}}
\ead{massimo.cafaro@unisalento.it}
\author [unile] {Marco Pulimeno}
\ead{marco.pulimeno@unisalento.it}
\author [unile] {Giuseppe Morleo}
\ead{giuseppe.morleo@studenti.unisalento.it}
\address[unile]{University of Salento, Lecce, Italy}

\begin{abstract}
We present \an{UDDSketch} (Uniform DDSketch), a novel sketch for fast and accurate tracking of quantiles in data streams. This sketch is heavily inspired by the recently introduced \an{DDSketch}, and is based on a novel bucket collapsing procedure that allows overcoming the intrinsic limits of the corresponding \an{DDSketch} procedures. Indeed, the \an{DDSketch} bucket collapsing procedure does not allow the derivation of formal guarantees on the accuracy of quantile estimation for data which does not follow a sub-exponential distribution. On the contrary, \an{UDDSketch} is designed so that accuracy guarantees can be given over the full range of quantiles and for arbitrary distribution in input. Moreover, our algorithm fully exploits the budgeted memory adaptively in order to guarantee the best possible accuracy over the full range of quantiles. Extensive experimental results on synthetic datasets confirm the validity of our approach.
\end{abstract}


\newtheorem{thm}{Theorem}
\newtheorem{lem}[thm]{Lemma}
\newdefinition{rmk}{Remark}
\newproof{pf}{Proof}
\newtheorem{prop}[thm]{Proposition}
\newtheorem*{cor}{Corollary}
\newdefinition{defn}{Definition}
\newtheorem{conj}{Conjecture}
\newtheorem{exmp}{Example}
\newtheorem{case}{Case}

\end{frontmatter}


\section*{Declaration of interest}

Declarations of interest: none.

\section{Introduction}
\label{intro}

A data stream $\sigma$ can be thought as a sequence of $n$ items drawn from a universe $\mathcal{U}$. In particular, the items need not be distinct, so that an item may appear multiple times in the stream. Data streams are ubiquitous, and, depending on the specific context, items may be IP addresses, graph edges, points, geographical coordinates, numbers etc. 

Since the items in the input data stream come at a very high rate, and the stream may be of potentially infinite length (in which case $n$ refers to the number of items seen so far), it is hard for an algorithm in charge of processing its items to compute an expensive function of a large piece of the input. Moreover, the algorithm is not allowed the luxury of more than one pass over the data. Finally, long term archival of the stream is usually unfeasible. A detailed presentation of data streams and streaming algorithms, discussing the underlying reasons motivating the research in this area is available to the interested reader in \cite{TCS-002}.

In this paper we are concerned with the problem of accurately tracking quantiles in data streams. The difficulty is strictly related to the underlying nature of the input data stream, since it is a well-known fact that computing exact quantiles is impossible without storing all of the data~\cite{MunroPaterson}. Therefore, approximate solutions such as those provided by sketches are the only viable possibility.

Formally, given a multi-set $S$ of size $n$ over $\mathbb{R}$, let $R(x)$ be the rank of the element $x$, i.e., the number of elements in $S$ smaller than or equal to $x$. Then, the lower (respectively upper) $q$-quantile  item $x_q \in S$ is the item $x$ whose rank $R(x)$ in the sorted multi-set $S$ is $\floor{1+q(n-1)}$ (respectively $\ceil{1+q(n-1)}$) for  $0 \leq q \leq 1$. By definition, $x_0$ and $x_1$ are respectively the minimum and maximum element of $S$, and $x_{0.5}$ is the median.

Regarding tracking accuracy, it can be defined in two different ways, as follows.

\begin{defn} Rank accuracy. $\forall$ item $v$ and $\epsilon$, return an estimated rank $\tilde{R}$ such that $\abs{\tilde{R}(v) - R(v)} \leq \epsilon n$.
\end{defn}

\begin{defn} Relative accuracy. $\tilde{x}_q$ is an $\alpha$-accurate $q$-quantile if
  $\abs{\tilde{x}_q - x_q} \leq \alpha x_q$ for a given $q$-quantile item
  $x_q\in S$. A sketch data structure is an $\alpha$-accurate $(q_0,q_1)$-sketch if it can output $\alpha$-accurate $q$-quantiles for ${q_0 \leq q \leq q_1}$.
 \end{defn}

Even though for long time research efforts have been focused on data structures providing rank accuracy, data sets with heavy tails are such that rank-error guarantees can return values with large relative errors. In particular, rank accuracy is not viable for tracking higher order quantiles of heavy-tailed distributions.
 
\an{DDSketch} (Distributed Distribution Sketch) \cite{Masson} is a recent sketch data structure providing relative accuracy for tracking quantiles in data streams whose underlying distribution is heavy-tailed. This sketch is conceptually very simple and can be implemented either using an unlimited number of buckets or fixing a desired maximum number of buckets to be used. 
In the former case, the space used may grow unbounded, whilst in the latter case when the current number of buckets in the sketch exceeds the predefined maximum a bucket collapsing procedure must be executed in order to guarantee that the number of buckets is always bounded from above.
 
Unfortunately, the authors of \an{DDSketch} do not provide formal guarantees on the estimation's accuracy for a collapsed sketch when the input data is not drawn from sub-exponential distributions.

 In this paper we introduce and discuss a novel collapsing strategy for DDSketch. The main contributions of this paper are the following ones: (i) we formally model the relationship between accuracy and space occupied by the sketch for arbitrary input distributions; (ii) our algorithm fully exploits the budgeted memory adaptively in order to guarantee the best possible accuracy over the full range of  quantiles.

\section{Related Work}
\label{related}



The problem of quantile computation has been extensively studied in the scientific literature, there are indeed several  publications about it, with algorithms characterized by very different approaches. The common goal is to provide the most accurate result possible with the minimum use of resources. 

The first works for the determination of a quantile sketch date back to the 80's when Munro and Paterson \cite{MUNRO1980315} demonstrated the first quantile sketching algorithm with formal guarantees. They proved the relationship between the amount of space needed related to the number of steps required to select the highest order statistical $k$-th on a dataset of $N$ elements. 

Munro and Paterson designed a probabilistic algorithm to estimate the median by keeping $s$ samples out of the $N$. If the data are presented in random order and $ s = \Theta(N^{\frac{1}{2}})$, then the algorithm has a high probability of storing samples containing the median. This algorithm can be adapted to find a specific quantile. The main result obtained is the proof that the amount of memory required by a deterministic $p$-pass selection algorithm is $\Omega(N^{\frac{1}{p}})$. For a data stream, where only one-pass is allowed, i.e. $p = 1$, the computation of the exact value of any quantile requires $\Omega(N)$ memory space. This result led subsequent work to focus on algorithms providing approximate quantile values.

A common technique used in practice for the selectivity estimation problem, is to maintain histograms of frequency, that is buckets containing groups of values that approximate the true value and its frequency according to the statistics maintained by each bucket. Gibbons et. al \cite{Gibbons:2002:FIM:581751.581753} presented two fast and efficient procedures for maintaining two classes of histogram: equi-depth histograms and compressed histograms.
In the equi-depth histogram, the elements are grouped into buckets so as to ensure the same number of elements for each of them (same height). In the compressed histogram, the $n$ highest frequencies are stored into $n$ separate buckets, the rest of the elements are partitioned according to the equi-depth histogram.
An equi-depth histogram approximates the exact histogram by relaxing the requirements on the number of elements in the bucket and counting accuracy. Its distance from the real histogram can be measured by the following error metric. Consider an approximate equi-depth histogram with $\beta$ buckets for $N$ elements, the error metric $\mu_{ed}$ is the standard deviation of the buckets sizes from their average, normalized with respect to the average of the buckets sizes. The variant with compressed histograms is also treated in a similar way, obviously with the necessary modifications to adapt the algorithm to that class of histograms.

These algorithms provide a summary of the data using histograms and can be used to estimate quantiles according to a different error metric, however, they need to perform multiple passes on the whole input dataset.                 

Manku et al. \cite{Manku:1998:AMO:276305.276342} designed an algorithm whose accuracy bound is independent from the input distribution and the approximation error is uniformly distributed over all quantiles. The algorithm uses $b$ buffers which store $k$ elements each. Each buffer $B$ is associated with a weight $w_B$ which represents the occurrences of the input items fallen in the buffer. When the algorithm starts, all buffers are empty and they are populated with the elements from the input dataset; when all of the buffers are full, the collapsing procedure applies, modifying the weight of the collapsed buffer accordingly.
The authors proved that the error $\epsilon$ committed on the estimation of the $q$-quantile is bounded and the space required to guarantee the error bound is $O(\frac{1}{\epsilon} log^2 \epsilon N)$. The algorithm is efficient and offers opportunities for parallelism, however, it requires to know in advance the size $N$ of the dataset which makes the algorithm not suitable for data streams processing.

Summarizing large datasets is important because of limited memory resources. The GK sketch algorithm by M. Greenwald and S. Khanna \cite{Greenwald01space-efficientonline} addressed the problem of designing a space-efficient algorithm based on quantile summaries. A summary consists of a small number of items sampled from the input sequence. These items are then used to respond to any quantile request.
The algorithm provides an $\epsilon$-approximate estimate $r'_q$ of the $q$-quantile. A summary is $\epsilon$-approximate if the estimate of a $q$-quantile differs from the exact value $r_q$ by  $|r'_q -r_q| \leq \epsilon N$. The algorithm requires memory space $O(\frac{1}{\epsilon}log(\epsilon N))$ and it is independent by the input distribution. GK Sketch offers excellent results in terms of approximation and space used, however it is not fully mergeable, which makes it impossible to use it in a distributed setting, moreover the memory required depends on the size of the input dataset which is not known when processing data streams.

Summarizing distributions which have high skew using uniform quantiles is not informative because having a uniformly spread-out summary of a stretched distribution does not describe the interesting tail region adequately. Motivated by this, Cormode et al. \cite{Cormode:2005:ECB:1053724.1054027} designed an algorithm to efficiently estimate the high-biased quantiles. The high-biased quantiles are defined as: $1-\phi, 1-\phi^2, \cdots, 1-\phi^k$ with $0<\phi<1$.
The algorithm keeps information about particular items from the input, and also stores some additional tracking information. The intuition for this algorithm is as follows: suppose we have kept enough information so that the median of a dataset with $N$ elements can be estimated with an absolute error of $\epsilon N$ in rank. Now suppose that there are $N$ more insertions of items above the median, so that this item is pushed up to being the first quartile. If the same absolute uncertainty of $\epsilon N$ is maintained, then this corresponds to a relative error of size $\epsilon/2$, considering that the number of items is doubled.
Inspired by the GK algorithm, Cormode et al. provided an algorithm which is able to support greater accuracy for the high-biased quantiles.

The Moment Sketch algorithm of E. Gan et al. \cite{Gan:2018:MQS:3236187.3269475} is based on a data structure defined as moment sketch. The sketch requires a minimal amount of space and it is mergeable and computationally efficient. The authors use the moments methods to build the $f(x)$ distribution function which can be used to describe the input dataset. Letting $k$ be the highest power used  for the moments, the moment sketch of a dataset $D$ includes: the minimum value $x_{min}$; the maximum value $x_{max}$; the number of items $n$; the sample of moments $\mu_i = \frac{1}{n}\sum_{x\in D} x^i$ for $i \in \lbrace 1, \cdots ,k \rbrace$; the logarithmic moments $\nu_i = \frac{1}{n} \sum_{x\in D} \lg^i(x)$ for $i \in \lbrace 1, \cdots  ,k \rbrace$.
To estimate a quantile from a moments sketch, the moments method is applied to build the PDF $f(x)$ whose moments match those stored in the sketch and that maximizes the entropy defined as $H \left[ f \right] = -\int_\chi f(x) \lg f(x) \ dx$. $f(x)$ is then used to estimate the quantiles of the dataset. 
The moments sketch proves to be very fast, with an average error of less than 0.01 using about 200 bytes of space. However, there could be pathological situations with certain distributions for which it is not possible to compute finite moments. Moreover, the error is guaranteed in the average case, but not in the worst case, and errors caused by floating point multiplications can occur.

The work done by T. Dunning and O. Ertl \cite{dunning2019computing} has introduced a new data structure known as $t$-digest, formed by clustering real value samples. This structure differs from the previous ones in several ways: the data are grouped and summarized in the $t$-digest structure, however the range of data included in different clusters may overlap; the buckets are represented by a centroid value and a weight value that represents the number of samples contributing to the bucket, instead of the classic lower and upper limits; the samples are accumulated in such a way that only a few of them help determining extreme quantiles, so that relative error is bounded instead of maintaining constant the absolute error.
The accuracy of the $q$-quantile estimate is near to $q(1-q)$. In this algorithm the accuracy depends on the quantile and is more accurate for computing quantiles close to 0 and 1.

With this work the authors  provide a solution to the problem of quantile computation on data streams.

Given a set of elements $x_1, \cdots ,x_n$, the quantile $x$ is the fraction of elements in the stream such that $x_i \le x$, i.e., the rank of $x$ denoted by $R(x)$. 
A data structure is accurate for all quantiles if for each $x$, letting $\tilde{R}(x)$ be the estimated rank of $x$, with probability $1-\delta$ is holds that $\vert \tilde{R}(x) - R(x) \vert \epsilon n$.

Z. Karnin, K. Lang and E. Liberty \cite{karnin2016optimal} designed their algorithm as a reinterpretation of the work of \cite{Agarwal} and \cite{Manku:1999} from a different point of view. The algorithm is based on the concept of a \textit{compactor}, a data structure that can store $k$ elements all with the same weight $w$, and if necessary can compact its $k$ elements into $k/2$ elements of weight $2w$ in the following way: items are sorted, then odd (respectively even) items are selected and the non-selected even (respectively odd) items are discarded, and the weight $w$ of each selected item is doubled. Each compactor will eliminate odd or even items with equal probability. The rank estimation after this process depends at most on $w$. 
The output elements of a compactor are put into another one and so on, and since each compactor has half of the elements in the sequence there will be at most $H \le \lceil \lg (n/k) \rceil$ compactors chained together creating a hierarchy with variable capacity.
Considering an algorithm run ending with $H$ different compactors the theorem proved by the authors states that there is a streaming algorithm that calculates a $\epsilon$ approximation for the rank of a single item with probability $1-\delta$ whose space complexity is $O((1/\epsilon) \sqrt{\lg(1/\delta)})$.
Moreover, there is another streaming algorithm that produces mergeable summaries and computes an $\epsilon$ approximation for the rank of a single item with probability $1-\delta$ whose space complexity is $O((1/\epsilon) \sqrt{\lg(1/\delta)} + \lg(\epsilon n))$.
An additional optimization guarantees the rank computation of a single element with $1-\delta$ probability and with a space complexity of only $O((1/\epsilon) \lg \lg (1/ \delta))$ for the non-mergeable version and $O((1/\epsilon) \lg^2 \lg (1/ \delta))$ for the mergeable version.
The algorithm provides a randomized solution to the problem of computing quantiles on data streams with a probability of error of $1 - \delta$ and a minimum amount of space used, ensuring the property of full mergeability. However, the algorithm provides estimates with a greater relative error for the high quantiles on heavy-tailed data.

\section{DDSketch}
\label{dds}

A basic version of \an{DDSketch}, described in \cite{Masson}, can provide $\alpha$-accurate $q$-quantiles for any $0 \leq q \leq 1$. This version of the algorithm is both simple to understand and implement, and provides support for item insertion/deletion and merging of two compatible sketches (i.e., sketches characterized by the same $\alpha$ value). The main drawback of this algorithm is that the accuracy is obtained by trading off the space required: the number of buckets in a sketch can grow without bound. Owing to this limitation, the authors of \an{DDSketch} introduced in \cite{Masson} an advanced version of \an{DDSketch} that can deliver $\alpha$-accurate $q$-quantiles for $q_0 \leq q \leq 1$ with a bounded number of buckets. In this manuscript we will only deal with the second improved version of \an{DDSketch}.
 
\an{DDSketch} works by dividing $\mathbb R_{> 0}$ into indexed buckets. Let $B_i$ be the bucket with index $i$ and $m$ the maximum number of buckets. The algorithm works reactively, by invoking a collapsing procedure if inserting a value causes the number of buckets to grow beyond $m$.

Denoting by $\gamma$ the quantity $\frac{1+\alpha}{1-\alpha}$ where $\alpha$ represents the user's defined accuracy, the bucket $B_i$ is a counter holding the occurrences of values $x$ falling between the interval given by $\gamma^{i-1} < x \leq \gamma^i$. Algorithm \ref{dds} shows the pseudo-code related to the insertion procedure of an item $x$. We assume that the number $b$ of buckets stored in the sketch at any time is $0 \leq b \leq m$, i.e., the number of buckets maintained is dynamic, depending on the sequence of insertions and deletions operations. Of course, the bucket indexes are dynamic as well. A bucket always holds a positive count. This is certainly true for insertion-only streams. However, \an{DDSketch} also allows deletions, in which case a bucket count may be zero. When this happens, a bucket is discarded and thrown away. 

To insert a value $x$, the index $i$ of the bucket in which $x$ falls is computed as $i = \lceil\log_\gamma{x}\rceil$. If the bucket $B_i$ has been already inserted into the sketch, then the bucket's counter is incremented by one. Otherwise, $B_i$ is added to the sketch with a count initialized to one. Then, if the number of buckets exceeds $m$ after inserting $x$, a bucket collapsing procedure is executed, by collapsing the initial two buckets. Note that, in general, the first two buckets are not $B_1$ and $B_2$, since the indexes of these buckets depend on the actual insertions done. Therefore, we denote in the pseudo-code these buckets as $B_y$ and $B_z$. In particular, it holds that $y < z$ but it is not necessarily true that $z = y + 1$, i.e. the indexes need not be consecutive. The buckets $B_y$ and $B_z$ are updated so that the count stored by $B_y$ is added to $B_z$, and $B_y$ is removed from the sketch. Alternatively, the collapsing procedure can be applied to the last two buckets.

\begin{algorithm}
  \caption{DDSketch-Insert($x, \mathcal{S}$)}
  \label{dds}
  \begin{algorithmic}
  \Require {$x \in \mathbb R_{> 0}$}
  \State $i \leftarrow \ceil{\log_\gamma{x}}$
  \If{$B_i \in \mathcal{S}$}
  	\State $B_i \leftarrow B_i + 1$
  \Else
  	\State $B_i \leftarrow 1$
  	\State $\mathcal{S} \leftarrow \mathcal{S} \cup B_i$
  \EndIf
  \If{$\abs{\mathcal{S}} > m$}
  	\State let $B_y$ and $B_z$ be the first two buckets
    \State $B_z \leftarrow B_y + B_z$
    \State $\mathcal{S} \leftarrow \mathcal{S} \smallsetminus B_y$
  \EndIf
    \end{algorithmic}
\end{algorithm}

The authors of \an{DDSketch} show that $m$ buckets suffice to $\alpha$-accurately answer a given $q$-quantile query if: $x_1 \leq x_q \gamma^{m-1}$, or,  equivalently: 
\begin{equation}
	\label{bound-condition}
	\frac{\log(x_1) - \log(x_q)}{\log(\gamma)} + 1 \leq m.
\end{equation}

Then, they prove the following theorem, which sets a bound to Eq. \ref{bound-condition} for datasets drawn from sub-exponential distributions.

\begin{thm}
	\label{ddsketch-bound}
	Let $X_{(1)} \leq X_{(2)} \leq \cdots \leq X_{(n)}$ be the order statistics for i.i.d. random variables $X_i$ distributed according to a sub-exponential distribution $F$ with parameters $(\sigma, b)$. Then with probability at least $1 - \delta_1 - \delta_2$, DDSketch is an $\alpha$-accurate $(q, 1)$-sketch with size at most $(\log X_{(n)} - \log X_{(q n)})/ log(\gamma) + 1$, which is bounded from above by: 
	\begin{equation}
		\frac{\log(2b \log(n/\delta_2) + \mathbb{E}X) - \log(F^{-1}(q - t))}{\log(\gamma)} + 1
	\end{equation}
	for: \\ $\gamma = (1 + \alpha)/(1 - \alpha)$, $t = \sqrt{\log(1/\delta_1)/2n}$, and $t < q < 1/2$.
\end{thm}

\an{DDSketch}, as described by its authors, only deals with $\mathbb R_{> 0}$. Therefore, in order to deal with $\mathbb{R}$, one must use two sketches, one of which devoted to negative values. 

 \section{UDDSketch}
\label{idds}

Theorem \ref{ddsketch-bound} holds for input data following a sub-exponential distribution and requires that the distribution parameters $\sigma \text{ and } b$ are known. However, for arbitrary and/or unknown input distributions, it is not possible to give formal guarantees on the accuracy of \an{DDSketch} when the number of buckets at disposal is limited. In such a case an error beyond the desired level may affect also the range of quantiles of interest.

We devised a different collapsing strategy for \an{DDSketch} that overcomes the problem discussed above and allows giving guarantees on the accuracy of the sketch for all of the quantiles. As expected, providing the user with an approximated result, there is a trade-off involved between the $\alpha$ accuracy that can be achieved and the amount of space at disposal. However, we can prove that, if the maximum and minimum of the values which can appear in input are known or can be estimated with a low probability of failure, then a strict relation exists between a desired level $\alpha$ of accuracy on a generic quantile query and the number of buckets needed to guarantee that accuracy.

The new collapsing strategy is named \textit{uniform collapse} and, differently from the \an{DDSketch} collapsing, does not involve only two buckets, but all of the buckets, which  are collapsed two by two. More precisely, for each pair of indices $(i,  i+1)$, where  $i$ is odd and $B_i \neq 0$ or $B_{i+1} \neq 0$, a new bucket with index $j = \lceil \frac{i}{2} \rceil$ is created, whose count is the sum of the counts of $B_i$ and $B_{i+1}$ and which replaces the collapsed buckets. Algorithm \ref{unif-collapse} reports the pseudocode of the uniform collapse procedure.

\begin{algorithm}
\caption{UniformCollapse($\mathcal{S}$)}
	\label{unif-collapse}
 \begin{algorithmic}
	\Require {sketch $\mathcal{S} = \lbrace B_i \rbrace_i$}
	\ForAll{ $\lbrace i: B_i > 0 \rbrace$ }
		 \State $j \leftarrow  \lceil \frac{i}{2} \rceil$
		 \State $B'_{j} \leftarrow B'_{j} + B_{i}$
	\EndFor
	\State \Return $\mathcal{S} \leftarrow \lbrace B'_i \rbrace_i$
\end{algorithmic}
\end{algorithm}

The following lemma formally shows and justifies how uniform collapse modifies the sketch and its accuracy.

\begin{lem}
	The collapsing procedure applied to an $\alpha$-accurate $(0,1)$-quantile sketch produces an $\alpha^\prime$-accurate $(0,1)$-quantile sketch on the same input data with ${\alpha^\prime = \frac{2\alpha}{1+\alpha^2}}$. Moreover, an item $x$ falling in bucket with index $i$ of a collapsing sketch, will fall in bucket with index $\lceil i/2\rceil$ of the collapsed sketch.
	\label{lemma:3.2}  
\end{lem}
\begin{proof}
	Let $B_i$ and $B_{i+1}$ be two adjacent buckets of the sketch to be collapsed. The collapsing procedure sums them up and replaces them with a new bucket, which we denote by $B'_{j}$. Let $U_i$, $U_{i+1}$ and $U'_j = U_{i} \cup U_{i+1}$ denote the intervals of values which refer respectively to buckets $B_i$, $B_{i+1}$ and $B'_j$. Let $\gamma = \frac{1+\alpha}{1-\alpha}$ and $\gamma^\prime = \frac{1+\alpha^\prime}{1-\alpha^\prime}$.
	
	We have that:
	\begin{align*}
		U_i = (\gamma^{i-1}, \gamma^i], U_{i+1} = (\gamma^i, \gamma^{i+1}],\\
		U_{i} \cup U_{i+1} = (\gamma^{i-1}, \gamma^{i+1}], \\ 
		U'_j = U_{i} \cup U_{i+1} = ({\gamma^\prime}^{j-1}, {\gamma^\prime}^{j}],
	\end{align*}

	\noindent from which we derive that:
	\begin{equation}
		\gamma^\prime = \frac{{\gamma^\prime}^j}{{\gamma^\prime}^{j-1}} = \frac{\gamma^{i+1}}{\gamma^{i-1}} = \gamma^2,
	\end{equation}

	\noindent and, as a consequence of the relation between $\alpha^\prime$ and $\gamma^\prime$, and $\alpha$ and $\gamma$:
	
	\begin{equation}
		\alpha^\prime = \frac{\gamma^\prime - 1}{\gamma^\prime + 1} = \frac{\gamma^2 - 1}{\gamma^2 + 1} = \frac{2\alpha}{1 + \alpha^2}.
	\end{equation}

\noindent Furthermore, we have that, if $B_i$ and $B'_j$ are the buckets in which a value $x$ falls, respectively, before and after the collapse, then it holds that:

\begin{equation}
	j = \lceil \log_{\gamma^\prime}{x} \rceil = \lceil \log_{\gamma^2}{x} \rceil = \left\lceil \frac{\log_\gamma{x}}{2} \right\rceil = \left\lceil \frac{i}{2} \right\rceil
\end{equation}

\noindent that proves the relation between the bucket keys of the collapsing sketch and those ones of the collapsed sketch.
\end{proof}

After collapsing the buckets, $\alpha^\prime$ represents the new theoretical error bound for the sketch. Each time we perform a collapse, $\alpha$ increases, i.e. we lose accuracy. However, we do not expect executing the collapsing procedure repeatedly up to the point where the loss in accuracy adversely impacts on the data structure precluding its use. The reason is that each time a collapsing is done, the input interval covered by the $m$ available buckets increases as well, so that a few collapsing are enough to process input data streams with very large range of values.


 According to the collapsing algorithm, we can formulate the following theorem which provides an upper bound on the accuracy of the results, i.e., on the error committed approximating the quantile computations when a limited number of buckets are at disposal.
 
 \begin{thm}
 \label{errorbound}
Given an input whose data domain is an interval $[x_{min}, x_{max}] \in \mathbb R_{>0}$ and an \an{UDDSketch} data structure using at most $m$ buckets to process the input, the approximation error committed by \an{UDDSketch} using the uniform collapse procedure is bounded by $\hat\alpha = \frac{\tilde\gamma^2-1}{\tilde\gamma^2+1}$, with $\tilde\gamma = \sqrt[m]{\frac{x_{max}}{x_{min}}}$. 
\end{thm}

\begin{proof}
In order to provide un upper bound on the accuracy achieved by the U\an{DDSketch} data structure, we analyze the worst case, i.e., the situation in which the $m$ buckets must uniformly cover the interval $[x_{min}, x_{max}]$. In such a case, the corresponding indexes are consecutive numbers denoted by $i_1, i_2, \cdots, i_m$. Let the covered interval be $(\tilde\gamma^{i_1 - 1}, \tilde\gamma^{i_m}]$. Choosing ${i_1 = \ceil{\lg_{\tilde\gamma} x_{min}}}$, it holds that $x_{min}$ falls into the first bucket $B_{i_1}$. Therefore, it holds that 

\begin{equation}
\label{firstbucket}
	\tilde\gamma^{i_1 - 1} < x_{min} \leq \tilde\gamma^{i_1}.
\end{equation}

We now show that $x_{max}$ falls into the last bucket $B_{i_m}$, i.e.,

\begin{equation}
\label{lastbucket}
	\tilde\gamma^{i_m} \geq x_{max}.
\end{equation}

 It holds that $\tilde\gamma^{i_m} = \tilde\gamma^{i_1} \tilde\gamma^{m}$ since the buckets indexes are consecutive and the buckets uniformly cover the whole interval. As a consequence, taking into account the definition of $\tilde\gamma$, equation \eqref{lastbucket} is equivalent to
 
 \begin{equation}
\label{equiv}
	\tilde\gamma^{i_1} \frac{x_{max}}{x_{min}} \geq x_{max}.
\end{equation}

Equation \eqref{equiv} holds, since $\frac{\tilde\gamma^{i_1}}{x_{min}} \geq 1$, owing to equation \eqref{firstbucket}.

Now consider an initial $\alpha$ value and a corresponding initial $\gamma$ such that an integer number of collapses which brings $\gamma$ to $\tilde{\gamma}$ does not exist, but it holds that $\gamma^{2^k} < \tilde{\gamma} < \gamma^{2^{k+1}}$, for a $k \in \mathbb{N}$. In this case, we may need a number $k+1$ of collapses to accommodate all of the input values and end up with a final value of $\gamma > \tilde{\gamma}$. However, even in this eventuality, the value of $\gamma$ can not grow beyond $\tilde{\gamma}^2$ and this is the reason why the upper bound of the accuracy is set to $\frac{\tilde{\gamma}^2-1}{\tilde{\gamma}^2+1}$.

\end{proof}

In Theorem \ref{errorbound}, we assume that the values $x_{min}$ and $x_{max}$ of the input data are known. This is not always true, but we can always estimate these values with a probability $\delta$ of failing our prediction. In that case, the bound showed by Theorem \ref{errorbound} holds with probability $1 - \delta$.

We now discuss how to, given a user desired level $\alpha$ of accuracy and a number of buckets sufficient to satisfy that accuracy based on Theorem \ref{errorbound}, choose the initial value of accuracy $\alpha_0 \leq \alpha$ to start our algorithm. By construction, the sequence of $\alpha_k$ values corresponding to the $\gamma_k$ values changed upon a collapsing procedure follows the recurrence equation:

\begin{equation}
\label{recurrence}
\alpha_k =  \begin{cases} \alpha_0 & \text{for $k = 0$} \\
 \frac{2 \alpha_{k-1}}{1+\alpha_{k-1}^2} & \text{for $k > 0$},
\end{cases}
\end{equation}
where $k$ denotes the number of collapses performed.

The solution to Eq. \ref{recurrence} is ${\alpha_k = \tanh{(2^{k-1} \arctanh{(\alpha_0)})}}$. Similarly, the equation
 
\begin{equation}
 	\label{recurrence-sol}
  	\alpha_0 = \tanh{\left(\frac{\arctanh{(\alpha_k)}}{2^{k-1}}\right)},
\end{equation}
 
allows to compute $\alpha_0$ given a final accuracy $\alpha_k$ corresponding to $k$ collapses.

We can use Eq. \ref{recurrence-sol} to compute the initial value of the accuracy parameter for our algorithm by setting $\alpha_k$ to the value of user desired accuracy and  $k$ to the number of collapses that we are willing to accept. There is a trade-off to take into consideration: if we go backward too far, i.e., we set $k$ too large, we could end up with too many collapses and a decrease of performance, but  with a favourable input distribution, we can obtain a better accuracy. On the contrary, if we compute the initial accuracy with a few collapses or no collapses at all ($\alpha_0$ and user $\alpha$ coincide), we improve the performance and we may require less space, but we can not do better in terms of accuracy than guaranteeing the desired alpha. We have seen that a good empirical value for $k$ is 10.
 
\section{Experimental Results}
\label{results}
In this Section, we present and discuss the experimental investigation carried out in order to compare \an{UDDSketch} against \an{DDSketch}. 

Both \an{UDDSketch} and \an{DDSketch} algorithms have been implemented in C and compiled using the GCC compiler v4.8.5 on linux CentOS 7 with optimization level O3. The tests have been executed on a workstation equipped with 64 GB of RAM and two 2.0 GHz exa-core Intel Xeon CPU E5-2620 with 15 MB of cache level 3. The source code is freely available for inspection and reproducibility of results\footnote{https://github.com/cafaro/UDDSketch}. 

\begin{table*}
	\caption{Synthetic datasets}
	\label{synthetic-data}
	\centering
	\ra{1.4}
	\begin{tabular}{@{}llll@{}}
		\toprule
		\textbf{Dataset} & \textbf{Min value} & \textbf{Max value} & \textbf{Distribution}\\
		\hline
		betaL & $3.04 \times 10^{-2}$ & $0.99$ & $\textit{Beta}(5, 1.5)$ \\
		betaR& $8.34 \times 10^{-7}$ & $0.97$ & $\textit{Beta}(1.5, 5)$ \\
		chisquare & $5.67\times 10^{-3}$ & $42.9$ & $\chi^2(5)$ \\
		exponential & $1.19\times10^{-7}$ & $34.9$ & $\textit{Exp}(0.5)$ \\
		extremevalue & $14.5$ & $54.2$ & $\textit{Extremevalue}(20, 2)$ \\
		gamma & $2.99 \times 10^{-3}$ & $81.8$ & $\Gamma(2, 4)$ \\
		gumbel & $31.0$ & $1.11\times10^{2}$ & $\textit{Gumbel}(100, 4)$ \\
		halfnormal & $4.01\times10^{-7}$ & $13.7$ & $\textit{Halfnormal}(0.5)$ \\
		inversegaussian & $0.17$ & $5.12\times10^{2}$ & $\textit{IG}(10, 5$) \\
		laplace & $26.4$ & $3.67\times10^{2}$ & $\textit{Laplace}(200, 10)$ \\
		logistic & $32.8$ & $3.68\times10^{2}$ & $\textit{Logistic}(200, 10)$ \\
		lognormal & $1.08\times10^{-3}$ & $7.91\times10^{3}$ & $\textit{Lognormal}(1, 1.5)$ \\
		normal & $39.7$ & $60.5$ & $\textit{N}(50, 2)$ \\
		pareto & $2.0$ & $8.64\times 10^{-12}$ & $\textit{Pareto}(2, 0.5)$ \\
		uniform & $2.18\times10^{-3}$ & $2.49\times10^4$ & $\textit{Unif}(0, 2.5\times10^4)$ \\
		\bottomrule
	\end{tabular}
\end{table*}

\begin{figure*}[h]
	\caption{Statistical distributions}
	\label{distributions}
	\centering
	\includegraphics[width=1.0\textwidth]{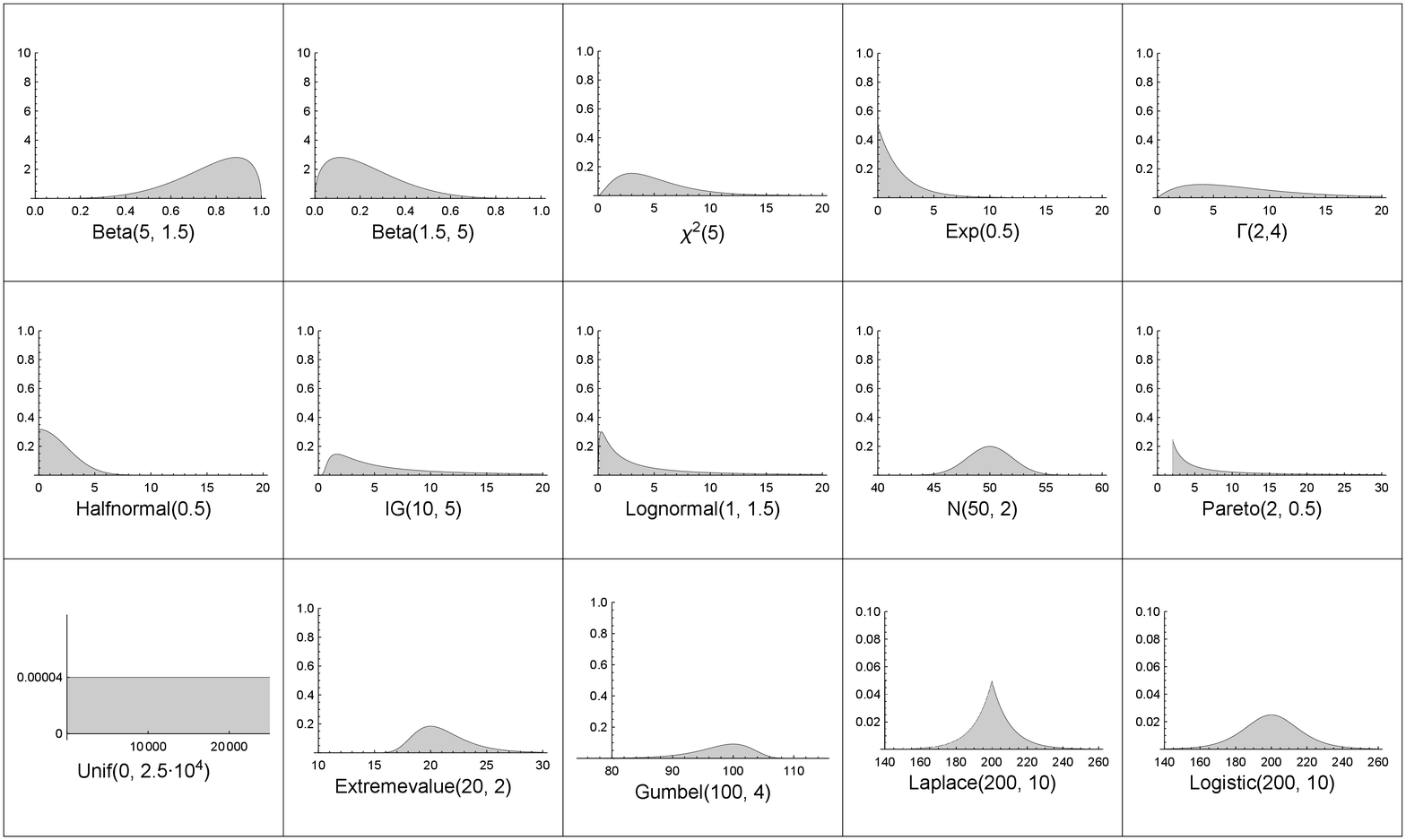}
\end{figure*}

The tests have been performed on $15$ various synthetic datasets, whose properties are summarized in Table \ref{synthetic-data}. Each dataset consists of 10 million real values. Figure \ref{distributions} shows the statistical distributions from which the datasets are drawn. 

\an{DDSketch} is executed using in each experiment all of the possible collapsing strategies: collapses of buckets with higher IDs (\an{DDSketch H}), collapses of buckets with lower IDs (\an{DDSketch L}) and a third variant (\an{DDSketch D}) where the available buckets have been equally partitioned between two sketches, one \an{DDSketch H} and one \an{DDSketch L}. In that case, each quantile is estimated through the most accurate sketch, i.e. the sketch whose estimation comes from a non collapsed bucket. If both answers come from the collapsed bucket, the estimation from the sketch with less overall collapses is chosen.

The three variants of \an{DDSketch} and \an{UDDSketch} have been executed on each dataset in Table \ref{synthetic-data} varying the value of $\alpha$ and the maximum number of buckets available to the algorithm. The sets of values used are shown in Table \ref{parameters}.

\begin{table}
	\caption{Synthetic datasets}
	\label{parameters}
	\centering
	\ra{1.2}
	\small
	\begin{tabular}{@{}ll@{}}
		\toprule
		\textbf{Parameter} & \textbf{Set of values}\\
		\hline
		\textbf{User} $\mathbf{alpha}$ & $\{0.00001, 0.0001, 0.001, 0.01, 0.1\}$ \\
		\textbf{Number of buckets} & $\{128, 256, 512, 1024, 2048\}$ \\
		\bottomrule
	\end{tabular}
\end{table}

In each test run, the performance, i.e., the number of values processed in a unit of time (updates per millisecond), and the accuracy, i.e., the relative errors committed on estimation of quantiles $q_0, q_{0.1}, q_{0.2}, \dots, q_{1}$, are measured for all of the collapsing strategies under investigation. 

Figures \ref{allqs-plots1} and \ref{allqs-plots2} show the estimation errors committed by \an{DDSketch L}, \an{DDSketch H} and \an{DDSketch D}, compared with \an{UDDSketch}. Figure \ref{allqs-plots1} refers to the \textit{betaL}, \textit{chisquare} and \textit{exponential}  datasets, whilst Figure \ref{allqs-plots2} is relative to the  \textit{normal}, \textit{pareto} and \textit{uniform} datasets. The number of buckets is set to $1024$ and $\alpha$ is set to $0.001$.
The results obtained when processing the other datasets in Table \ref{synthetic-data} are not reported here, for saving space since they exhibit similar behaviours. 

The plots show a major robustness of \an{UDDSketch} with reference to the distribution of the values in input. Even when the number of buckets granted to the algorithm is not enough to reach the desired $\alpha$ (dotted line), nonetheless \an{UDDSketch} guarantees an overall better accuracy, regardless of the input distribution. Even when we are only interested to specific quantiles, \an{DDSketch} does not always succeed in guaranteeing a bounded relative error, as \an{UDDSketch} does, independently of the chosen collapsing strategy. Furthermore, it is not possible for \an{DDSketch} to choose the best collapsing strategy a priori, when the input distribution is unknown. Particularly critical are the quantiles around the median, which rarely \an{DDSketch} can report with sufficient accuracy.
 
Figures \ref{boxplot-plots1} and \ref{boxplot-plots2} show how the median and interquartile range of the relative errors on quantiles change varying the number of buckets, when $\alpha$ is fixed to $0.001$. As in Figures \ref{allqs-plots1} and \ref{allqs-plots2}, the plots in each column refer to the same collapsing strategy, and the plots in each rows are relative to the same dataset. The datasets examined are the same as in the previous figures.

The observations made by inspecting Figures \ref{allqs-plots1} and \ref{allqs-plots2} are confirmed by Figures \ref{boxplot-plots1} and \ref{boxplot-plots2}. \an{UDDSketch} returns quantile estimations that are overall more accurate than \an{DDSketch}, also in terms of lower medians and shorter interquartile ranges of relative errors. Moreover, the experiments show that our solution puts to better use the number of buckets at disposal: in fact, \an{UDDSketch} keeps improving the estimate when the space granted grows, whilst \an{DDSketch} stops when the required $\alpha$ is reached and makes no use of the extra buckets at disposal.

At last, Figure \ref{update_ms-plots} shows the performance of the different \an{DDSketch} collapsing strategies compared with \an{UDDSketch}, when varying the number of buckets and with reference to the datasets \textit{betaL}, \textit{exponential}, \textit{uniform}. The other datasets lead to similar behaviours and are not reported.

The performance of the algorithms under test are in general comparable. \an{UDDSketch} is better when the number of buckets is low, \an{DDSketch L} and \an{DDSketch H} are more performant when the space grows, for they tends to not make use of the extra space. \an{DDSketch D} is always less performant due to the use of two sketches that must be updated at the same time.

\section{Conclusions}
\label{conclusions}

We have introduced \an{UDDSketch} (Uniform DDSketch), a novel sketch for fast and accurate tracking of quantiles in data streams. Our sketch was heavily inspired by the recently introduced \an{DDSketch}, and is based on a novel bucket collapsing procedure that allows overcoming the intrinsic limits of the corresponding \an{DDSketch} procedures. \an{UDDSketch} has been designed so that accuracy guarantees can be given over the full range of quantiles and for arbitrary distribution in input. Moreover, our algorithm fully exploits the budgeted memory adaptively in order to guarantee the best possible accuracy over the full range of quantiles. Extensive experimental results on synthetic datasets have confirmed the validity of our approach. 

\begin{figure*}[h]
	\centering
	\begin{tabular}{ccc}		
		\subfloat[]{
		\includegraphics[width=0.3\textwidth]{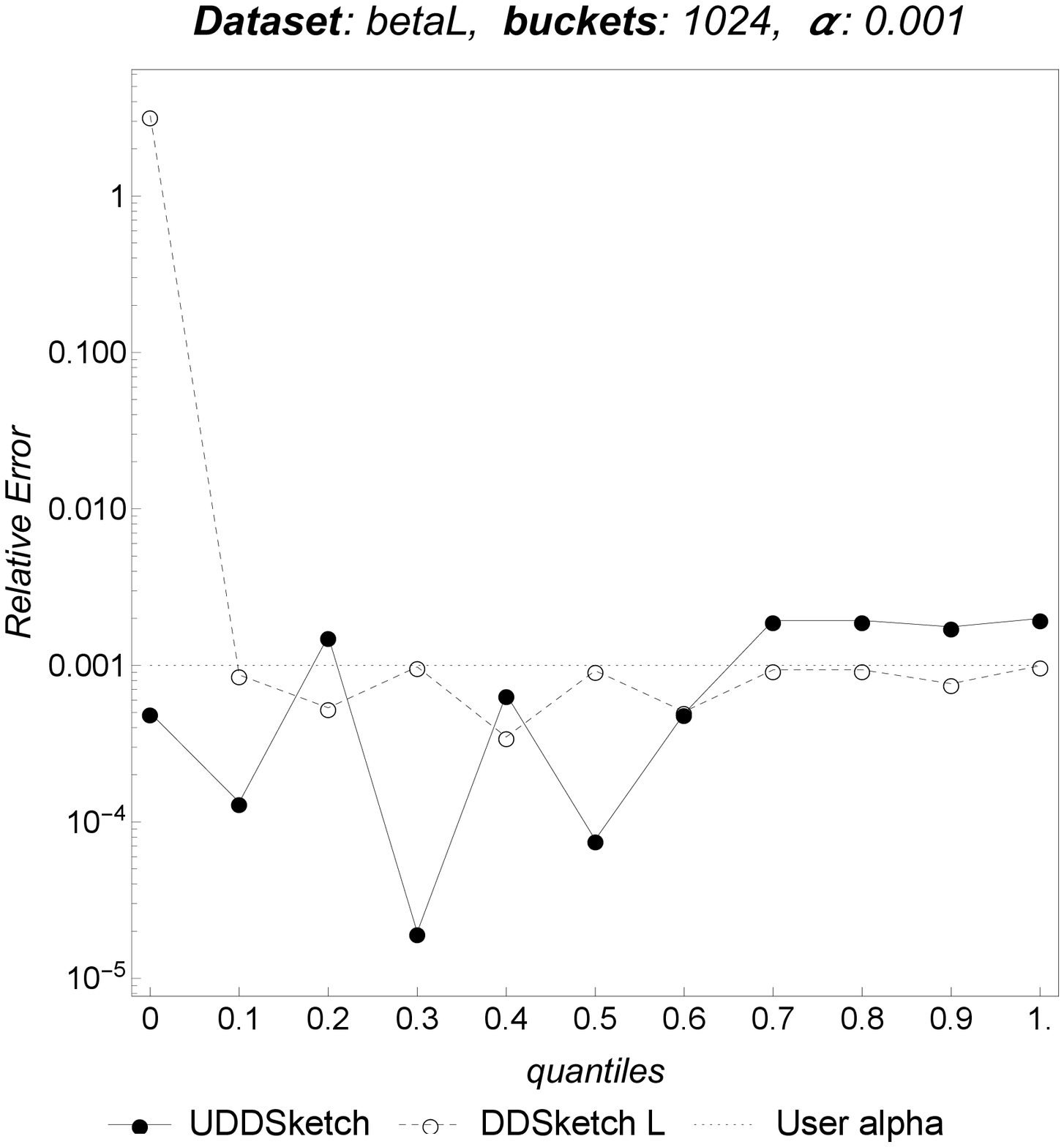}
		\label{betaL-allqs-ddsL}
		} &
		
		\subfloat[]{
			\includegraphics[width=0.3\textwidth]{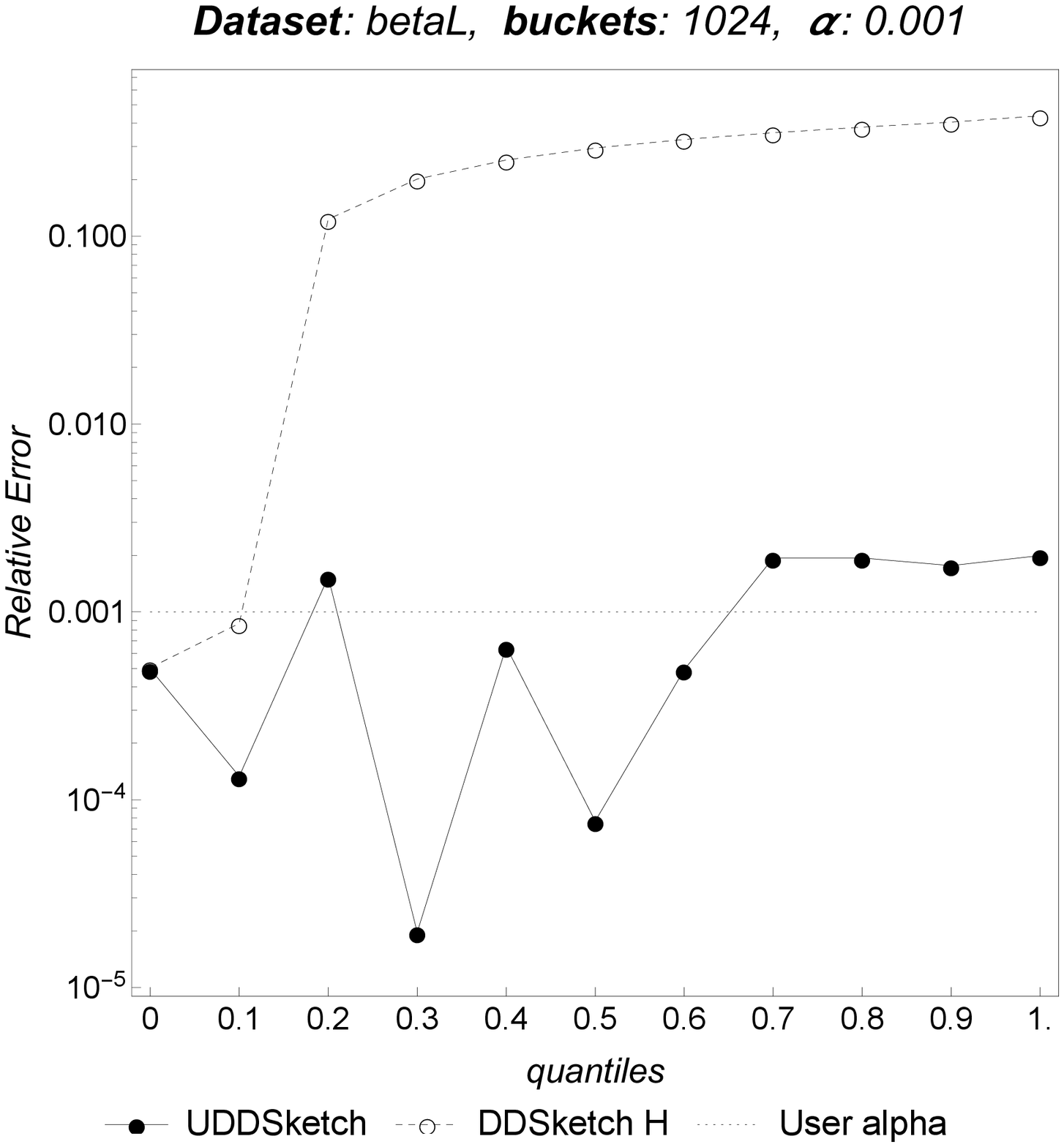}
			\label{betaL-allqs-ddsH}
		} &
		
		\subfloat[]{
			\includegraphics[width=0.3\textwidth]{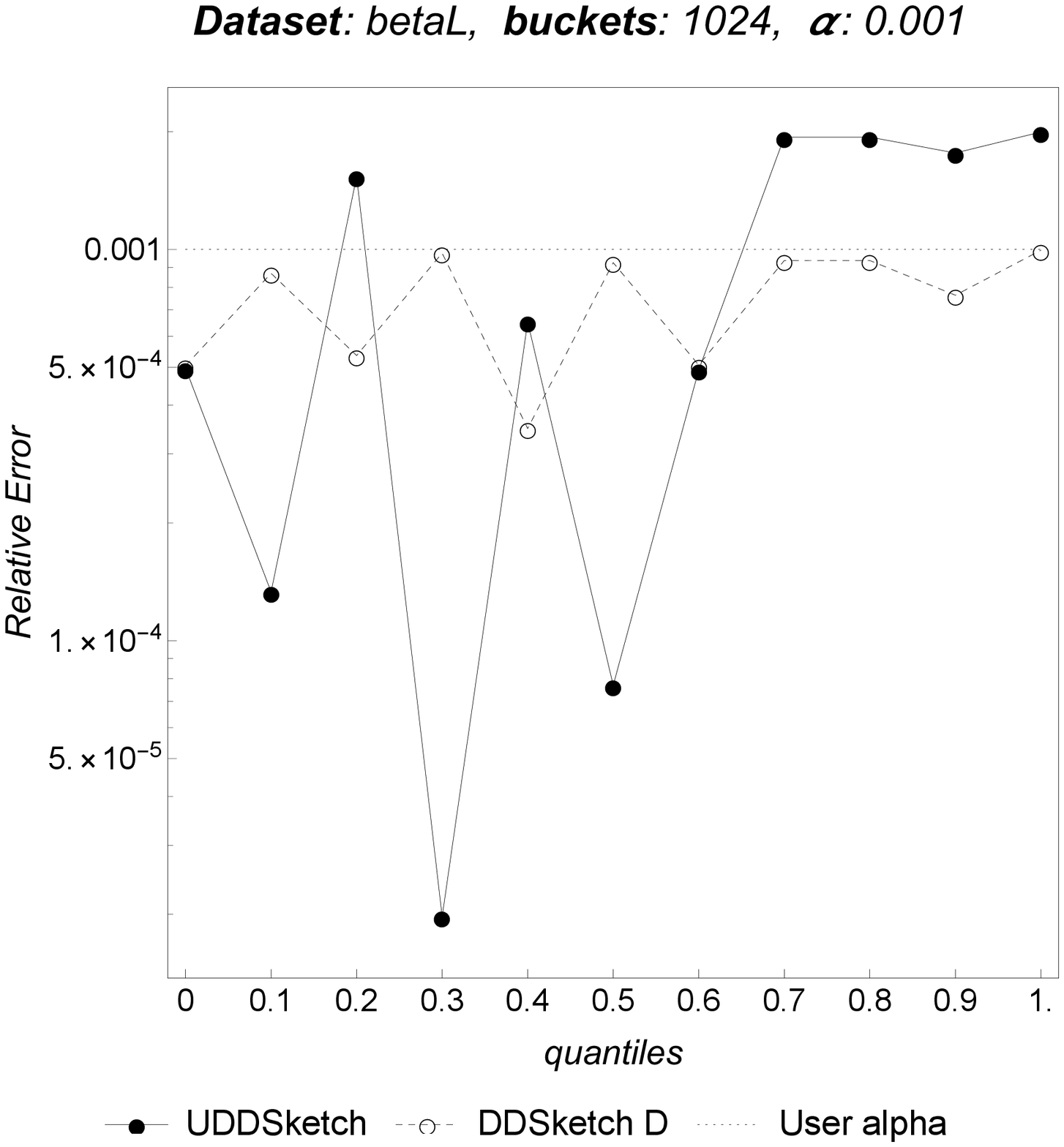}
			\label{betaL-allqs-ddsD}
		} \\
				
		\subfloat[]{
			\includegraphics[width=0.3\textwidth]{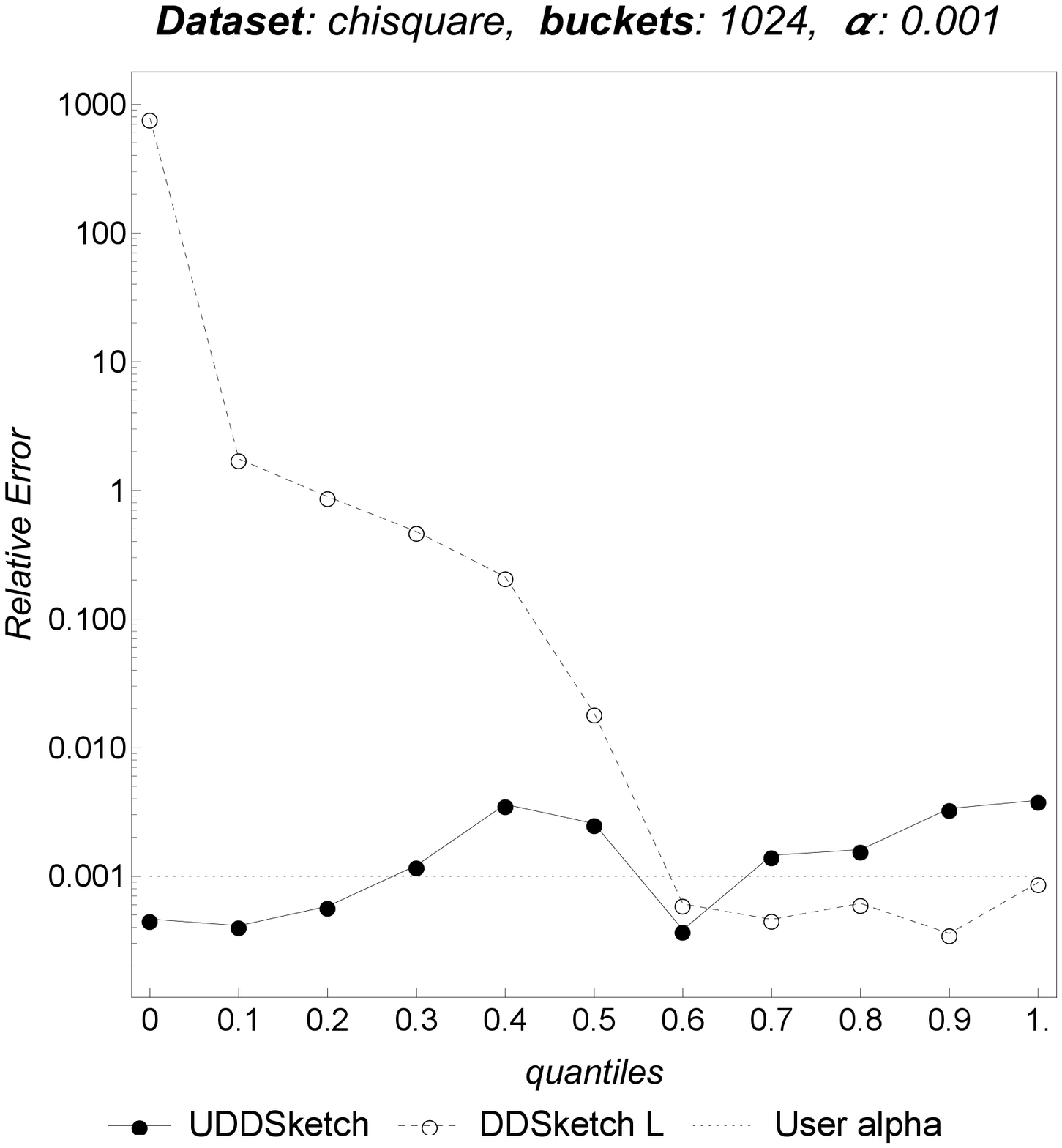}
			\label{chisquare-allqs-ddsL}
		} &
		
		\subfloat[]{
			\includegraphics[width=0.3\textwidth]{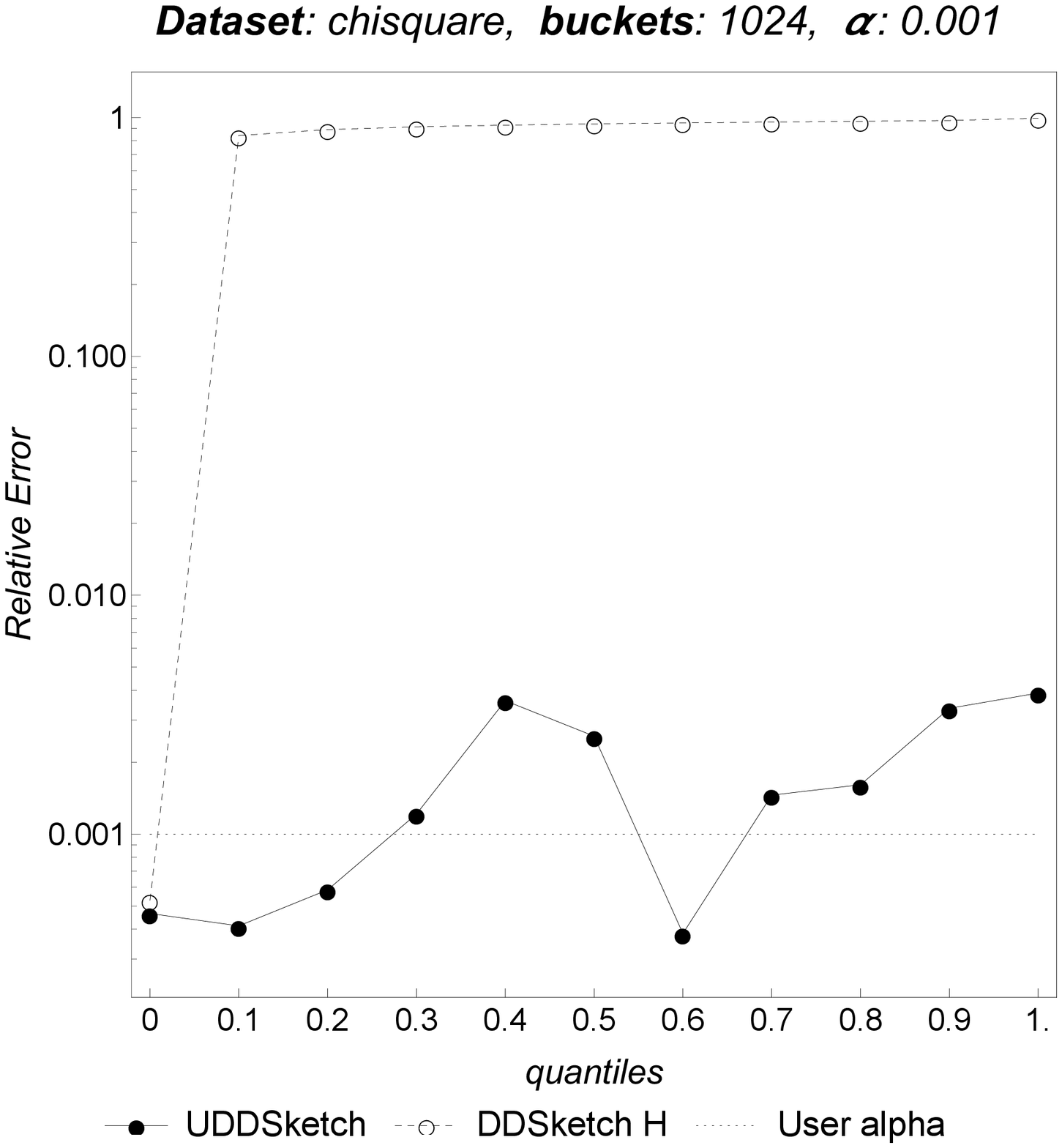}
			\label{chisquare-allqs-ddsH}
		} &
		
		\subfloat[]{
			\includegraphics[width=0.3\textwidth]{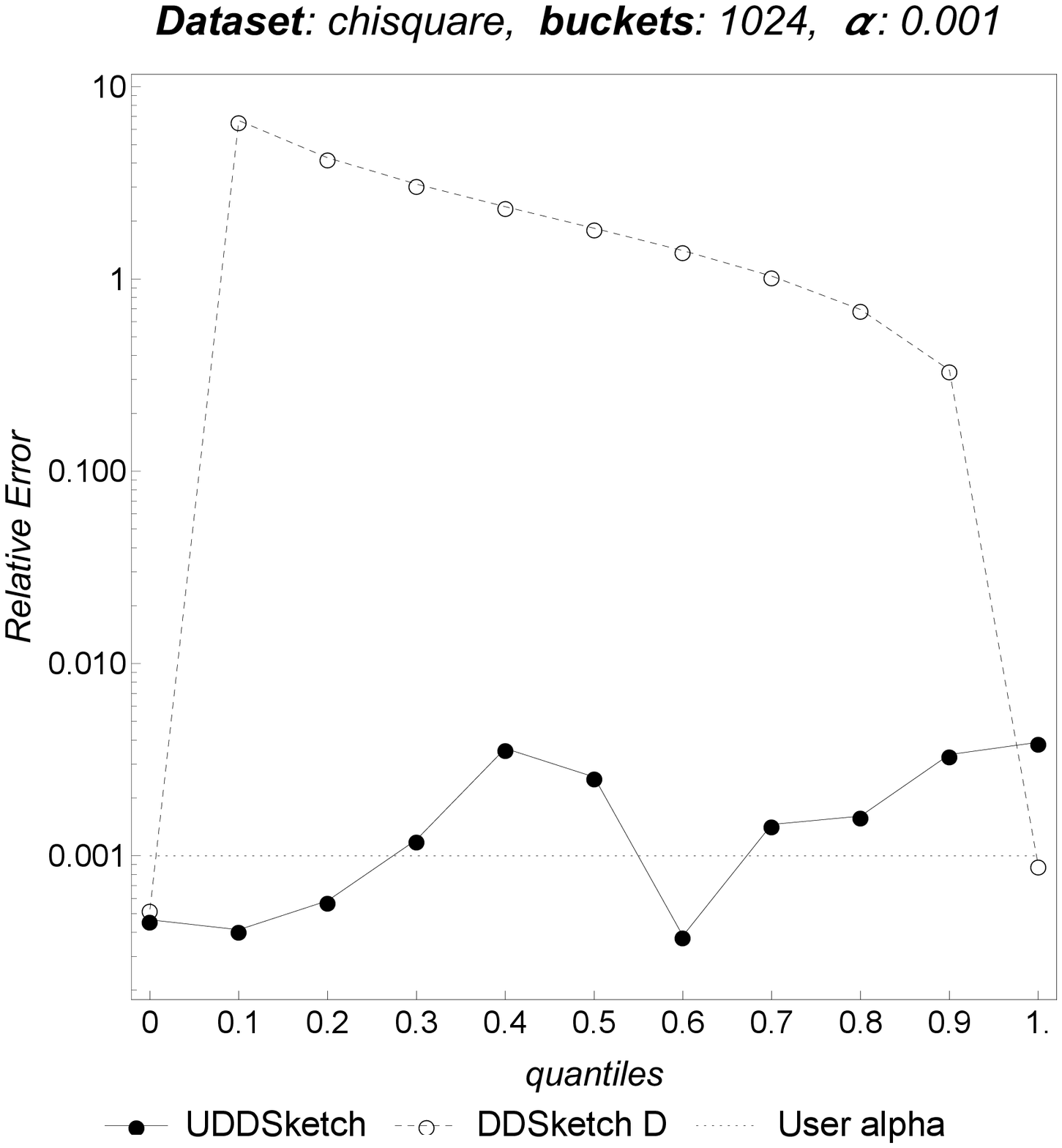}
			\label{chisquare-allqs-ddsD}
		} \\
	
		\subfloat[]{
			\includegraphics[width=0.3\textwidth]{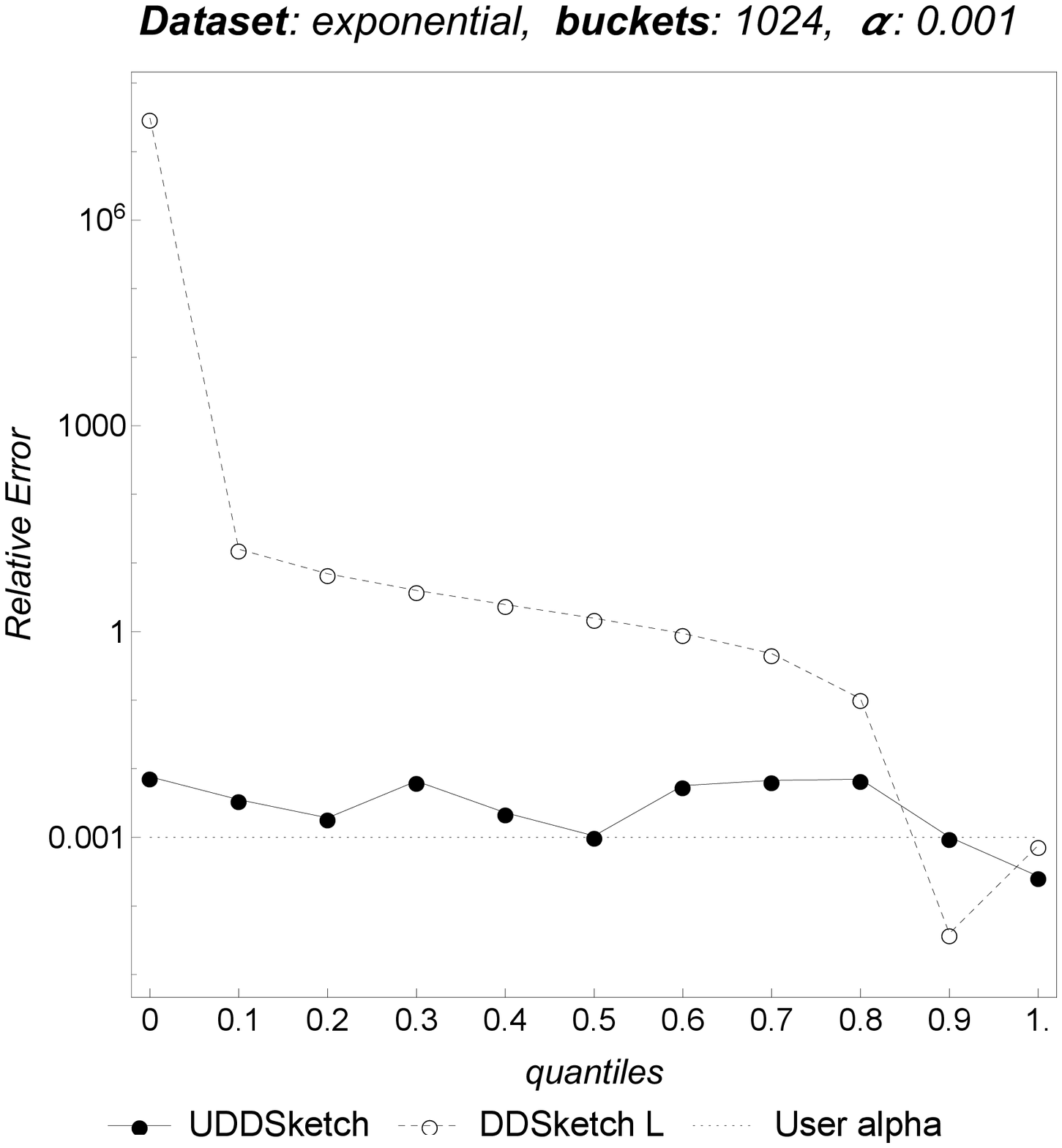}
			\label{exponential-allqs-ddsL}
		} &
		
		\subfloat[]{
			\includegraphics[width=0.3\textwidth]{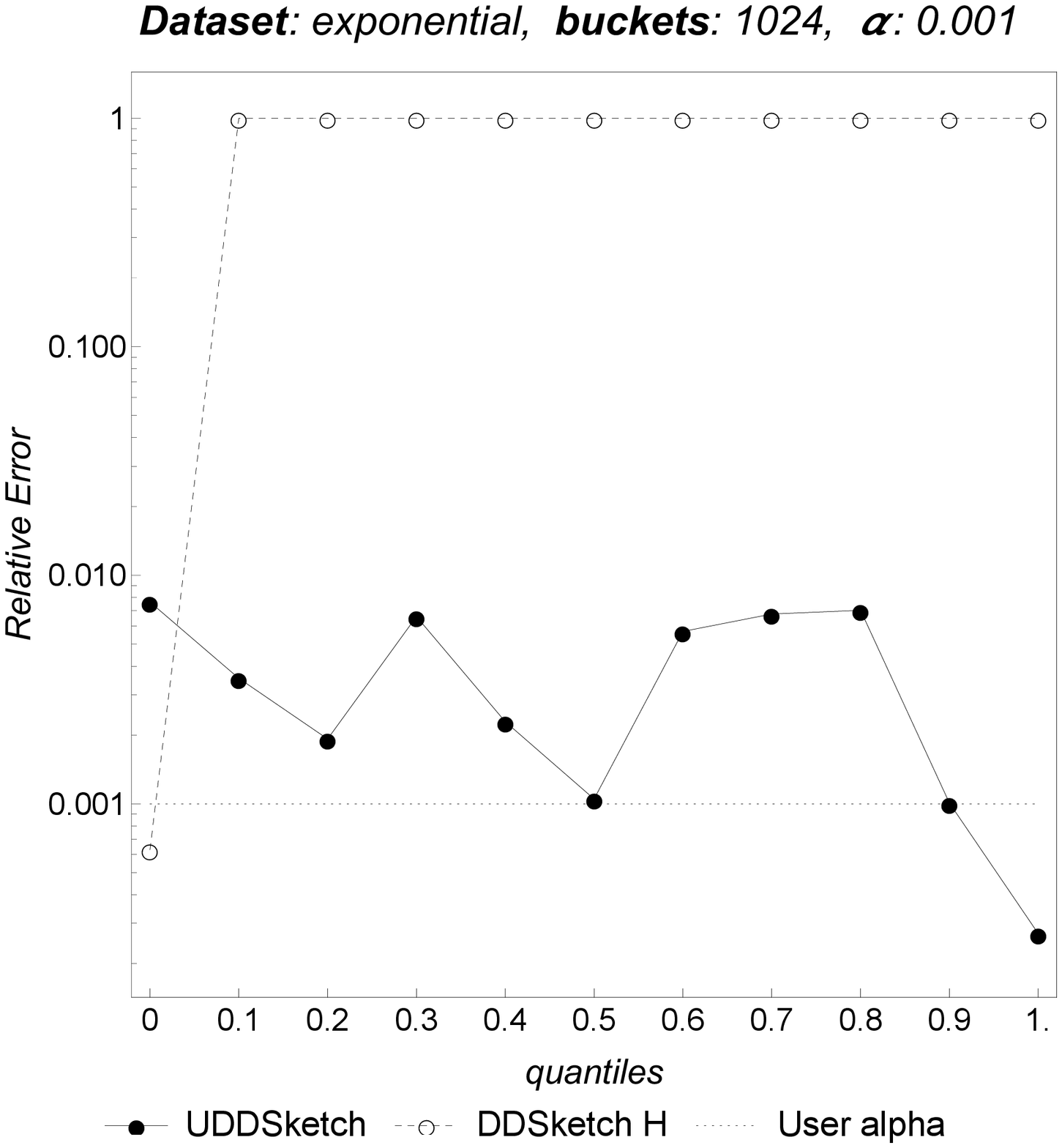}
			\label{exponential-allqs-ddsH}
		} &
		
		\subfloat[]{
			\includegraphics[width=0.3\textwidth]{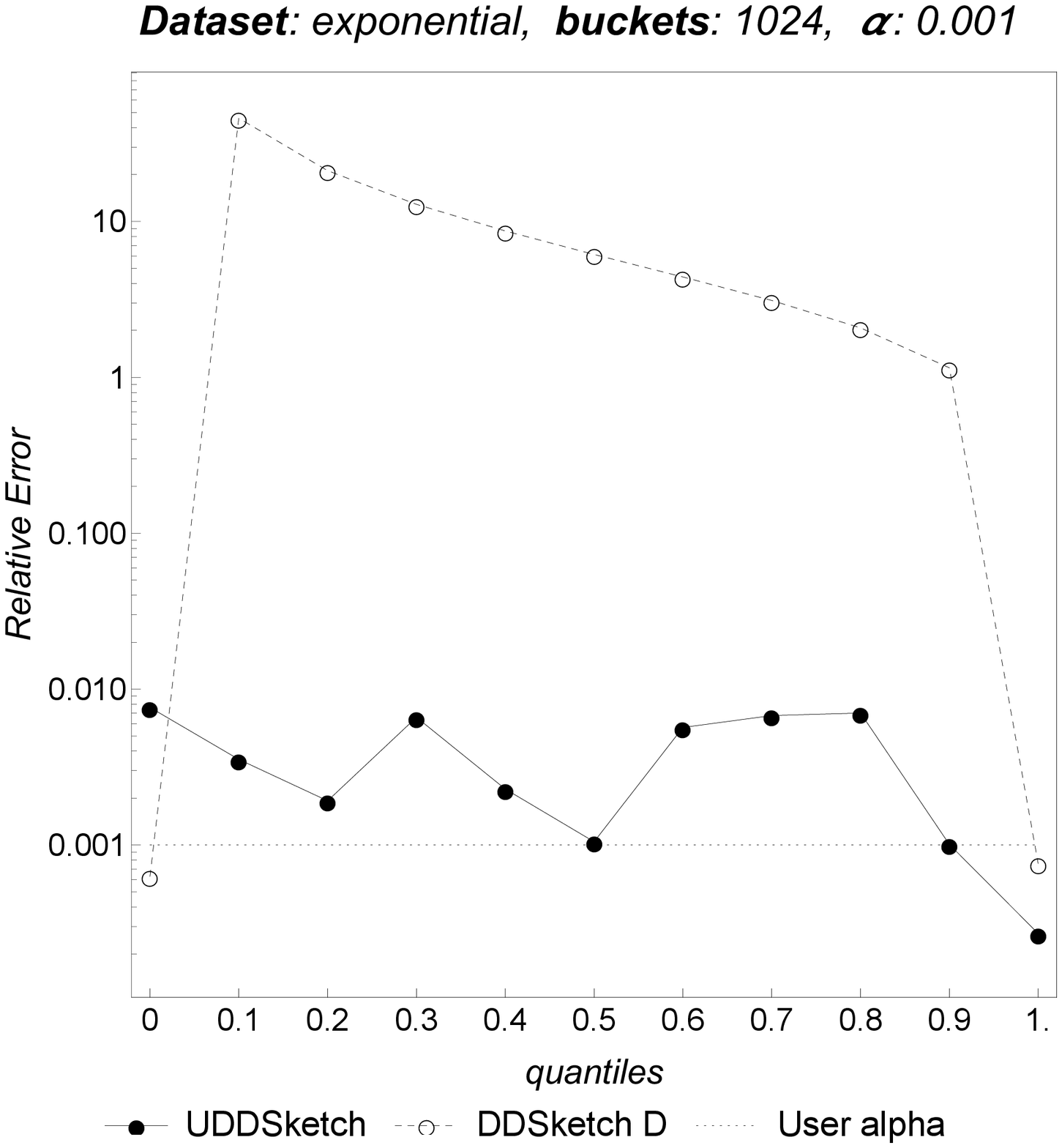}
			\label{exponential-allqs-ddsD}
		}  
	\end{tabular}
	
	\caption{Relative errors on quantiles $q_0, q_{0.1}, q_{0.2} \dots q_1$, varying the distribution and collapsing strategy.} 
	\label{allqs-plots1}
\end{figure*}

\begin{figure*}[h]
	\centering
	\begin{tabular}{ccc}		
		\subfloat[]{
			\includegraphics[width=0.3\textwidth]{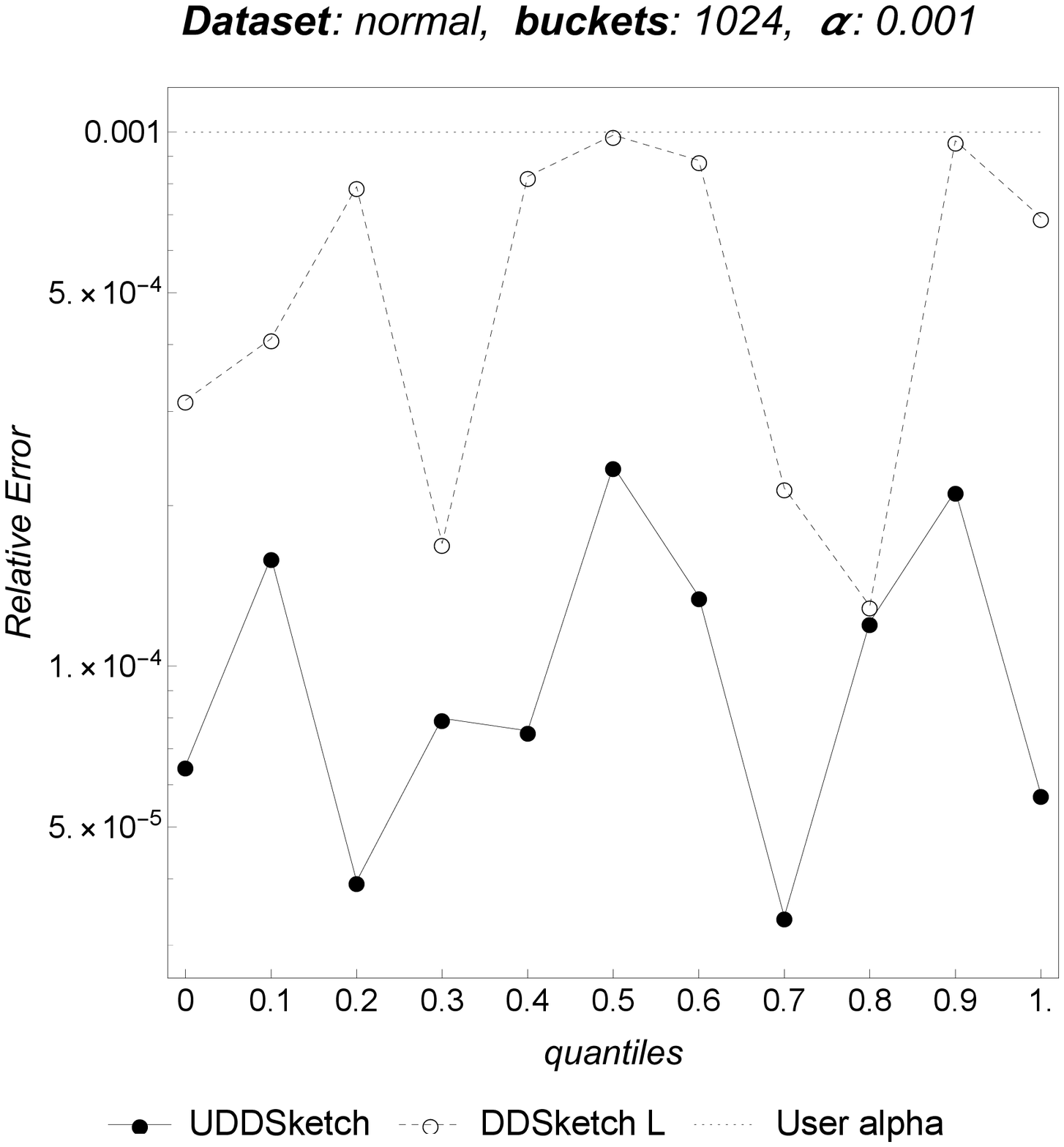}
			\label{normal-allqs-ddsL}
		} &
		
		\subfloat[]{
			\includegraphics[width=0.3\textwidth]{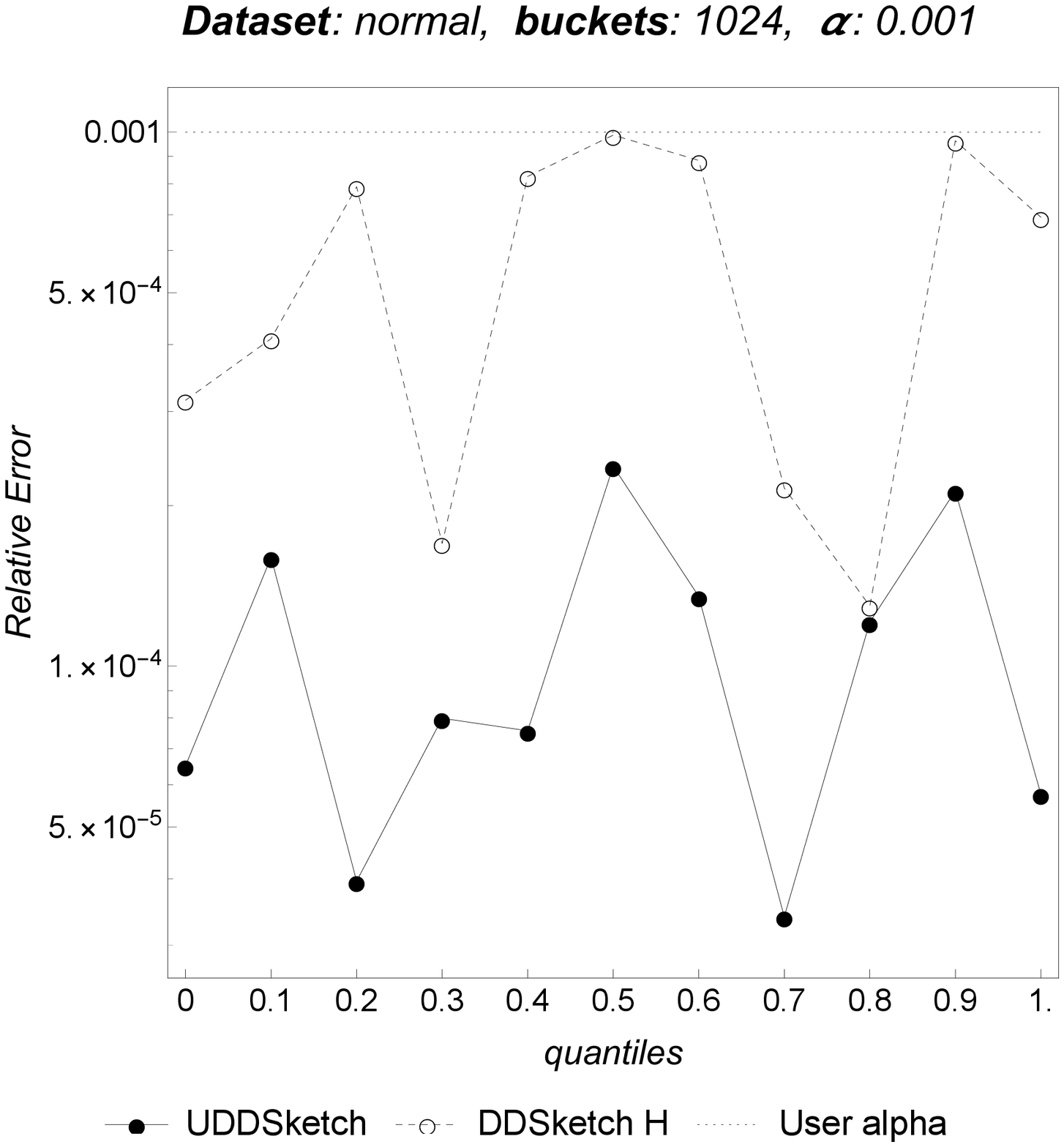}
			\label{normal-allqs-ddsH}
		} &
		
		\subfloat[]{
			\includegraphics[width=0.3\textwidth]{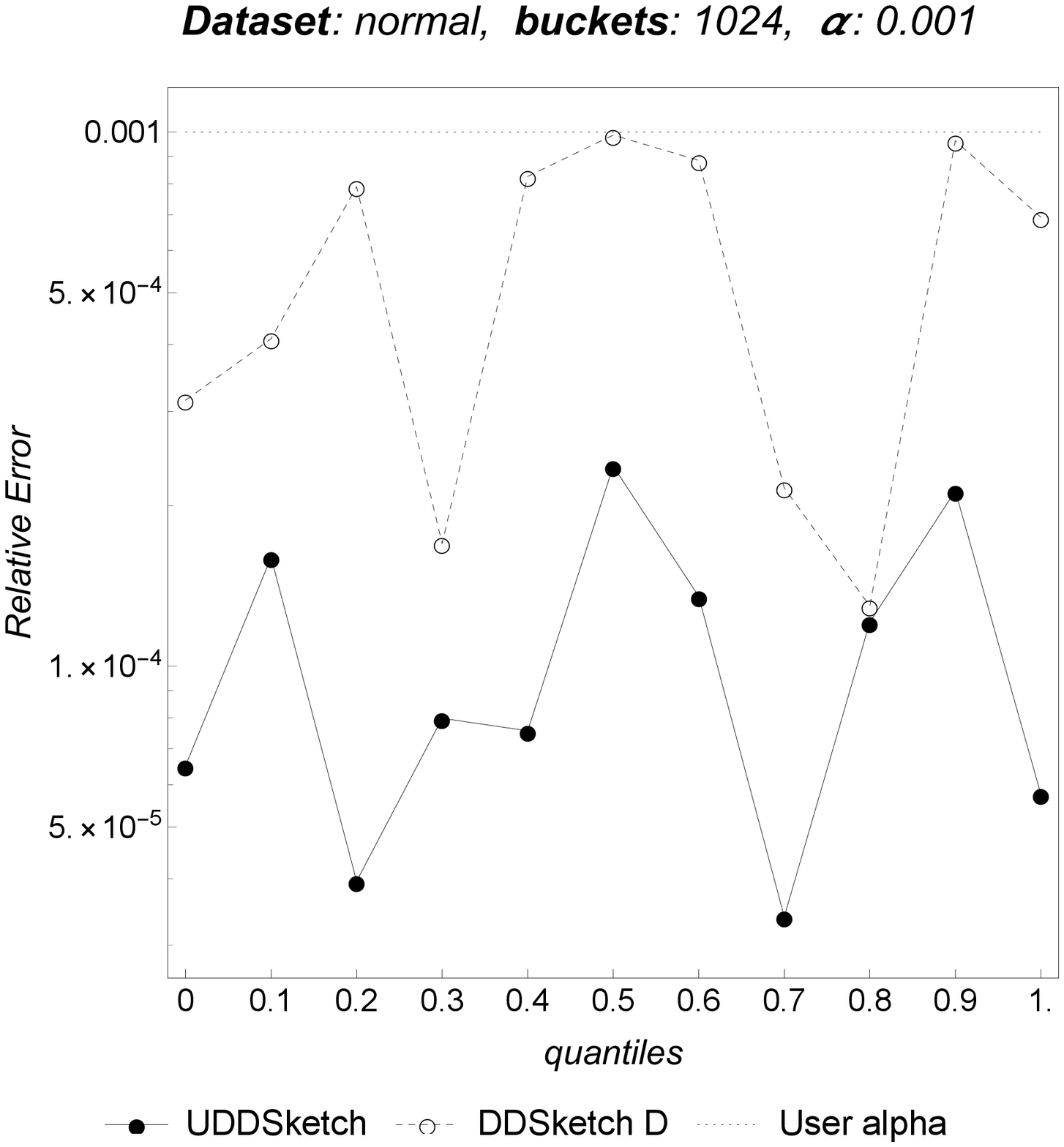}
			\label{normal-allqs-ddsD}
		} \\
		
		\subfloat[]{
			\includegraphics[width=0.3\textwidth]{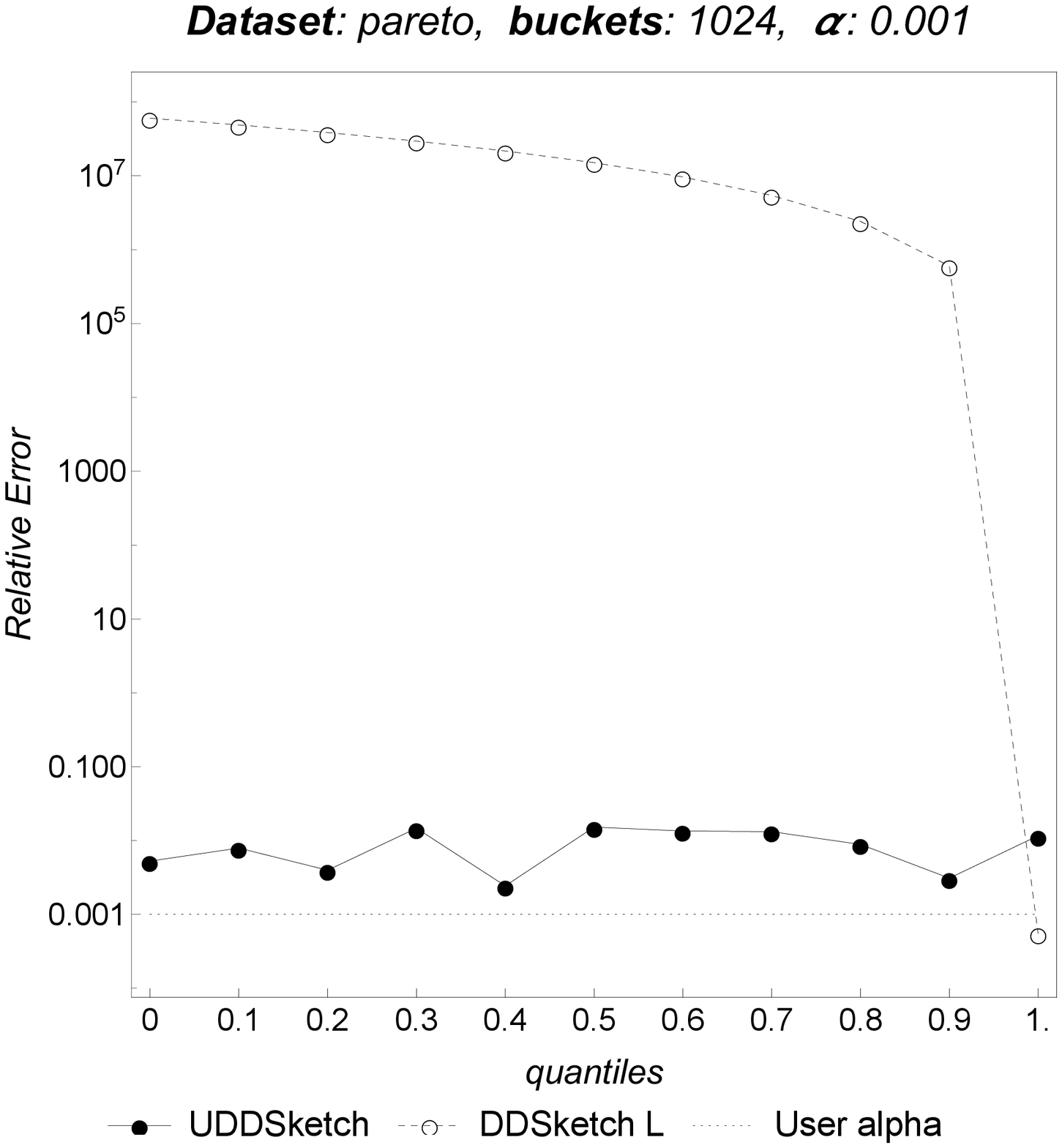}
			\label{pareto-allqs-ddsL}
		} &
		
		\subfloat[]{
			\includegraphics[width=0.3\textwidth]{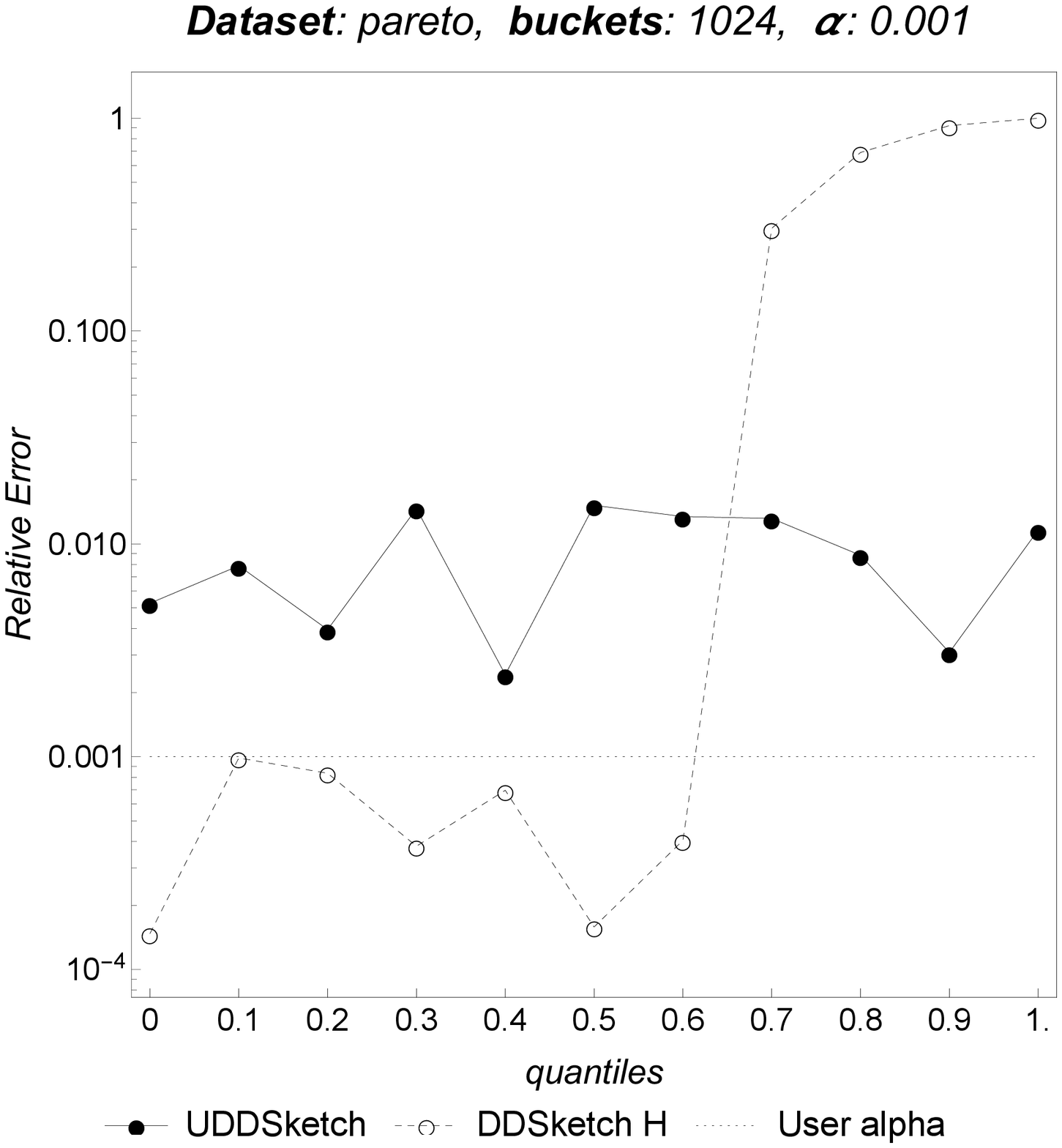}
			\label{pareto-allqs-ddsH}
		} &
		
		\subfloat[]{
			\includegraphics[width=0.3\textwidth]{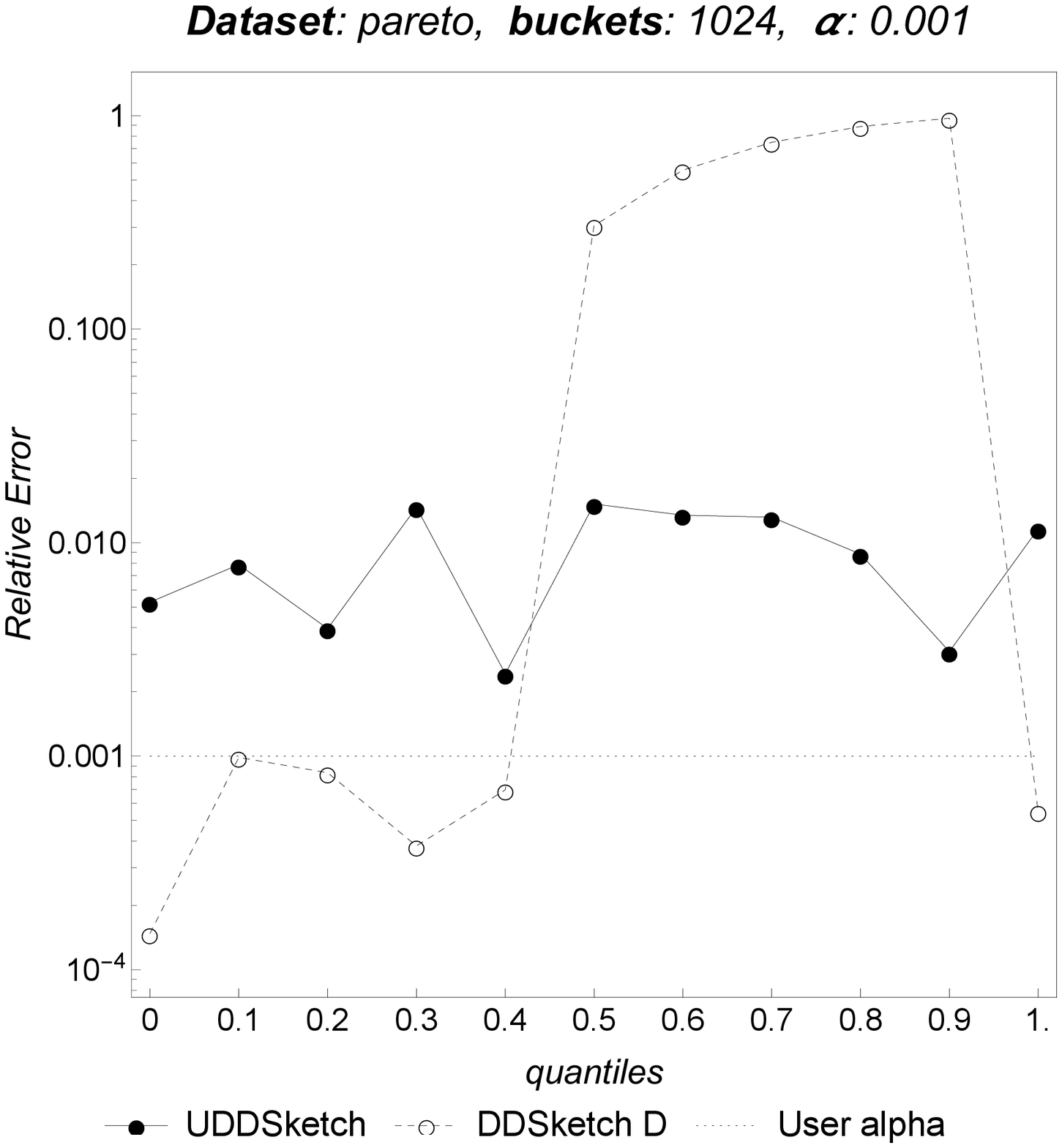}
			\label{pareto-allqs-ddsD}
		} \\
		
		\subfloat[]{
			\includegraphics[width=0.3\textwidth]{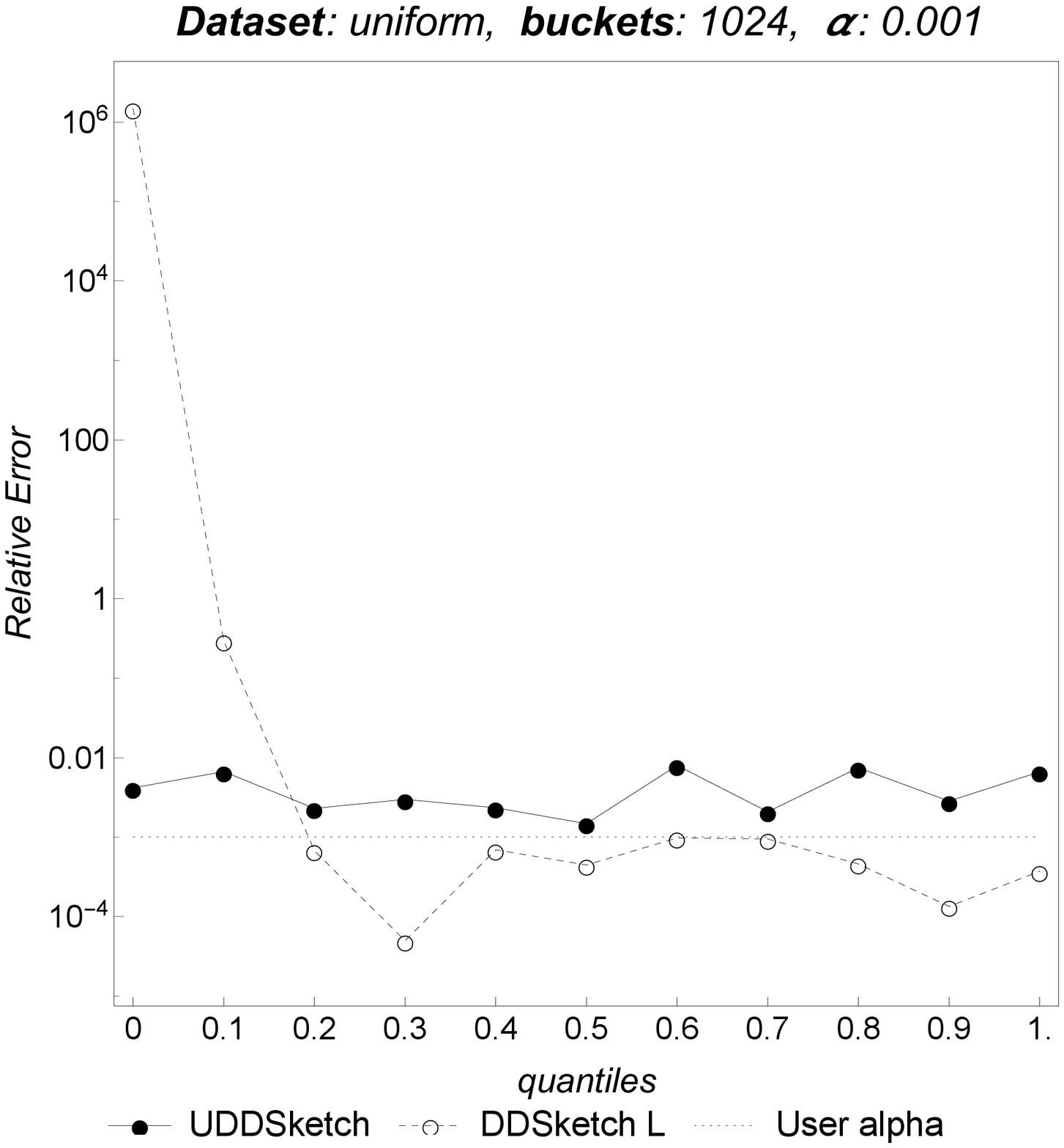}
			\label{uniform-allqs-ddsL}
		} &
		
		\subfloat[]{
			\includegraphics[width=0.3\textwidth]{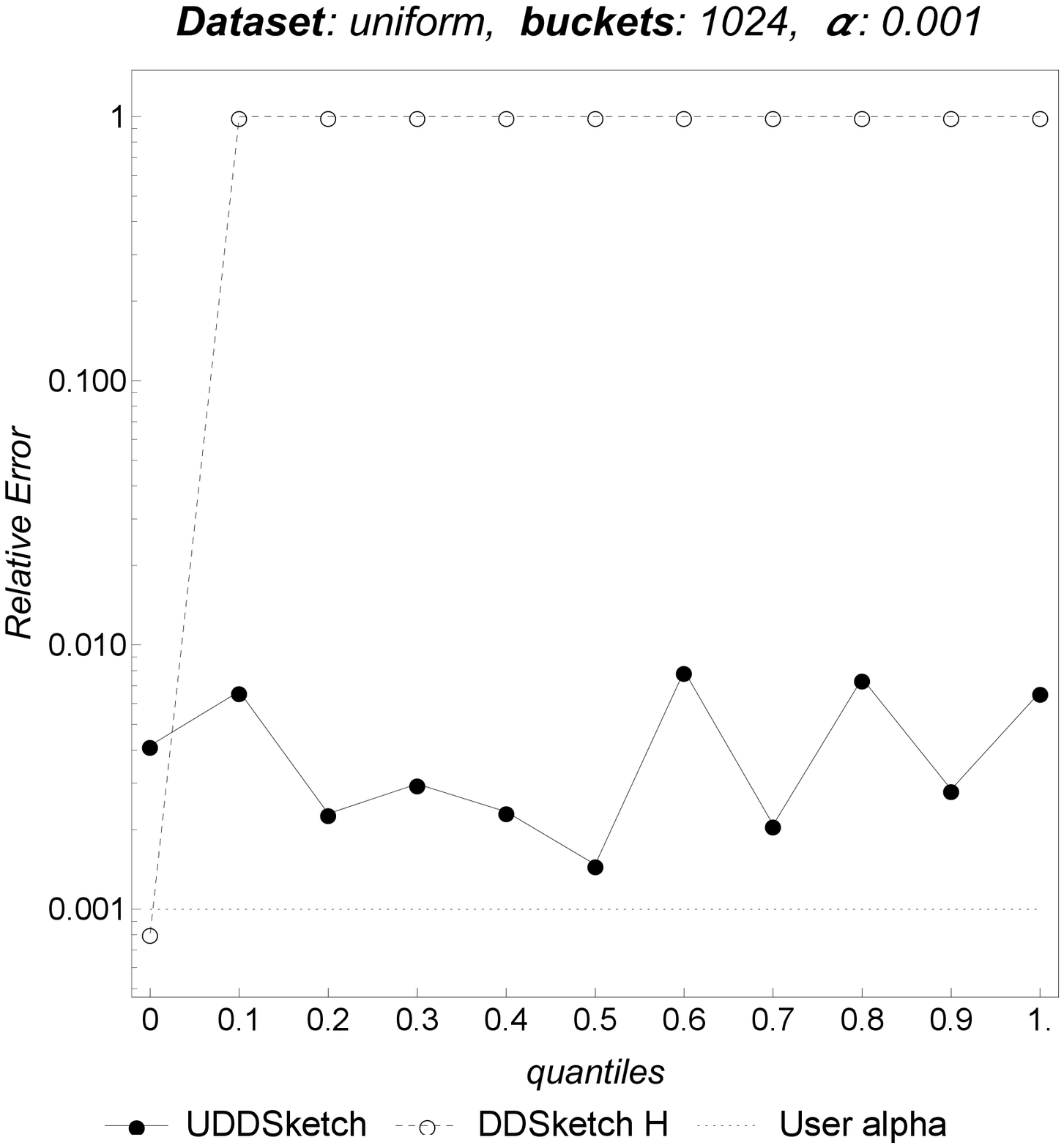}
			\label{uniform-allqs-ddsH}
		} &
		
		\subfloat[]{
			\includegraphics[width=0.3\textwidth]{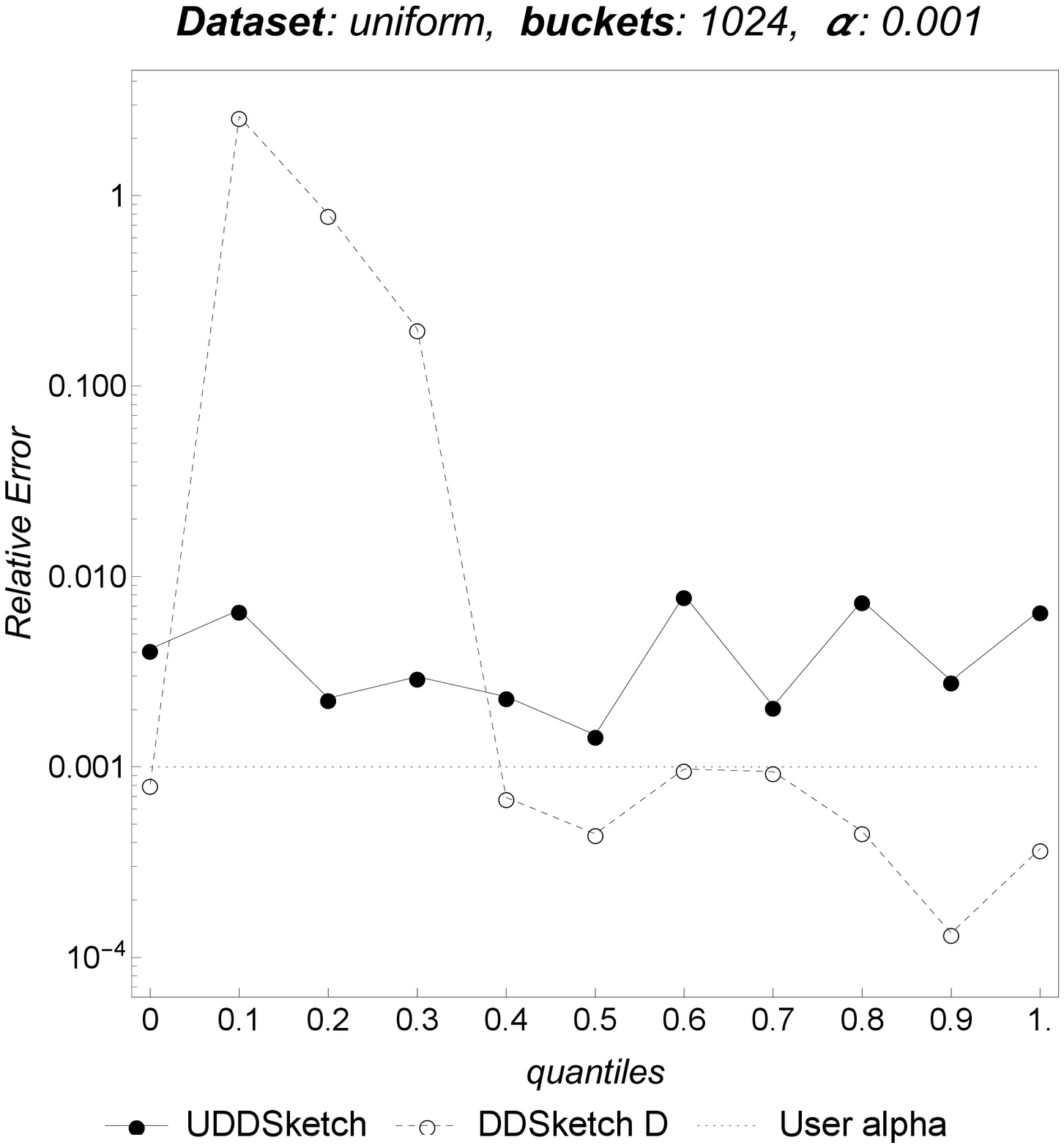}
			\label{uniform-allqs-ddsD}
		}  
	\end{tabular}
	
	\caption{Relative errors on quantiles $q_0, q_{0.1}, q_{0.2} \dots q_1$, varying the distribution and collapsing strategy.} 
	\label{allqs-plots2}
\end{figure*}

\begin{figure*}[h]
	\centering
	\begin{tabular}{ccc}		
		\subfloat[]{
			\includegraphics[width=0.3\textwidth]{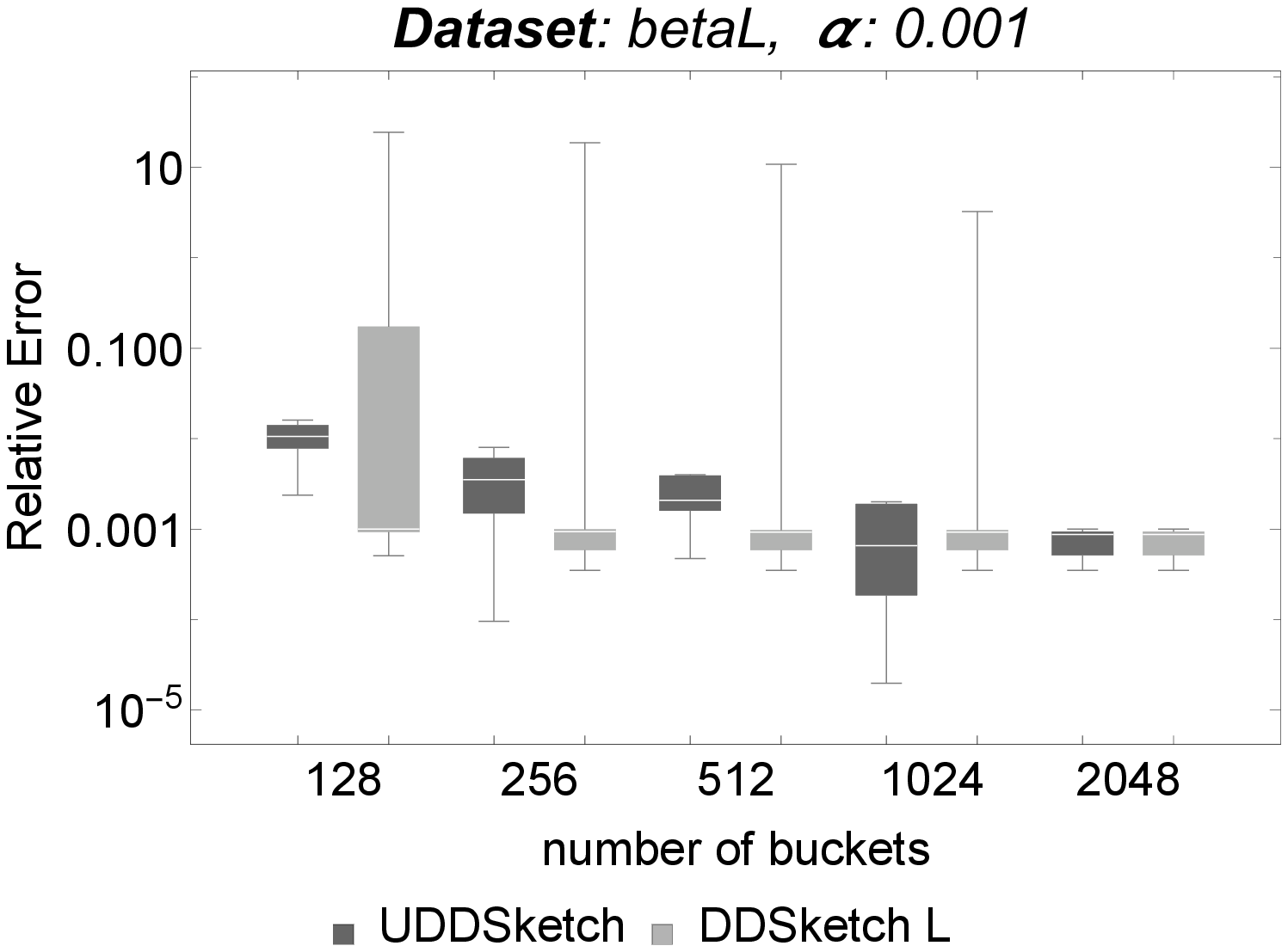}
			\label{betaL-boxplot-ddsL}
		} &
		
		\subfloat[]{
			\includegraphics[width=0.3\textwidth]{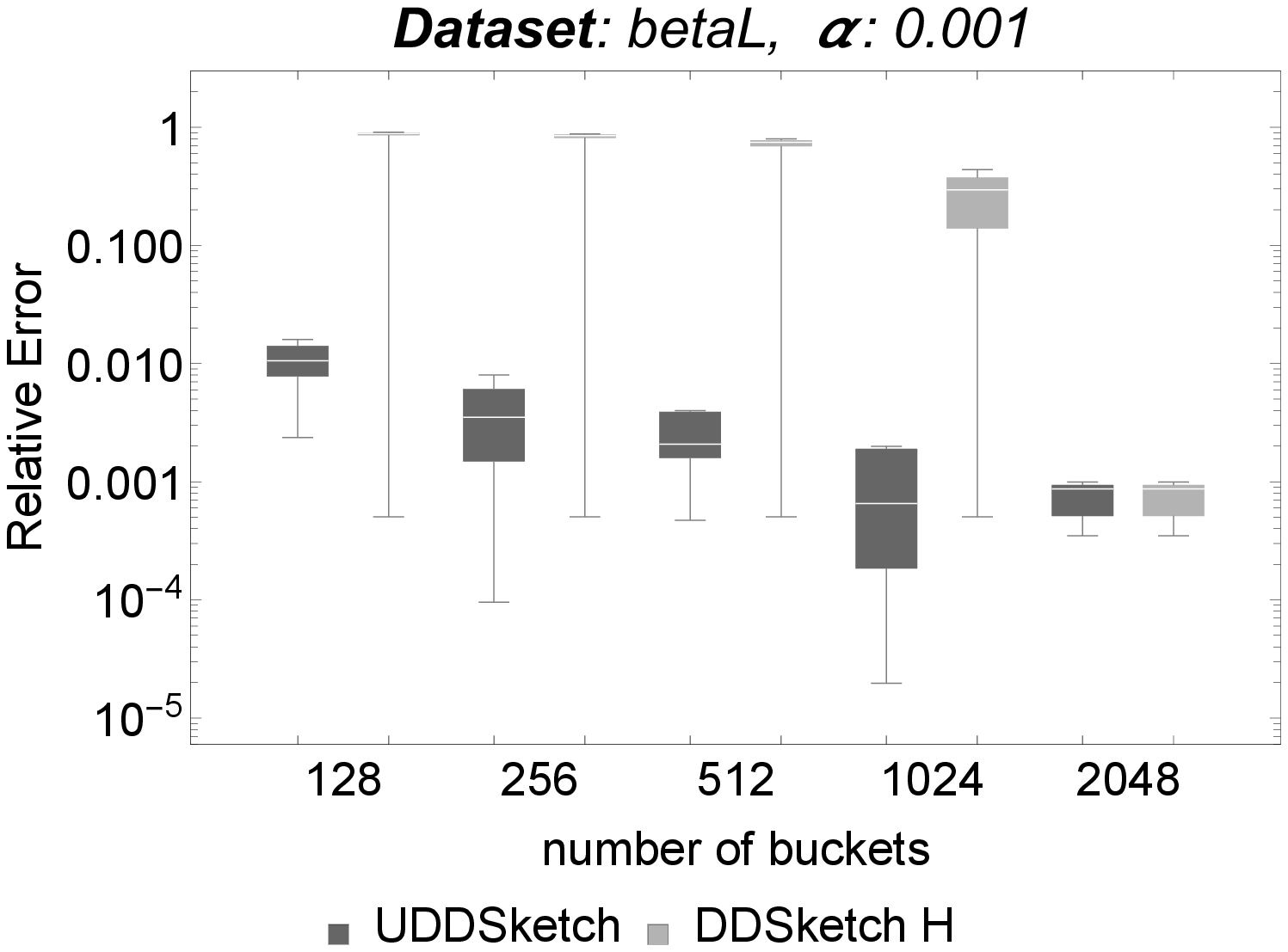}
			\label{betaL-boxplot-ddsH}
		} &
		
		\subfloat[]{
			\includegraphics[width=0.3\textwidth]{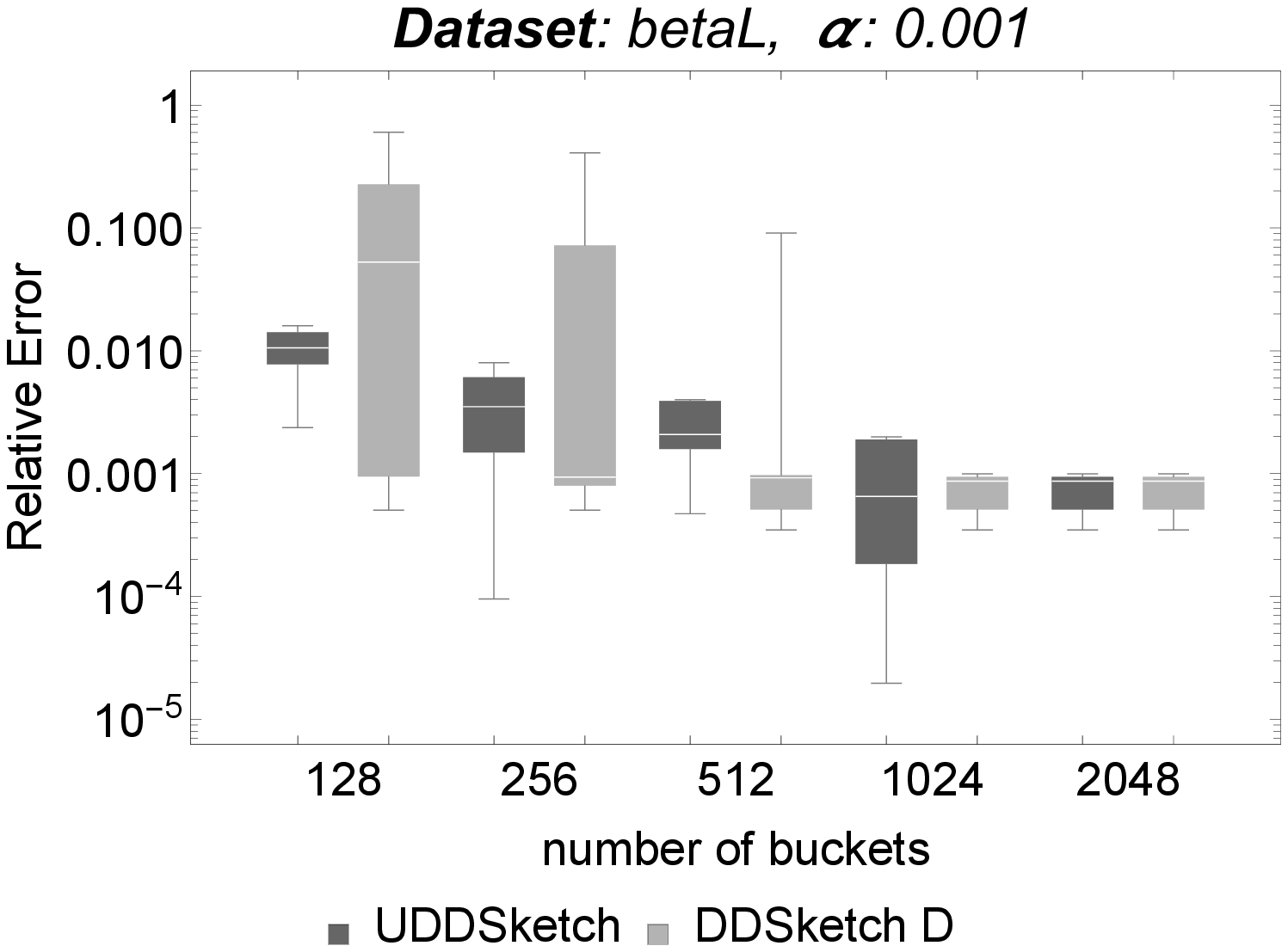}
			\label{betaL-boxplot-ddsD}
		} \\
		
		\subfloat[]{
			\includegraphics[width=0.3\textwidth]{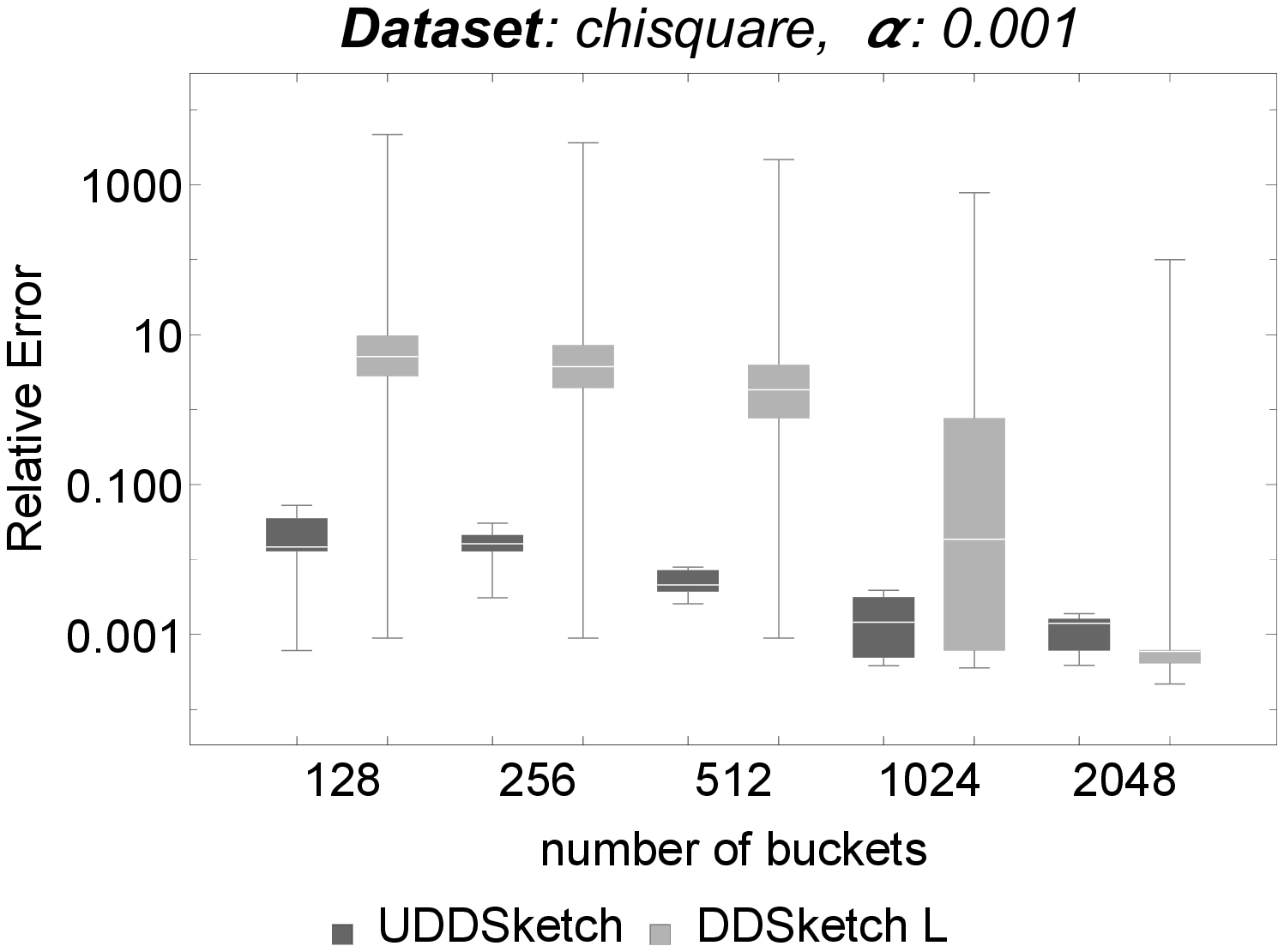}
			\label{chisquare-boxplot-ddsL}
		} &
		
		\subfloat[]{
			\includegraphics[width=0.3\textwidth]{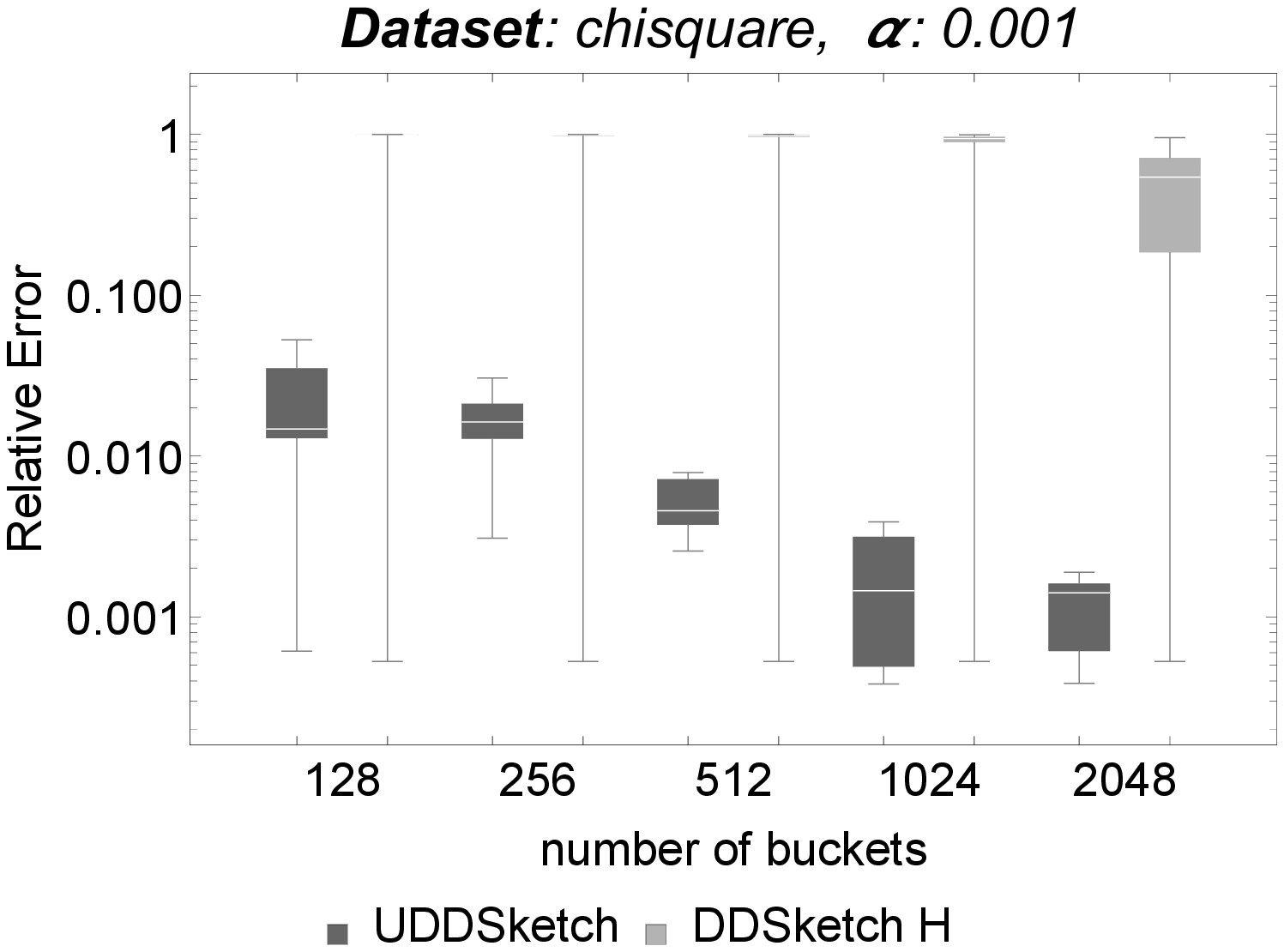}
			\label{chisquare-boxplot-ddsH}
		} &
		
		\subfloat[]{
			\includegraphics[width=0.3\textwidth]{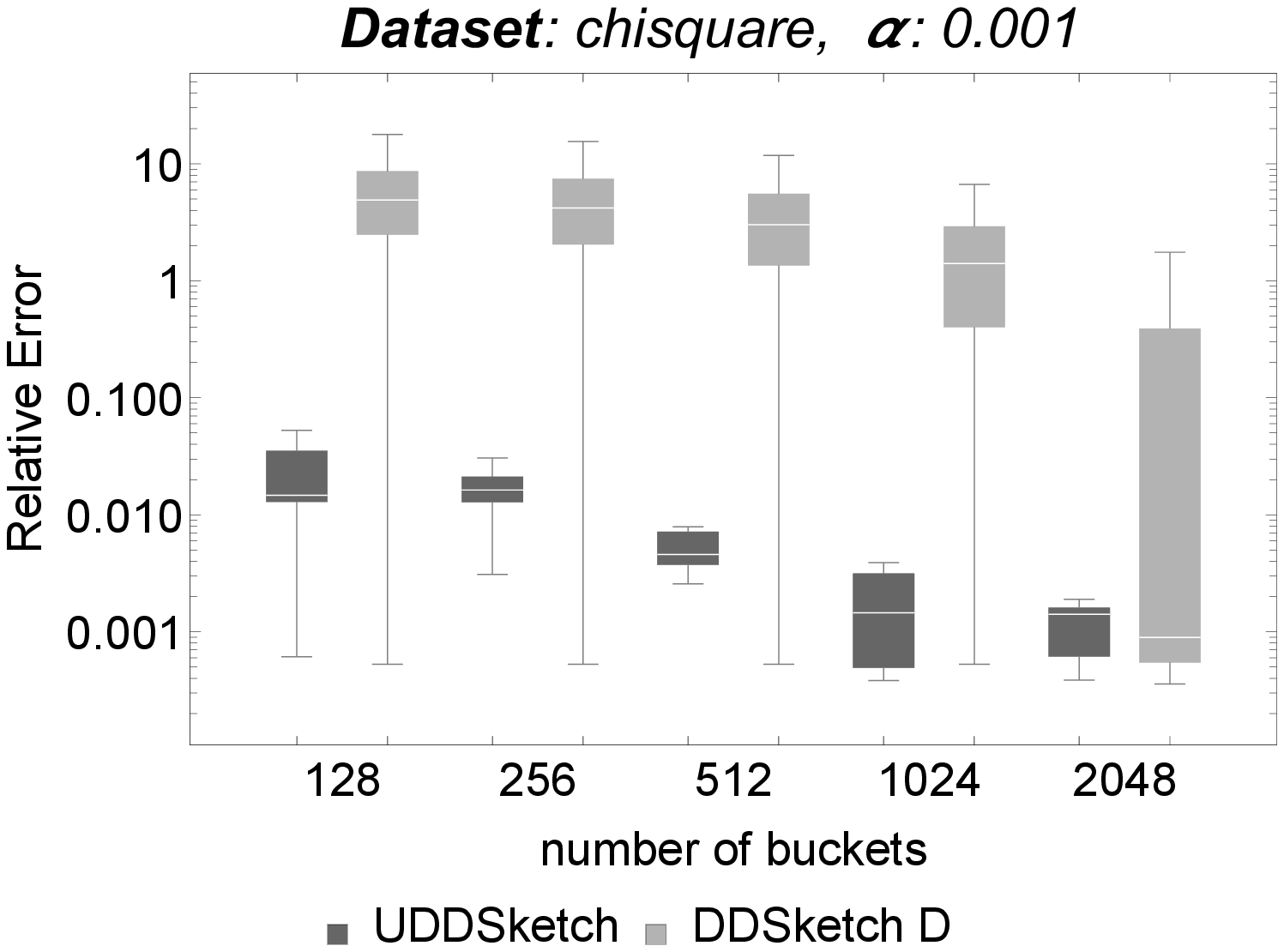}
			\label{chisquare-boxplot-ddsD}
		} \\
		
		\subfloat[]{
			\includegraphics[width=0.3\textwidth]{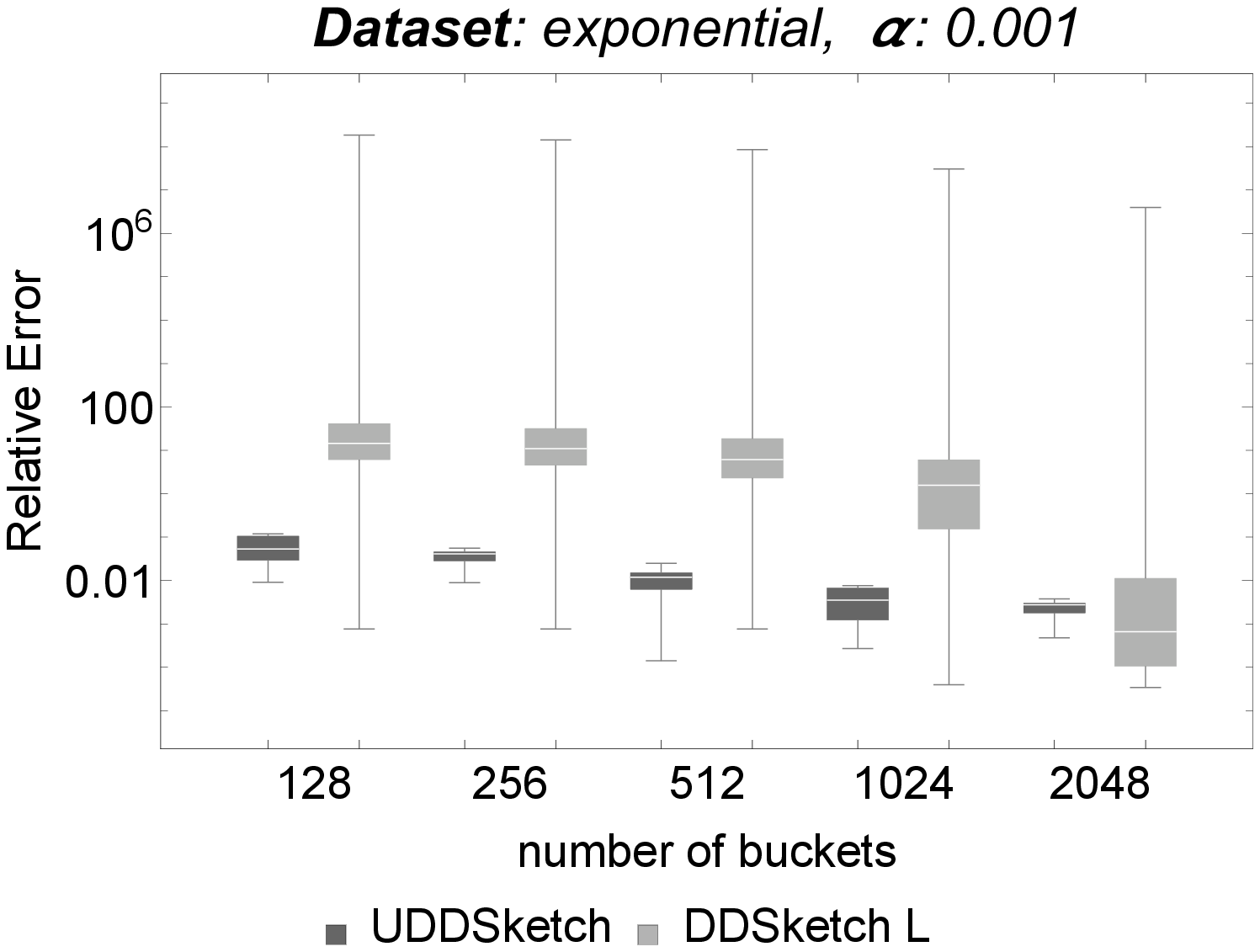}
			\label{exponential-boxplot-ddsL}
		} &
		
		\subfloat[]{
			\includegraphics[width=0.3\textwidth]{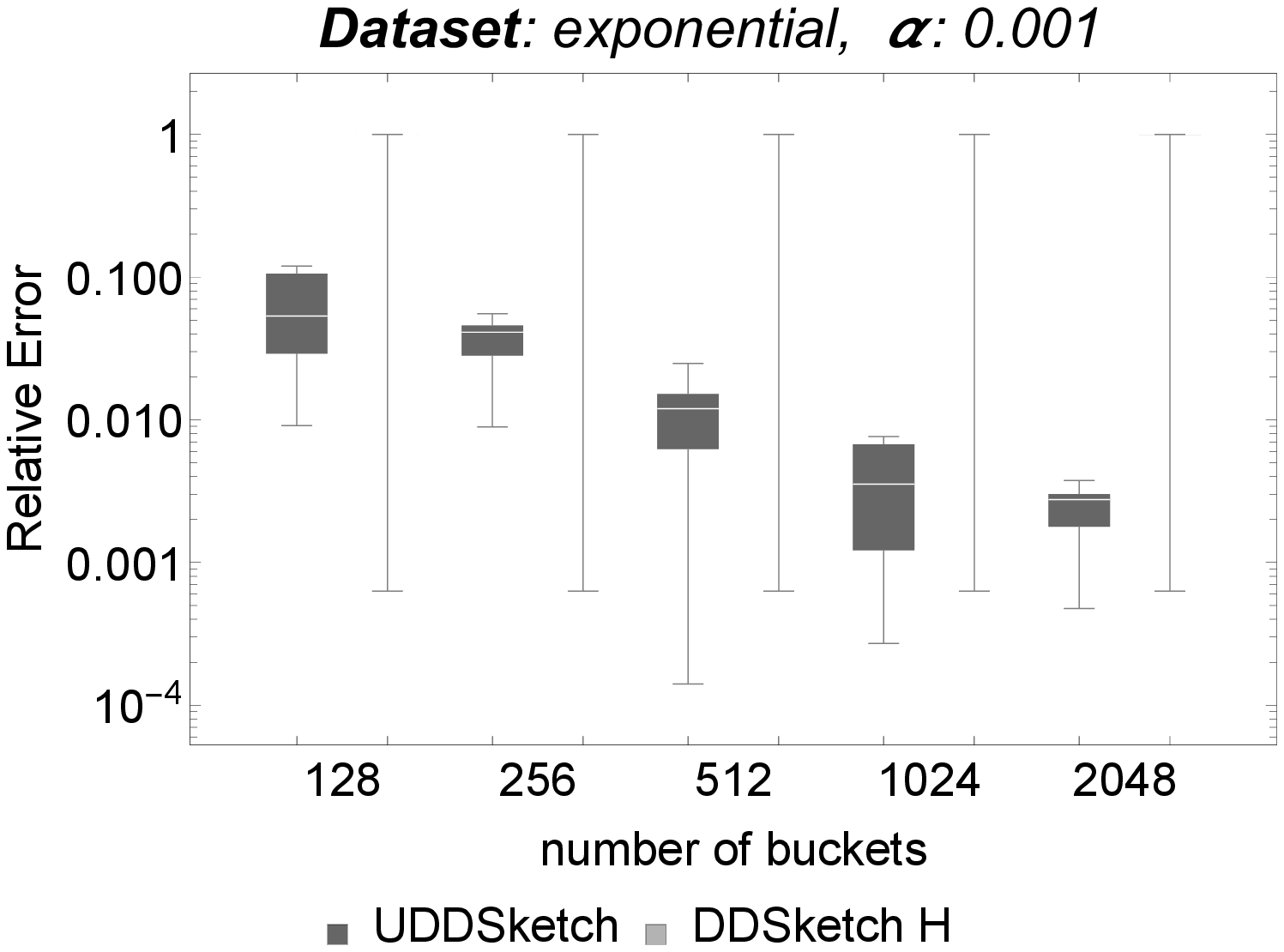}
			\label{exponential-boxplot-ddsH}
		} &
		
		\subfloat[]{
			\includegraphics[width=0.3\textwidth]{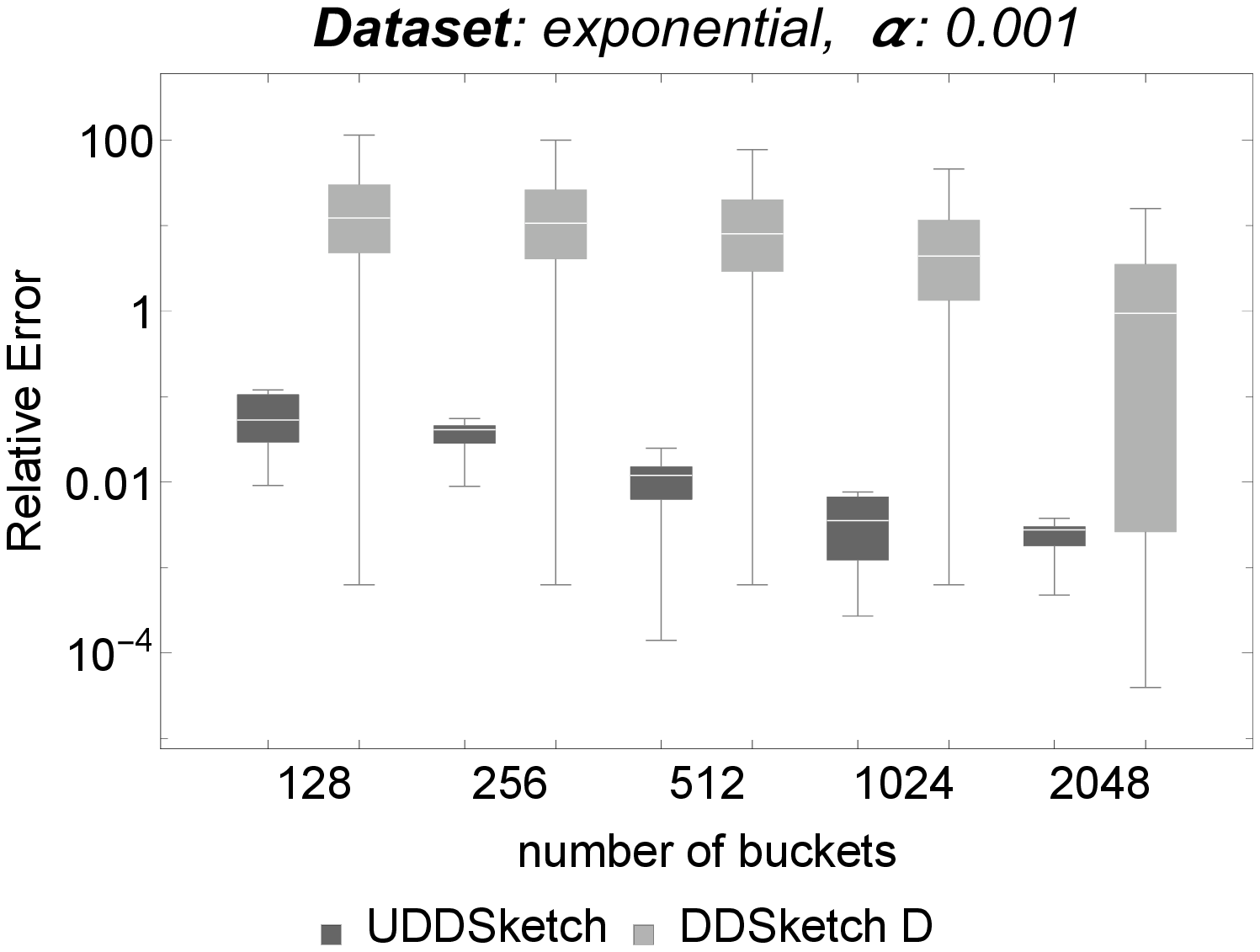}
			\label{exponential-boxplot-ddsD}
		} \\
	    
	    \subfloat[]{
	    	\includegraphics[width=0.3\textwidth]{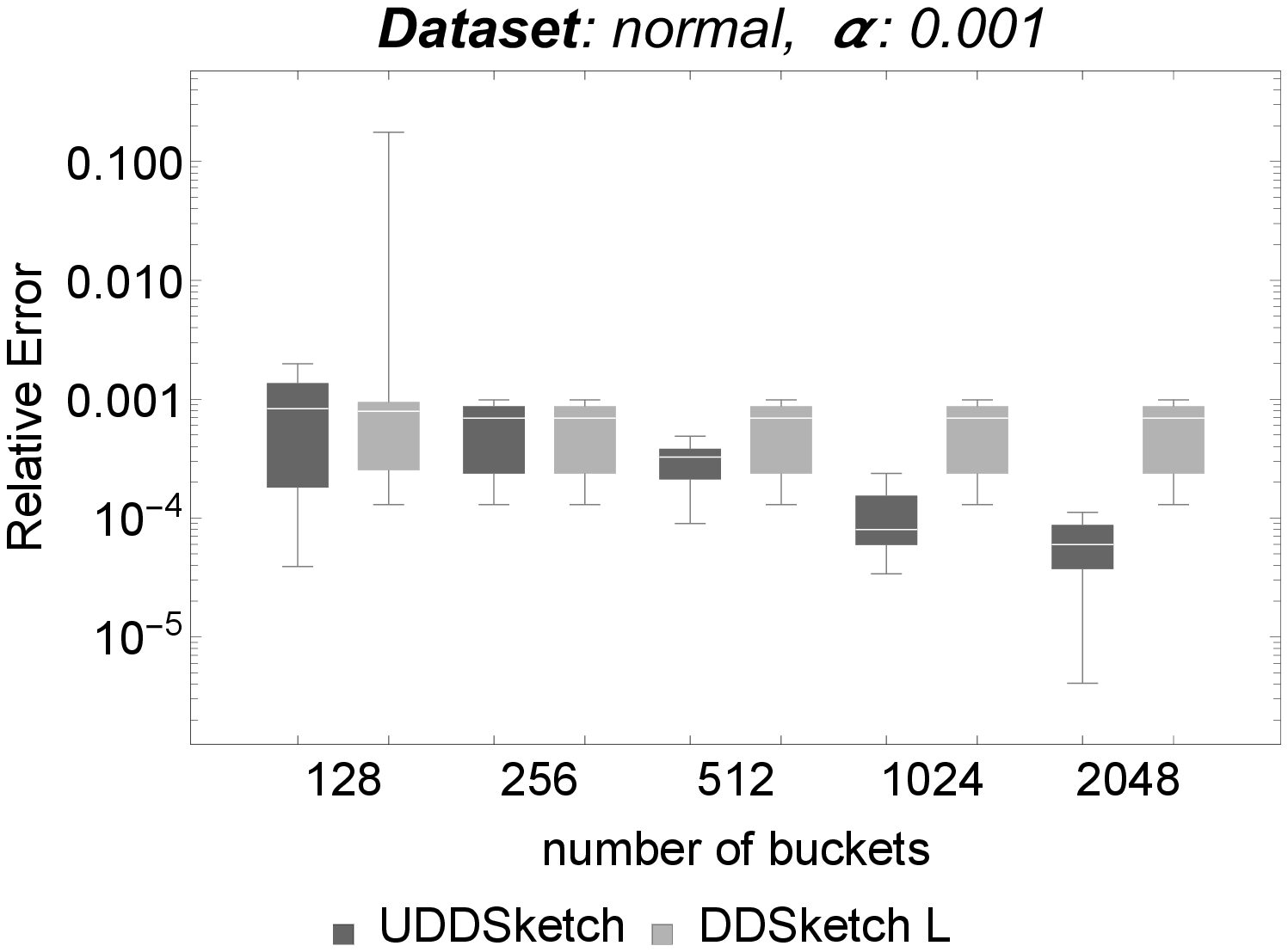}
	    	\label{normal-boxplot-ddsL}
	    } &
	    
	    \subfloat[]{
	    	\includegraphics[width=0.3\textwidth]{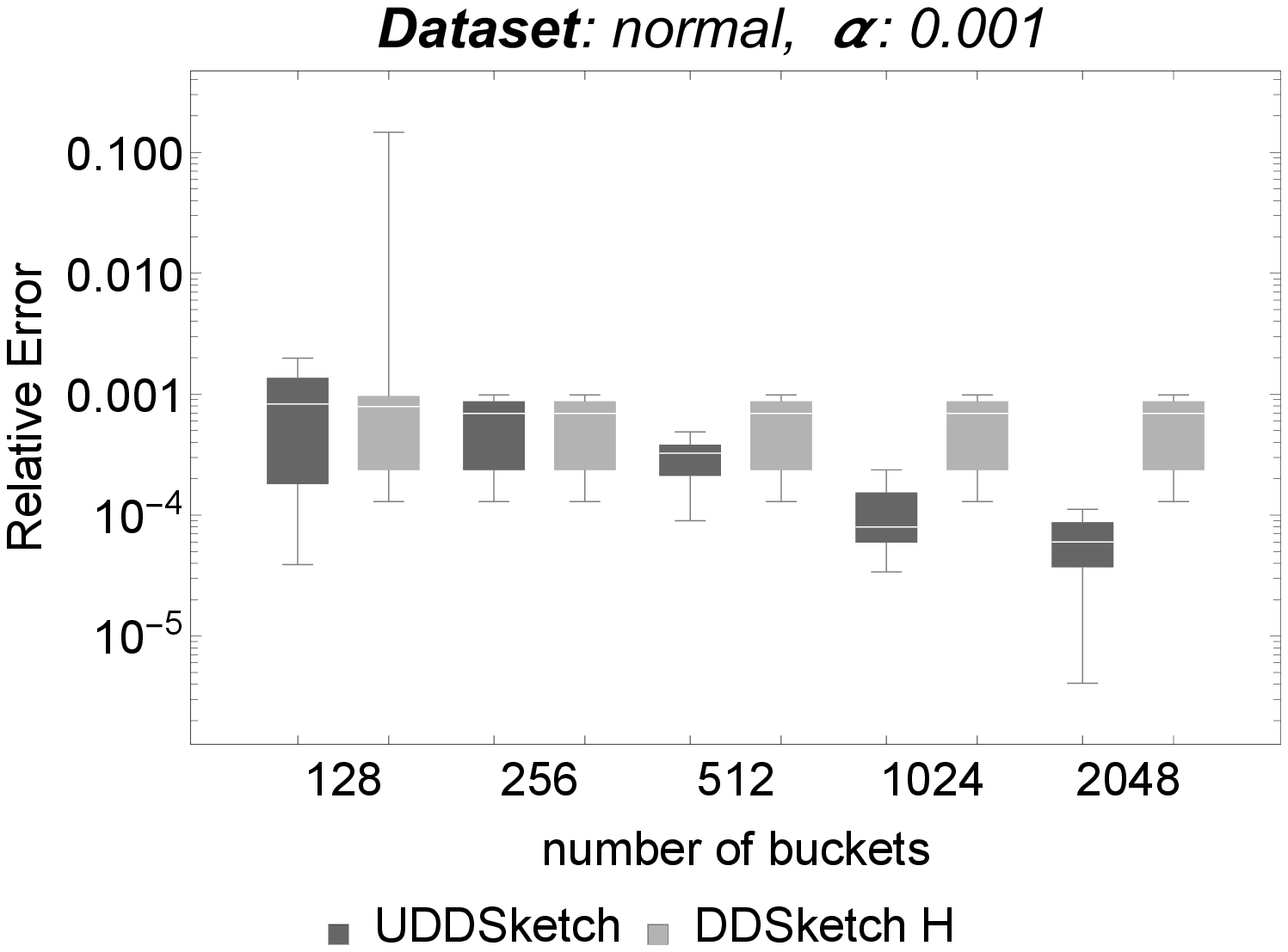}
	    	\label{normal-boxplot-ddsH}
	    } &
	    
	    \subfloat[]{
	    	\includegraphics[width=0.3\textwidth]{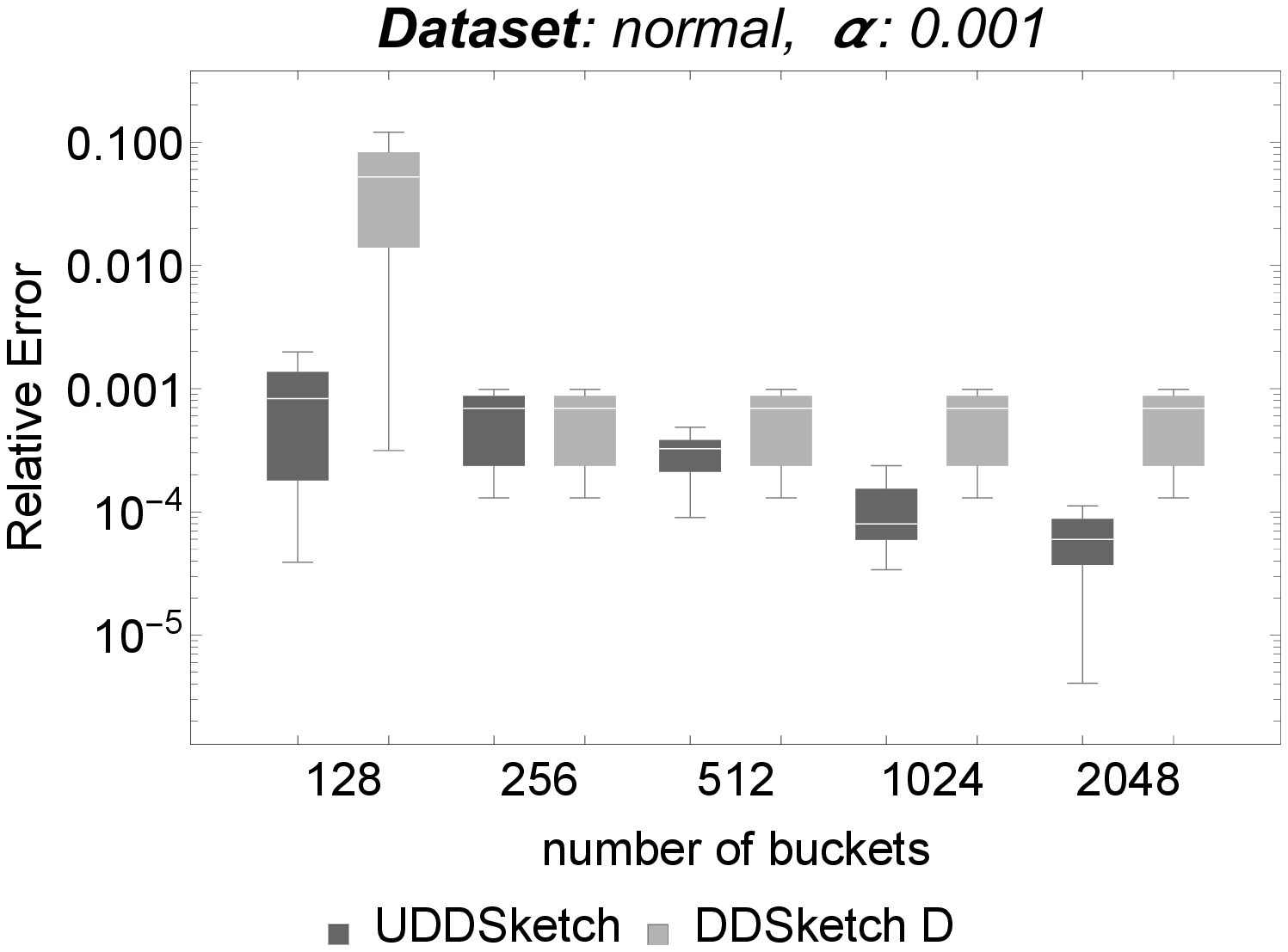}
	    	\label{normal-boxplot-ddsD}
	    } 
	\end{tabular}
	
	\caption{Relative errors on quantiles (boxplots), varying the number of buckets.} 
	\label{boxplot-plots1}
\end{figure*}

\begin{figure*}[h]
	\centering
	\begin{tabular}{ccc}		
		\subfloat[]{
			\includegraphics[width=0.3\textwidth]{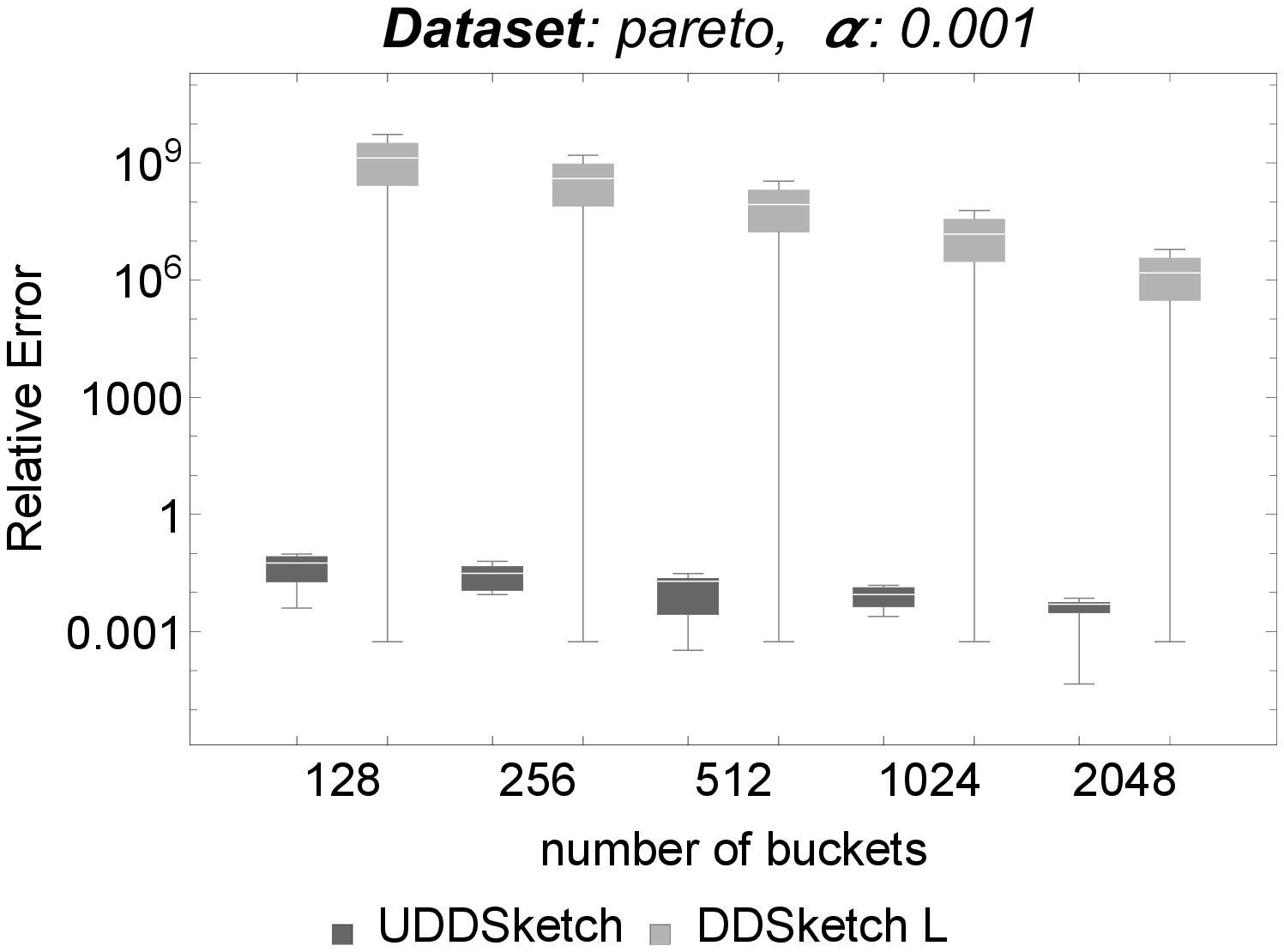}
			\label{pareto-boxplot-ddsL}
		} &
		
		\subfloat[]{
			\includegraphics[width=0.3\textwidth]{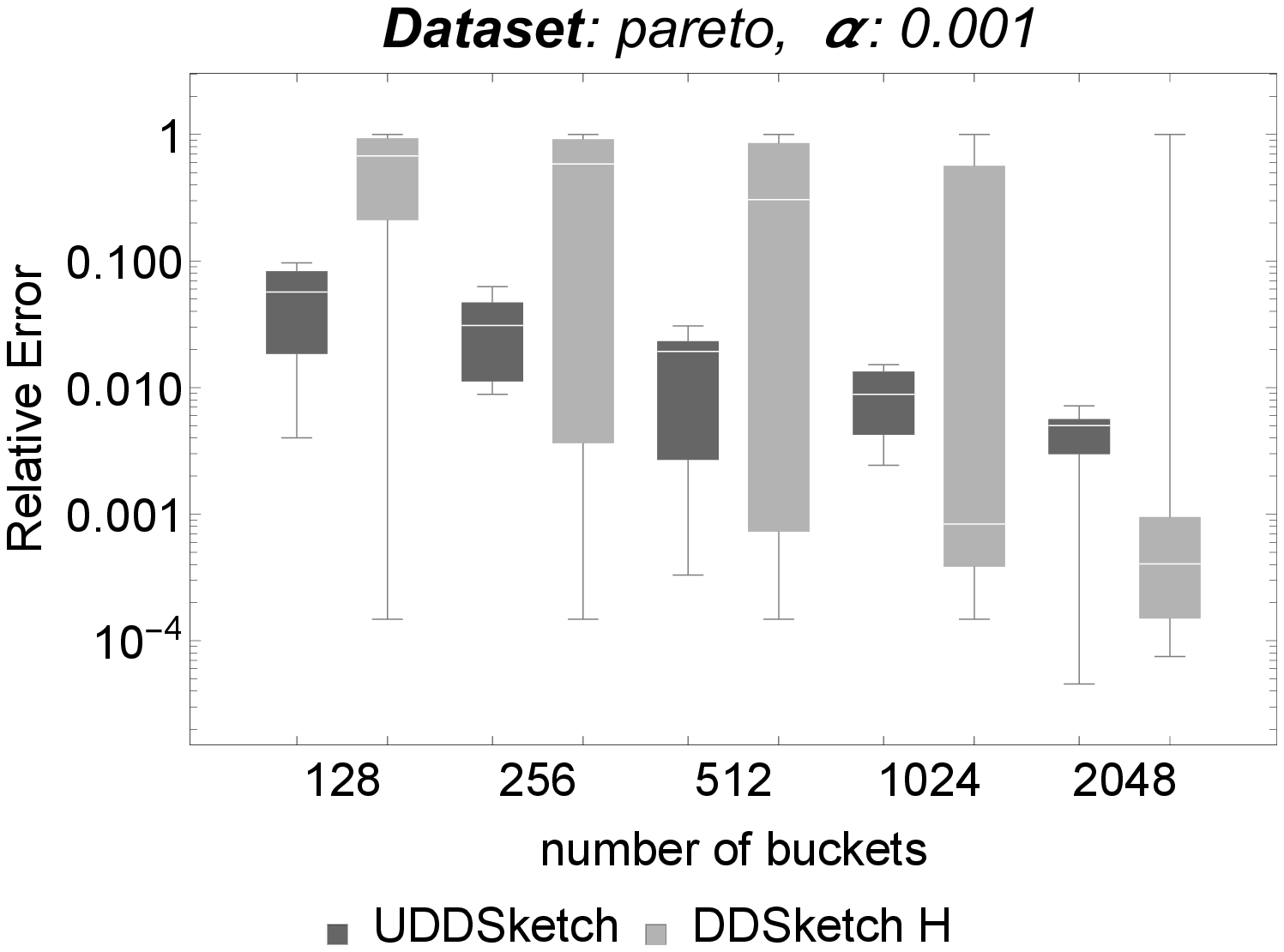}
			\label{pareto-boxplot-ddsH}
		} &
		
		\subfloat[]{
			\includegraphics[width=0.3\textwidth]{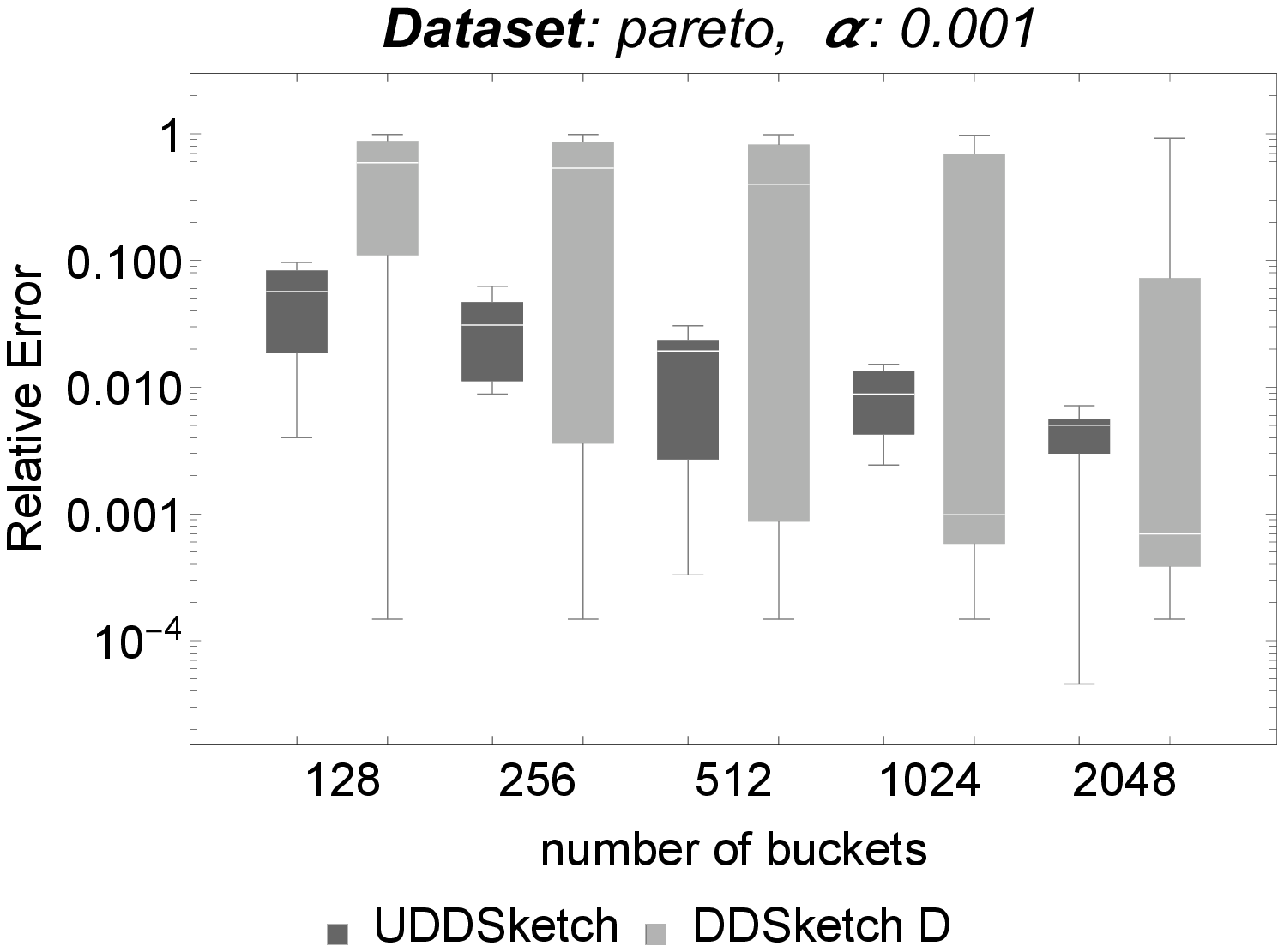}
			\label{pareto-boxplot-ddsD}
		} \\
		
		\subfloat[]{
			\includegraphics[width=0.3\textwidth]{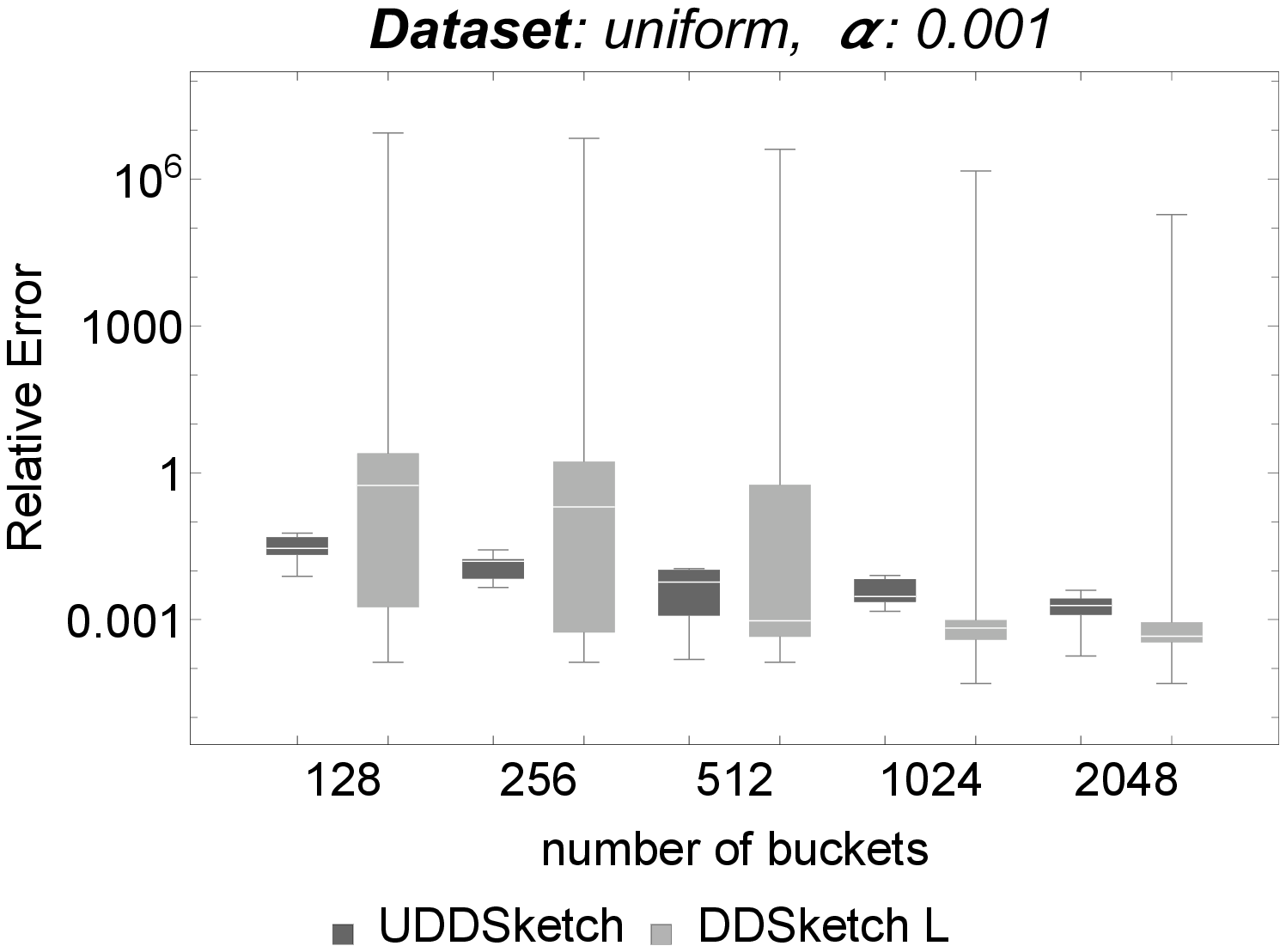}
			\label{uniform-boxplot-ddsL}
		} &
		
		\subfloat[]{
			\includegraphics[width=0.3\textwidth]{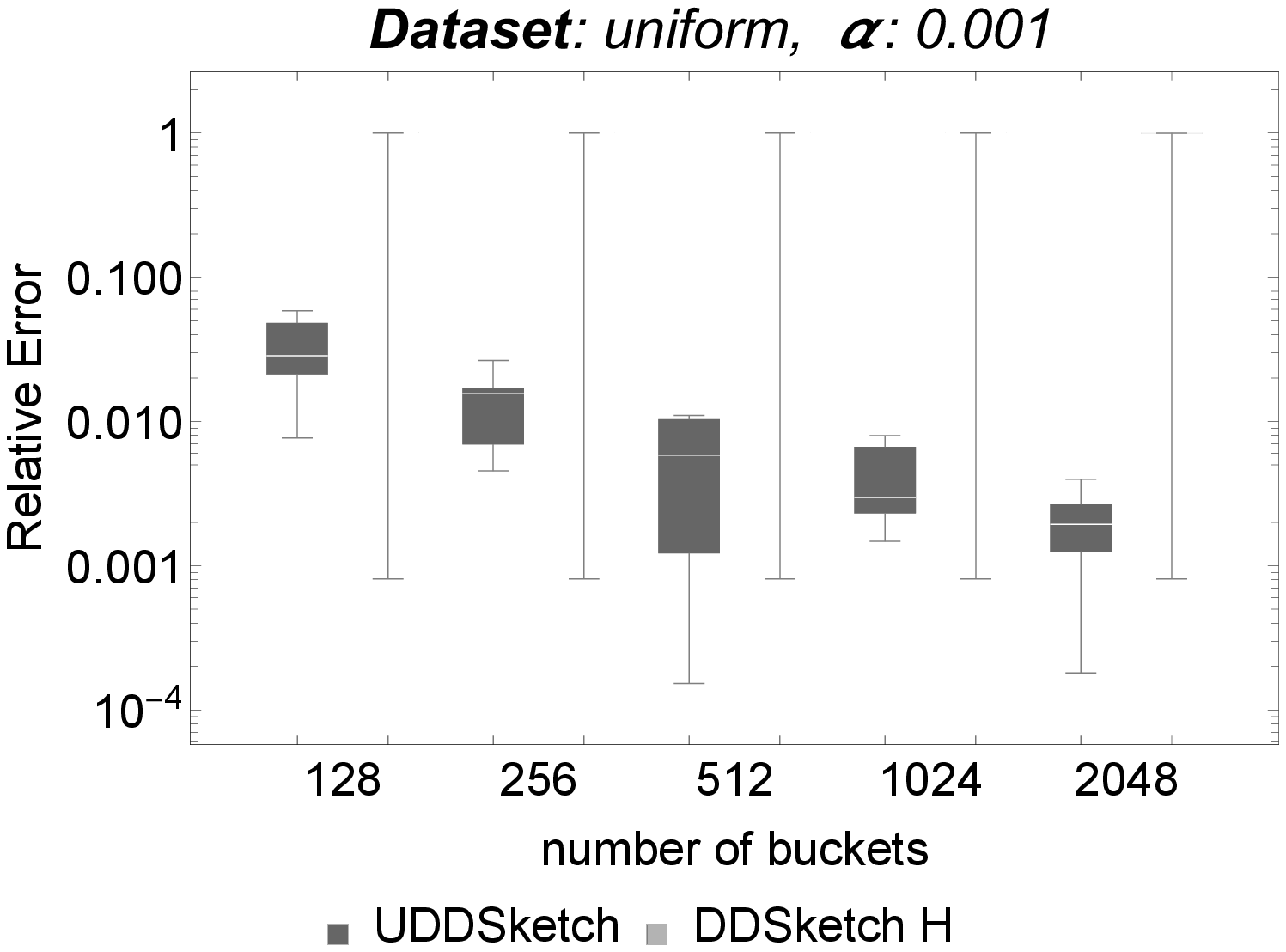}
			\label{uniform-boxplot-ddsH}
		} &
		
		\subfloat[]{
			\includegraphics[width=0.3\textwidth]{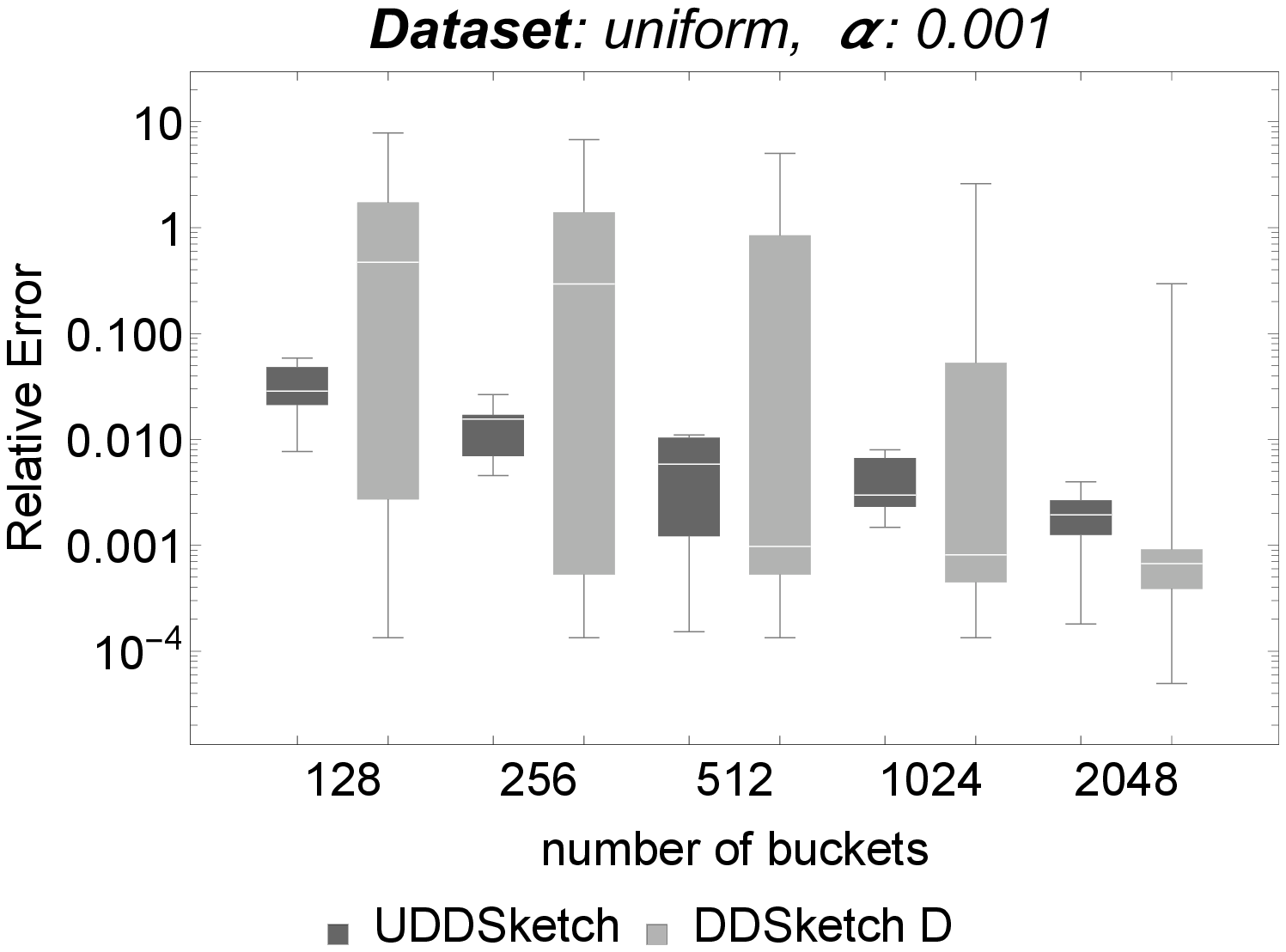}
			\label{uniform-boxplot-ddsD}
		} 
	\end{tabular}
	\caption{Relative errors on quantiles (boxplots), varying the number of buckets.} 
	\label{boxplot-plots2}
\end{figure*}

\begin{figure*}[h]
	\centering
	\begin{tabular}{ccc}		
		\subfloat[]{
			\includegraphics[width=0.3\textwidth]{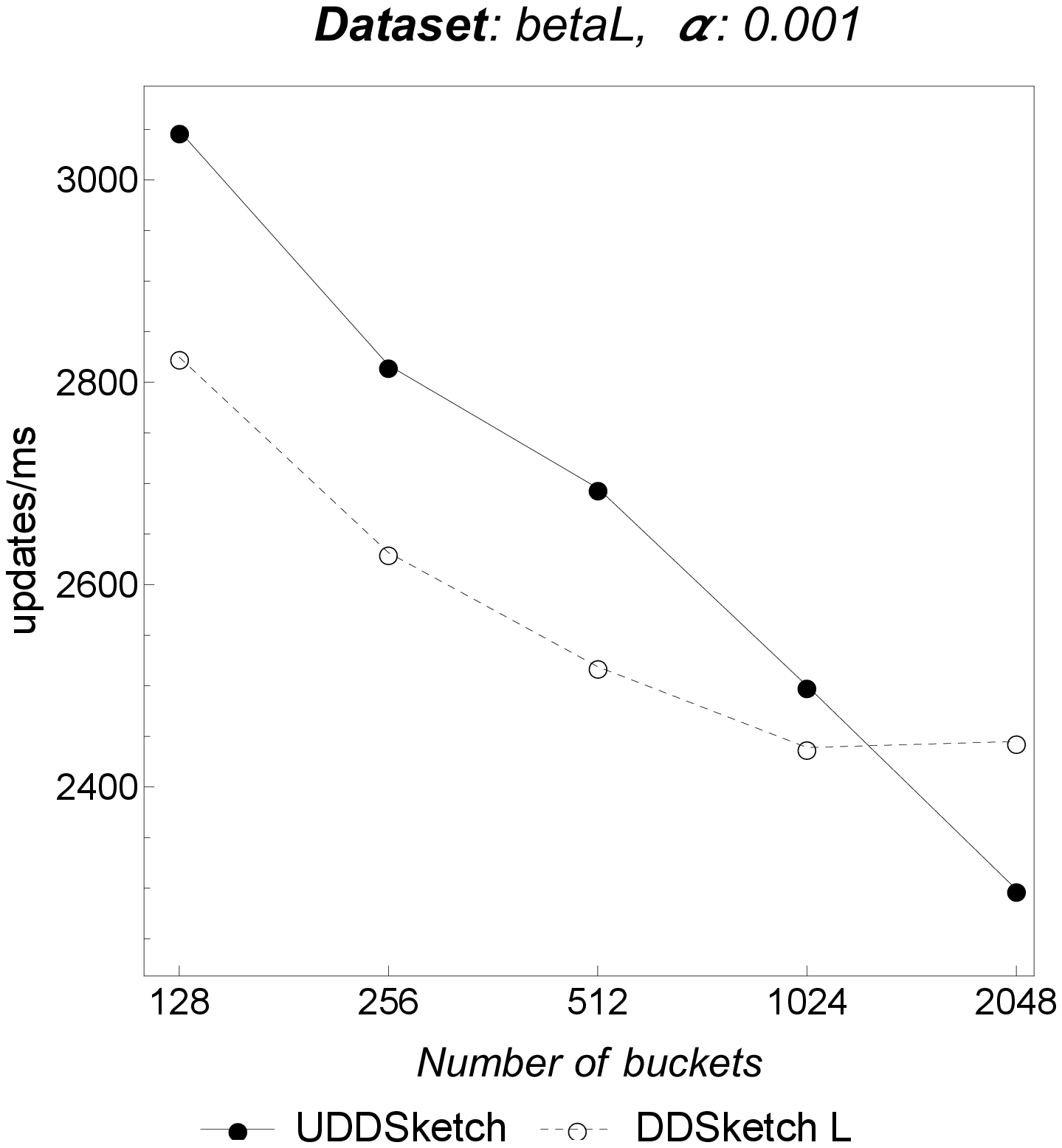}
			\label{betaL-updates_ms-ddsL}
		} &
		
		\subfloat[]{
			\includegraphics[width=0.3\textwidth]{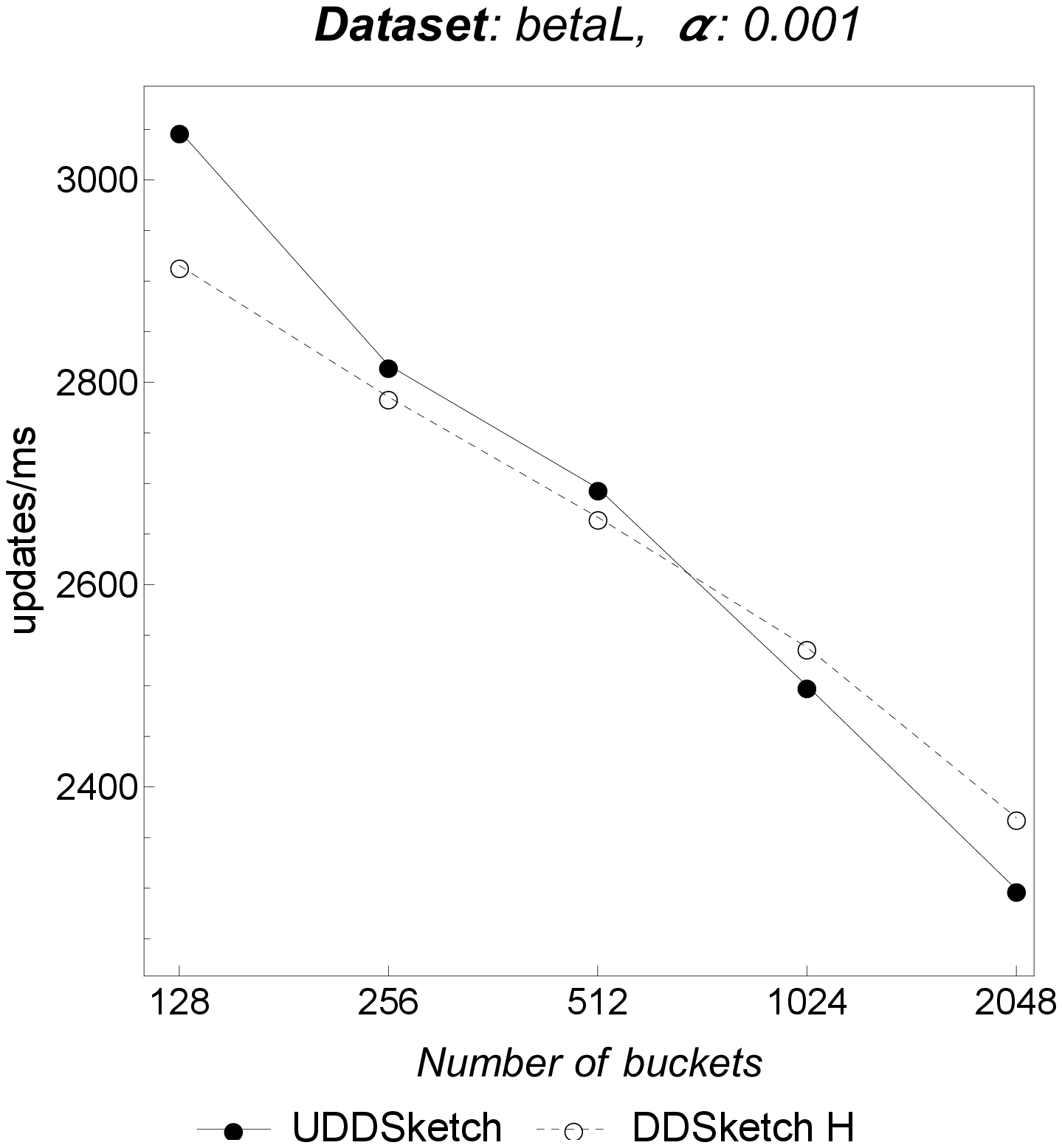}
			\label{betaL-updates_ms-ddsH}
		} &
		
		\subfloat[]{
			\includegraphics[width=0.3\textwidth]{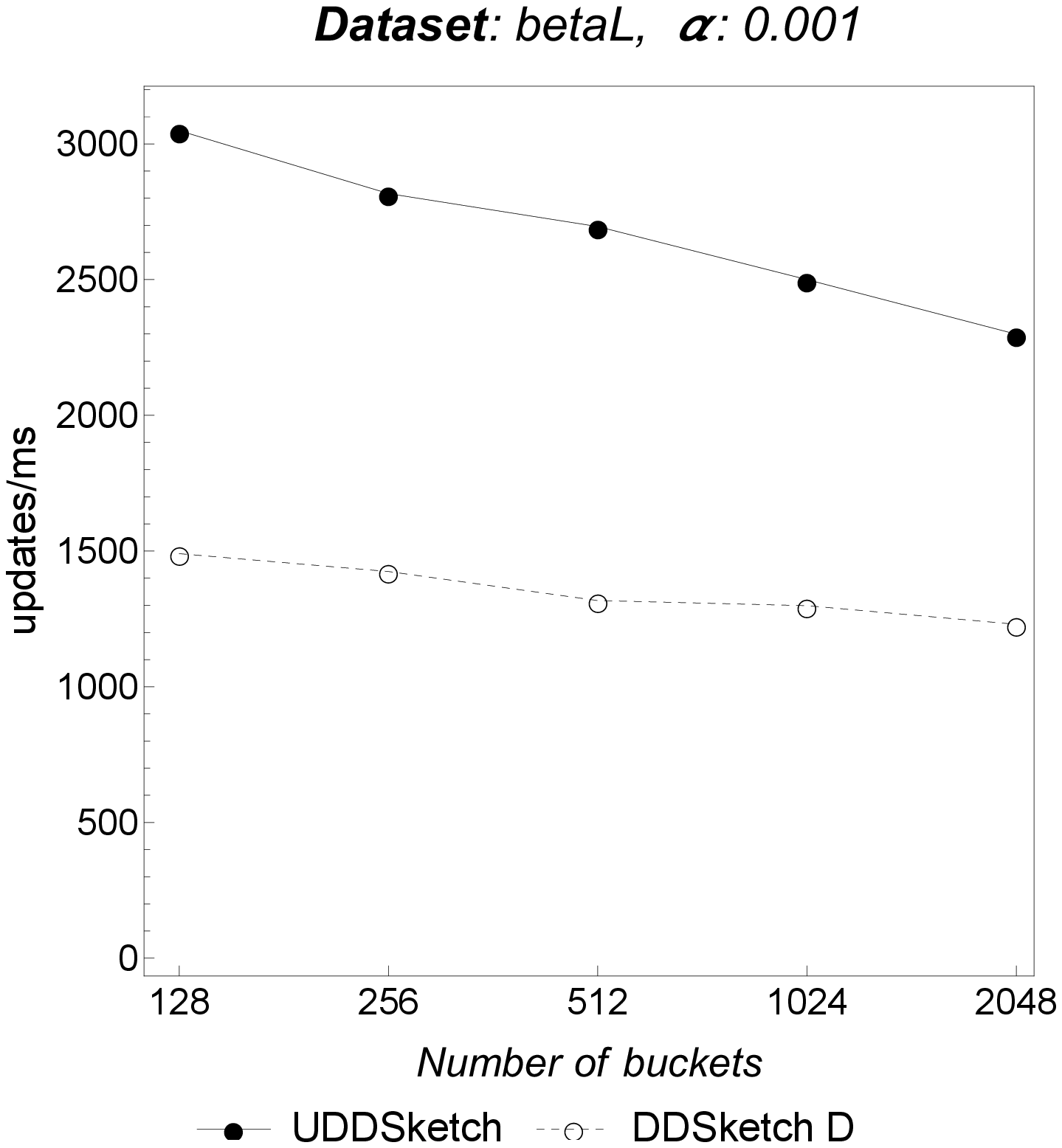}
			\label{betaL-updates_ms-ddsD}
		} \\
		
		\subfloat[]{
			\includegraphics[width=0.3\textwidth]{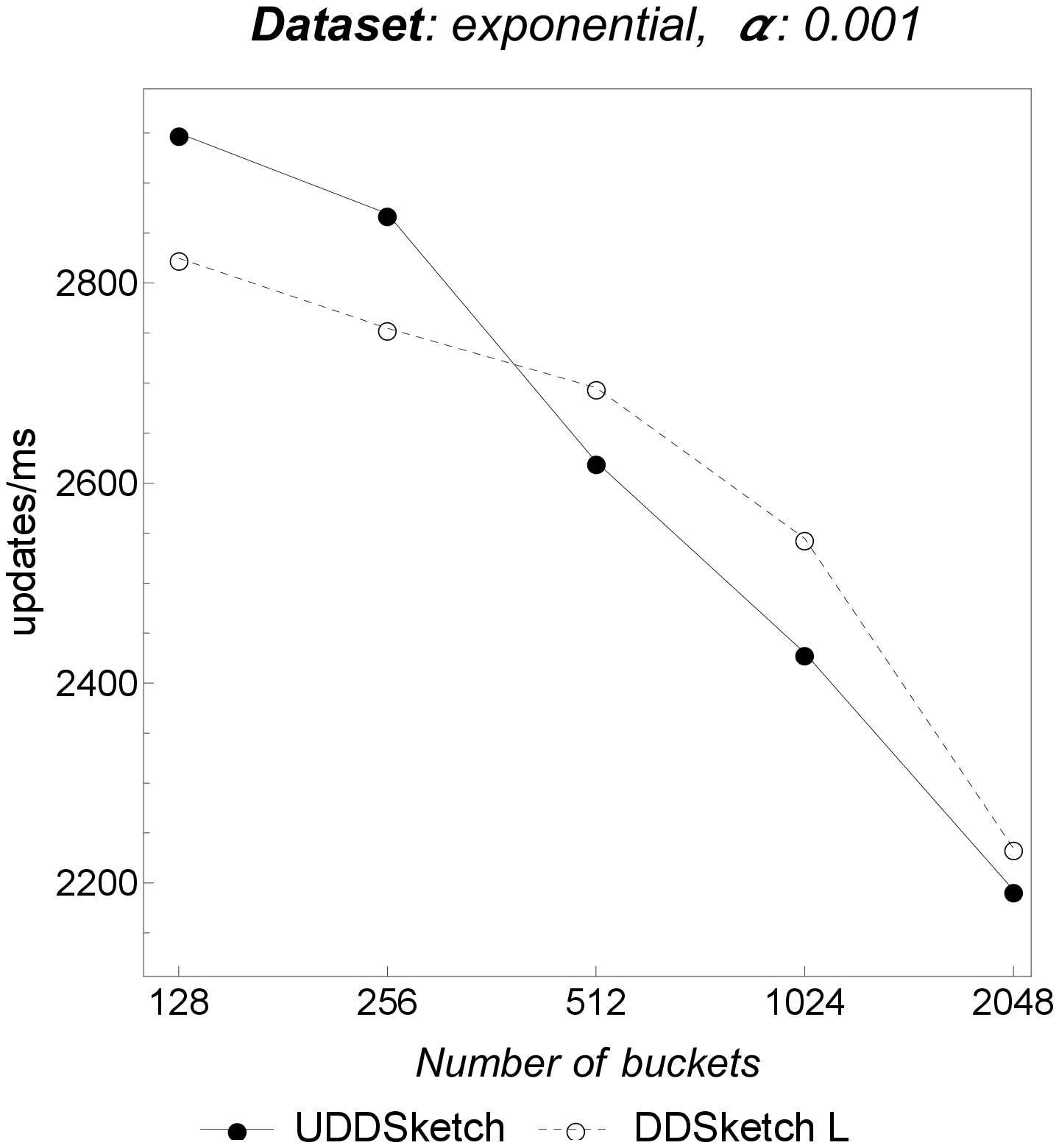}
			\label{exponential-updates_ms-ddsL}
		} &
		
		\subfloat[]{
			\includegraphics[width=0.3\textwidth]{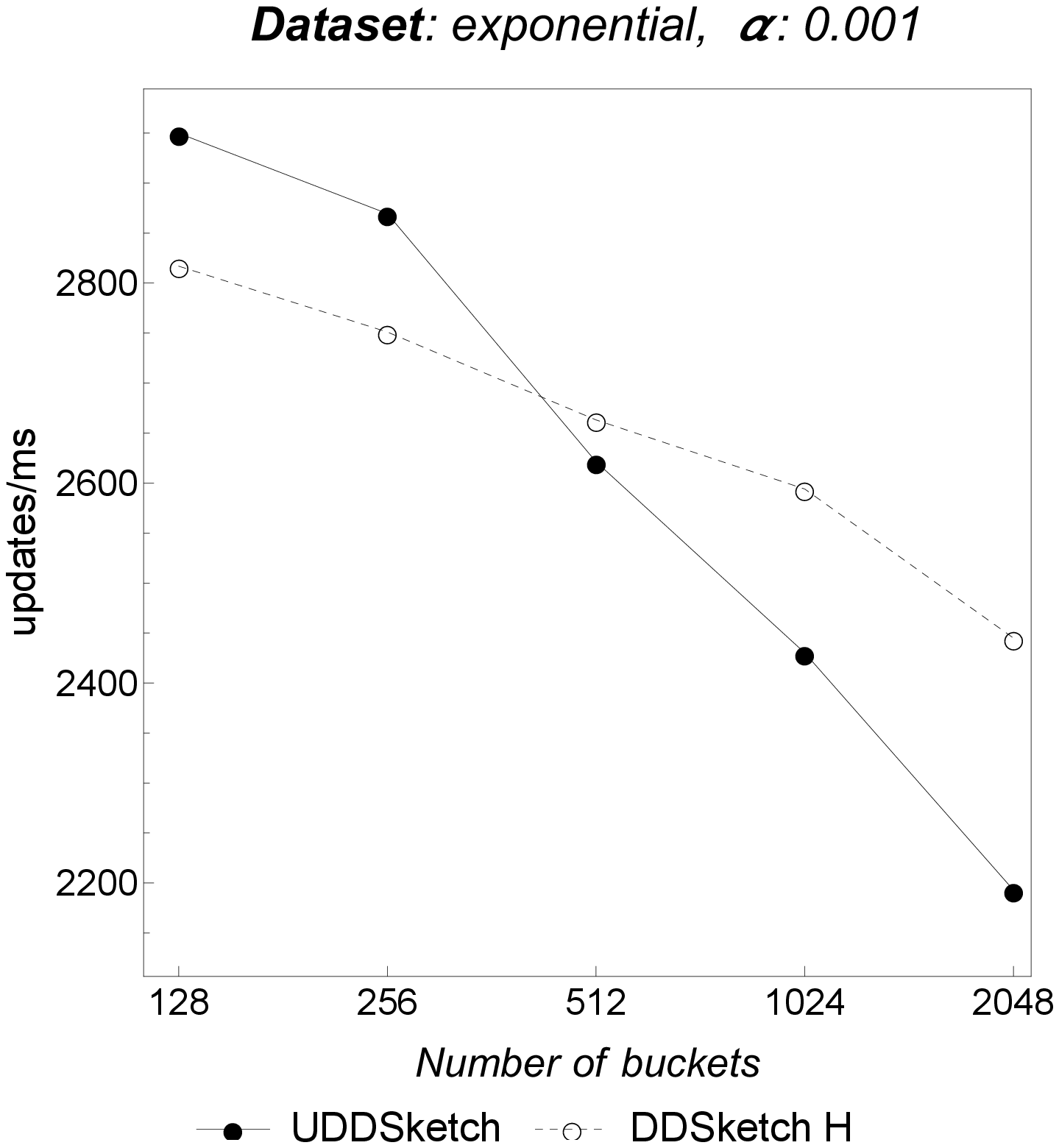}
			\label{exponential-updates_ms-ddsH}
		} &
		
		\subfloat[]{
			\includegraphics[width=0.3\textwidth]{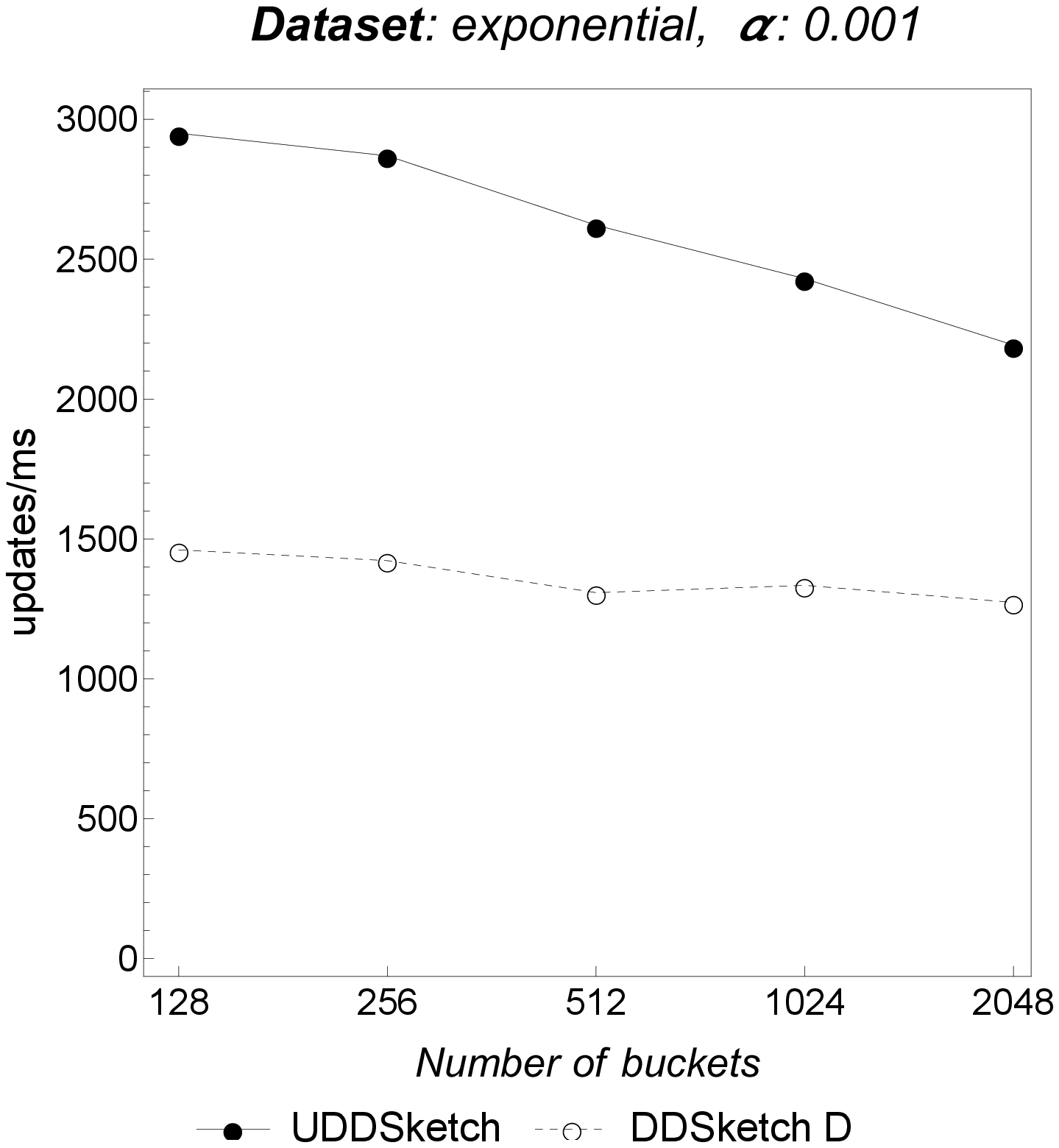}
			\label{exponential-updates_ms-ddsD}
		} \\
		
		\subfloat[]{
			\includegraphics[width=0.3\textwidth]{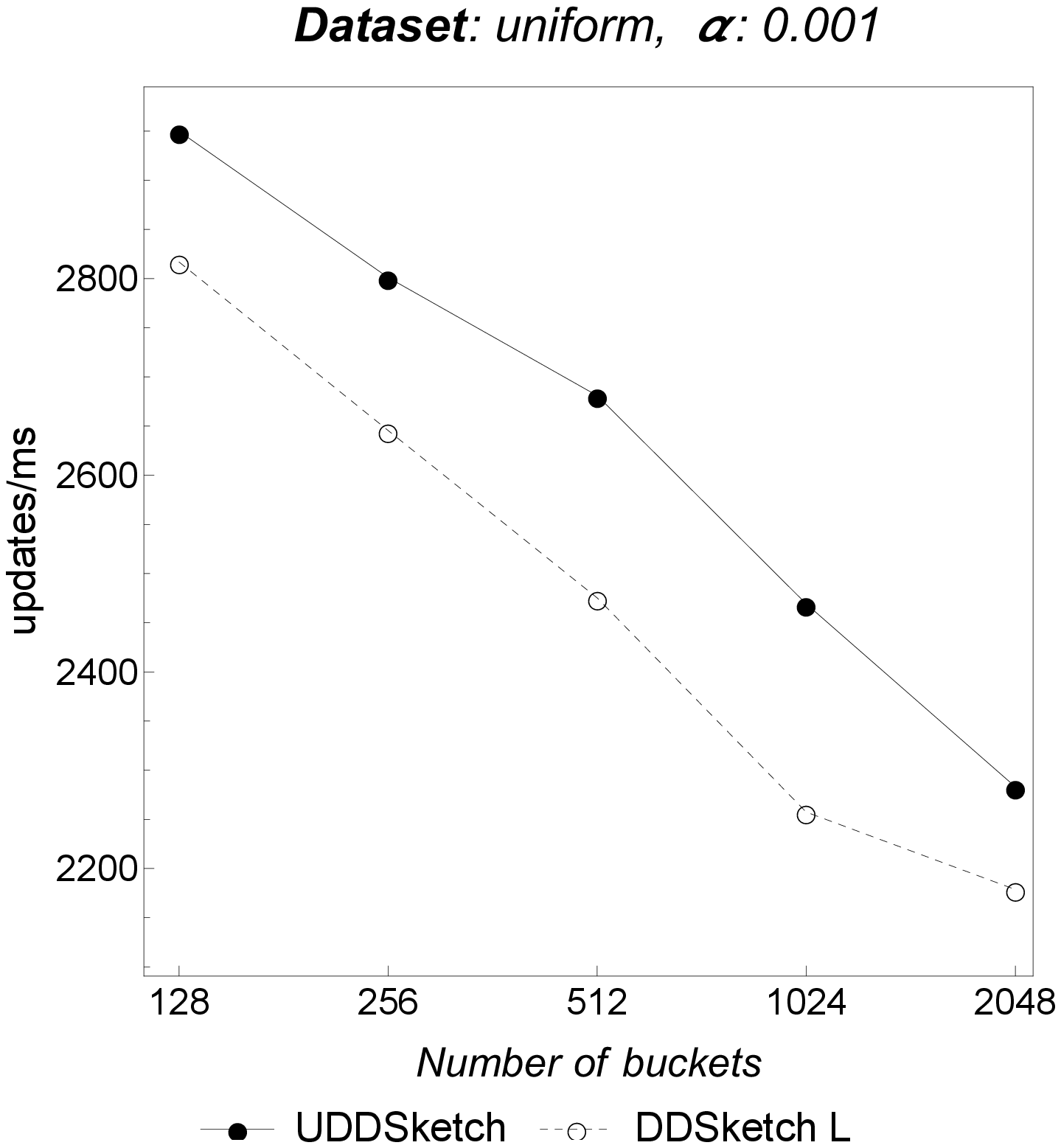}
			\label{uniform-updates_ms-ddsL}
		} &
		
		\subfloat[]{
			\includegraphics[width=0.3\textwidth]{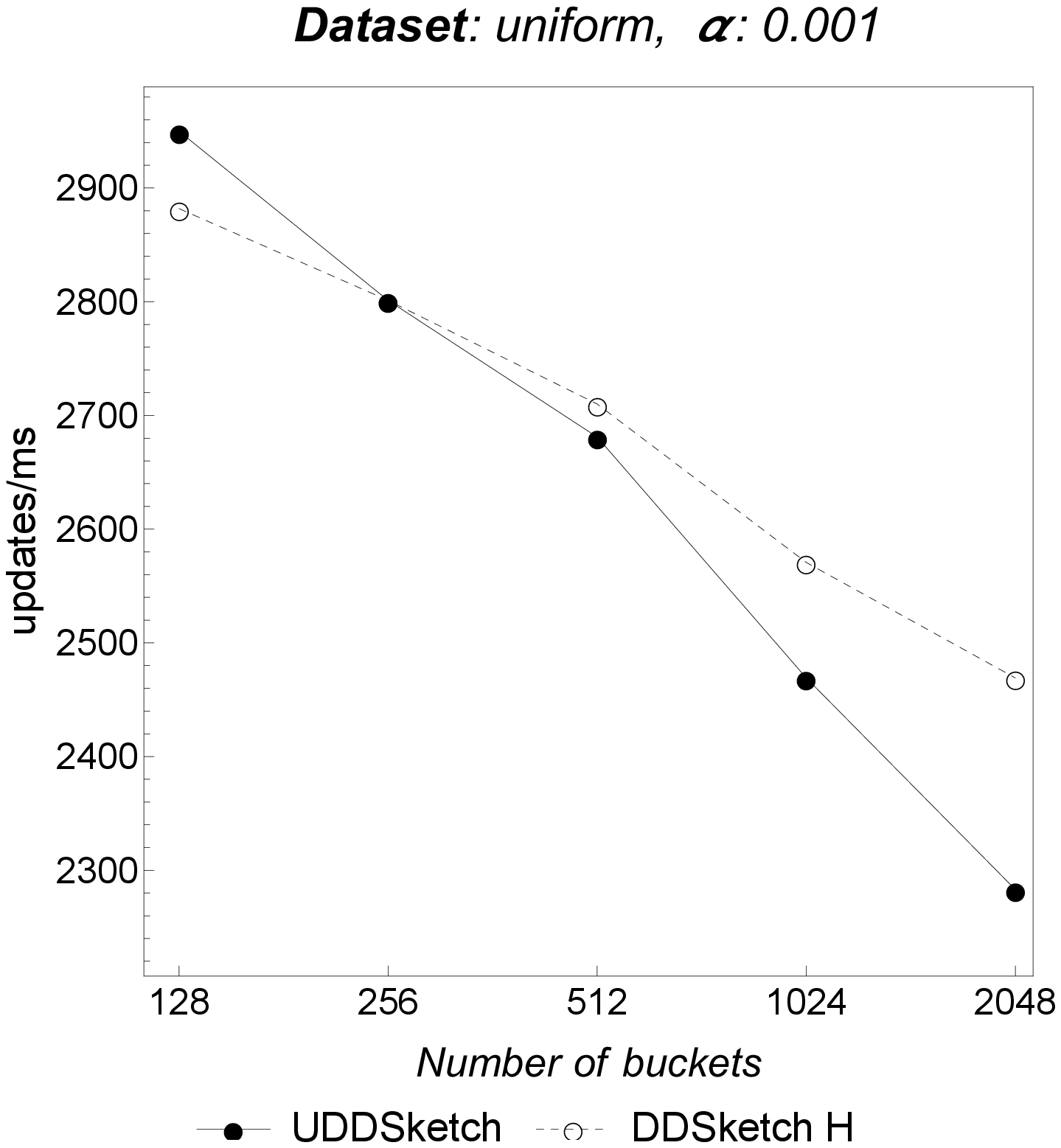}
			\label{uniform-updates_ms-ddsH}
		} &
		
		\subfloat[]{
			\includegraphics[width=0.3\textwidth]{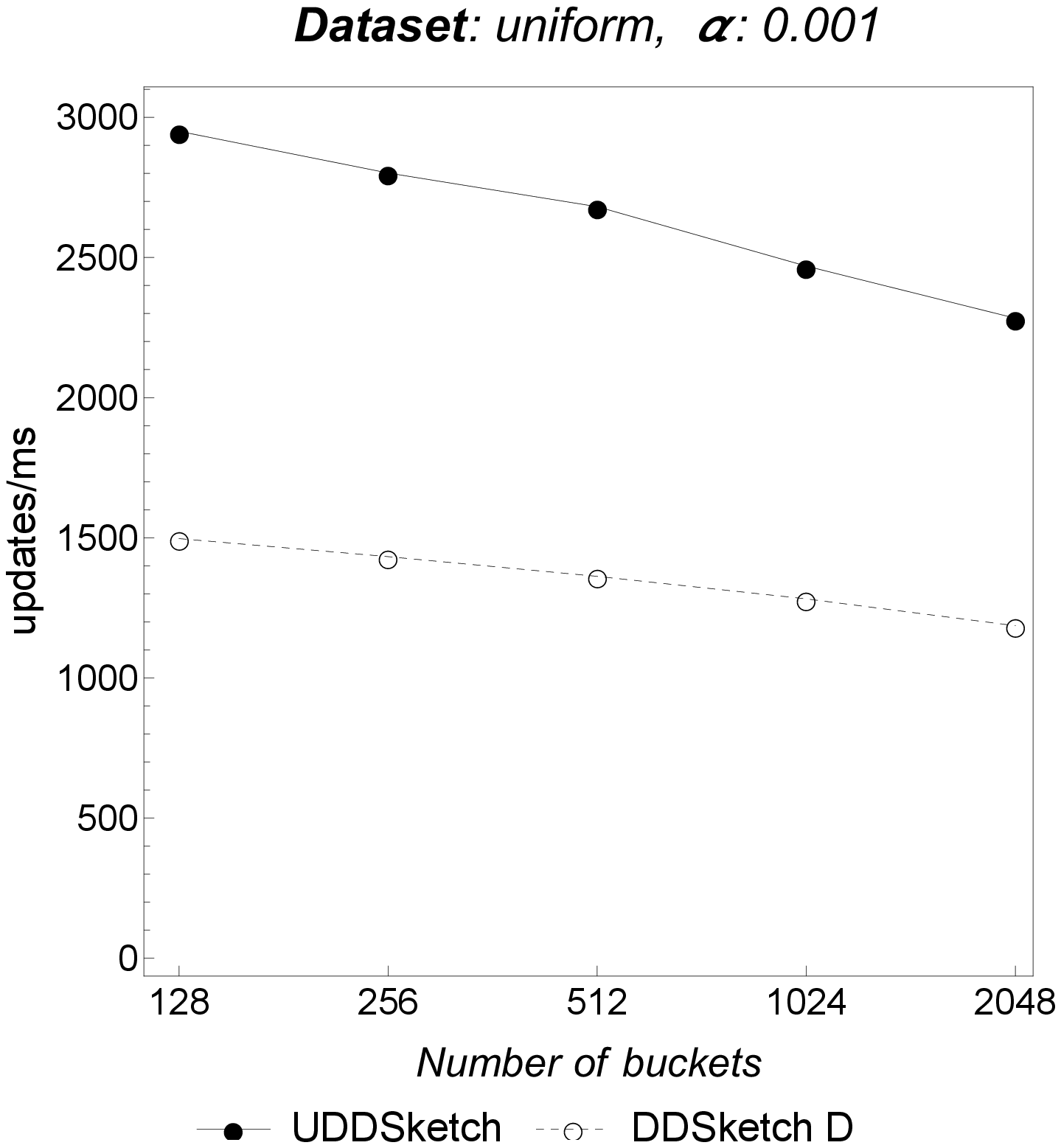}
			\label{uniform-updates_ms-ddsD}
		} 
	\end{tabular}
	
	\caption{Updates per ms, varying the number of buckets.} 
	\label{update_ms-plots}
\end{figure*}


\clearpage

\bibliographystyle{elsarticle-num}
\bibliography{bibliography}

\end{document}